\documentclass[a4paper]{article}

\usepackage[english]{babel}
\usepackage[utf8x]{inputenc}
\usepackage[T1]{fontenc}

\usepackage[a4paper,top=2cm,bottom=2cm,left=2cm,right=2cm,]{geometry}

\usepackage{soul}
\usepackage{amsmath}
\usepackage{amsthm}
\usepackage{amssymb}
\usepackage{graphicx}
\usepackage[colorinlistoftodos]{todonotes}
\usepackage[colorlinks=true, allcolors=blue]{hyperref}
\usepackage[linesnumbered]{algorithm2e}
\usepackage{tikz}
\usetikzlibrary{shapes.geometric}
\usetikzlibrary{arrows.meta,arrows}
\usepackage{mathtools}
\usepackage{enumitem}
\usepackage{bm}
\usepackage[normalem]{ulem}

\newcommand{\todoLeo}[1]{}
\newcommand{\leo}[1]{#1}
\newcommand{\leoo}[1]{#1}
\newcommand{\leooo}[1]{#1}
\newcommand{\leoooo}[1]{#1}
\newcommand{\charles}[1]{#1}

\newcommand{\delete}[1]{}
\newcommand{\deletee}[1]{}

\newcommand{\todoRemie}[1]{}
\newcommand{\remiee}[1]{#1}
\newcommand{\remie}[1]{#1}

\newcommand{\yuki}[1]{#1}

\newcommand{\yukiii}[1]{#1}
\newcommand{\yukim}[1]{#1}
\newcommand{\yukimm}[1]{#1}
\newcommand{\todoYuki}[1]{}

\newcommand{\markj}[1]{#1}

\usepackage{color}
\newcommand{\blue}[1]{#1}
\newcommand{\revL}[1]{#1}
\newcommand{\latest}[1]{#1}
\newcommand{\may}[1]{#1}
\newcommand{\shortened}[1]{#1}

\newcommand{\red}[1]{#1}

\newcommand{\marknew}[1]{#1}

\newtheorem{theorem}{Theorem}
\newtheorem{lemma}{Lemma}
\newtheorem{proposition}{Proposition}
\newtheorem{corollary}{Corollary}

\theoremstyle{definition}
\newtheorem{definition}{Definition}

\title{Orienting undirected phylogenetic networks\footnote{This research is based on discussions during and directly after the workshop ``Distinguishability in Genealogical Phylogenetic Networks'' held at the Lorentz Center, Universiteit Leiden, The Netherlands, 2018.}}
\author{Katharina T. Huber\footnote{School of Computing Sciences, University of East Anglia, NR4 7TJ, Norwich, United Kingdom, \{K.Huber, V.Moulton\}@uea.ac.uk} \and Leo van Iersel\footnote{Corresponding author, L.J.J.vanIersel@tudelft.nl.}\, \footnote{Delft Institute of Applied Mathematics, Delft University of Technology, Van Mourik Broekmanweg 6, 2628 XE, Delft, The Netherlands, \{L.J.J.vanIersel, R.Janssen-2, M.E.L.Jones, Y.Murakami\}@tudelft.nl. Research funded in part by the Netherlands Organization for Scientific Research (NWO), including Vidi grant 639.072.602, and partly by the 4TU Applied Mathematics Institute.} \and Remie Janssen\footnotemark[4] \and Mark Jones\footnotemark[4] \and Vincent Moulton\footnotemark[2] \and Yukihiro Murakami\footnotemark[4] \and Charles Semple\footnote{School of Mathematics and Statistics, University of Canterbury, Private Bag 4800, Christchurch 8140, New Zealand, charles.semple@canterbury.ac.nz. Charles Semple was supported by the New Zealand Marsden Fund.}
      }

\begin{document}

\maketitle
           
\begin{abstract}
This paper studies the relationship between undirected (unrooted) and directed (rooted) phylogenetic networks. We describe a polynomial-time algorithm for deciding whether an undirected \latest{nonbinary} phylogenetic network, given the locations of the root and reticulation vertices, can be oriented as a directed \latest{nonbinary} phylogenetic network. Moreover, we characterize when this is possible and show that, in such instances, the resulting directed \latest{nonbinary} phylogenetic network is unique. In addition, without being given the location of the root and the reticulation vertices, we describe an algorithm for deciding whether an undirected binary phylogenetic network~$N$ can be oriented as a directed binary phylogenetic network of a certain class. The algorithm is fixed-parameter tractable (FPT) when the parameter is the level of~$N$ and is applicable to \shortened{classes of directed phylogenetic networks that satisfy certain conditions. As an example, we show that the well-studied class of binary tree-child networks satisfies these conditions.}
\end{abstract}

\section{Introduction}

Phylogenetic networks are graphs \blue{which are} used to describe, for example, the evolutionary relationships of extant species~\cite{huson2010phylogenetic}. \blue{Such networks} generalize the more widely-known \blue{concept of} phylogenetic trees. The leaves of such a \blue{phylogenetic} network represent extant species, while the interior vertices represent hypothetical ancestors.


Phylogenetic networks are usually rooted \blue{acyclic} directed graphs, where the vertices and arcs combine to represent evolutionary events (e.g\remiee{.,} hybridization or \blue{horizontal} gene transfer). However, unrooted undirected graphs have also been studied which still aim to describe an explicit evolutionary history, but do not include directions on the edges~\cite{morrison2005networks}. Reasons for not including \blue{directions} can be uncertainty about the location of the root \blue{and} uncertainty about the order in which reticulate events occurred, that is, events where species or lineages merge. Moreover, it can be unclear which vertices represent reticulate events and which \blue{vertices} represent speciation events or ``vertical'' descent. See Figure~\ref{fig:intro} for an example of an undirected and a directed phylogenetic network which illustrates these differences in perspective. Note that unrooted networks are also used as a tool to display patterns within data (e.g.\remiee{,} split networks~\cite{neighbornet}) but, as these networks do not aim to explicitly represent the evolution of the underlying species, we do not focus on them here.

\blue{In addition} to directed \blue{and} undirected phylogenetic networks, a third option is partly-directed phylogenetic networks, \blue{that is, phylogenetic networks in which only some of the edges are oriented.} Such networks make sense in light of the discussion above and, indeed, several published \blue{phylogenetic} networks in the biological literature are partly-directed, e.g., \blue{of} grape \blue{cultivars}~\cite{grapes} and \blue{of the evolutionary history of Europeans}~\cite{humans}, or contain bi-directed arcs, e.g., \blue{of} bears~\cite{bears}. Also, the popular software tool SNAQ produces partly-directed phylogenetic networks~\cite{SNAQ}. \blue{Despite these publications,} partly-directed phylogenetic networks have yet to be studied from a mathematical perspective, even though this was suggested by David Morrison in 2013~\cite{blog}: ``Perhaps the possibility of partly directed phylogenetic networks needs more consideration.''

In this paper, we \blue{study two} fundamental questions regarding the relationship between undirected and directed phylogenetic networks. \blue{In the first part of the paper, we investigate the following.} Suppose we \blue{are given} the underlying undirected phylogenetic network of some directed (\leoo{nonbinary}) phylogenetic network \blue{$N$} as well as the location of the root \blue{of $N$} and the \yukim{desired in-degrees of the} reticulation vertices (the vertices where lineages merge) \blue{of $N$}. Does this give us enough information to uniquely reconstruct \blue{$N$}? We show that this is indeed the case. Moreover, given the locations of the root and the desired in-degrees of the reticulation vertices, we \blue{characterize} when an undirected phylogenetic network \blue{$N'$} can be oriented as a directed phylogenetic network (see Theorem~\ref{thm:nonbinary_characterization}). For an example of an \blue{undirected} binary phylogenetic network where this is \emph{not} possible, see Figure~\ref{fig:notorientable}. Following this, we give a linear-time algorithm in the number of edges of \blue{$N'$} to find such an orientation.
We also show how to apply the algorithm to partly-directed networks. In particular, we show how one can decide \blue{in quadratic time in the number of edges} whether a given partly-directed network is a semi-directed network, i.e., whether it can be obtained from some directed phylogenetic network by suppressing the root and removing all directions from non-reticulation edges \blue{(see Corollary~\ref{cor:semidirected})}.

\begin{figure}[t]
\centering
 \begin{tikzpicture}
	 \tikzset{edge/.style={thick}}
     \tikzset{arc/.style={-Latex,thick}}
     \tikzset{revarc/.style={Latex-,thick}}
	 \begin{scope}[xshift=0cm,yshift=0cm]
    \draw[thick, fill, radius=0.06] (0,0) circle node[below] {\small T. turgidum};
    \draw[thick, fill, radius=0.06] (2,0) circle node[below] {\small T. aestivum};
    \draw[thick, fill, radius=0.06] (0,.5) circle;
    \draw[thick, fill, radius=0.06] (2,.5) circle;
    \draw[thick, fill, radius=0.06] (3,1.5) circle;
    \draw[edge] (2,.5) -- (3,1.5);
    \draw[thick, fill, radius=0.06] (3.5,1) circle node[below right] {\small A. tauschii};
    \draw[thick, fill, radius=0.06] (3,2.5) circle;
    \draw[edge] (3,2.5) -- (3,1.5);
    \draw[thick, fill, radius=0.06] (3.5,2.5) circle node[right] {\small A. sharonensis};
    \draw[thick, fill, radius=0.06] (3,3.5) circle;
    \draw[edge] (3,2.5) -- (3,3.5);
    \draw[thick, fill, radius=0.06] (2,2.5) circle;
    \draw[edge] (2,2.5) -- (3,3.5);
    \draw[thick, fill, radius=0.06] (1,3.5) circle;
    \draw[edge] (2,2.5) -- (1,3.5);
    \draw[thick, fill, radius=0.06] (-1,3.5) circle;
    \draw[edge] (-1,3.5) -- (1,3.5);
    \draw[thick, fill, radius=0.06] (-1.5,4) circle node[above] {\small T. monococcum \quad};
    \draw[edge] (-1,3.5) -- (-1.5,4);
    \draw[thick, fill, radius=0.06] (-1,2.5) circle;
    \draw[edge] (-1,3.5) -- (-1,2.5);
    \draw[thick, fill, radius=0.06] (-1.5,2.5) circle node[left] {\small T. uartu};
    \draw[edge] (-1.5,2.5) -- (-1,2.5);
    \draw[thick, fill, radius=0.06] (-1,1.5) circle;
    \draw[thick, fill, radius=0.06] (2,1.5) circle;
    \draw[thick, fill, radius=0.06] (1.5,2) circle node[above left] {\small A. speltoides};
    \draw[edge] (2,1.5) -- (-1,1.5);
    \draw[edge] (1.5,2) -- (2,1.5);
    \draw[edge] (2,1.5) -- (2,2.5);
    \draw[edge] (2,1.5) -- (1.5,1.5);
    \draw[edge] (-1,1.5) -- (-1,2.5);
    \draw[edge] (-1,1.5) -- (0,.5);
    \draw[edge] (3,3.5) -- (1,3.5);
    \draw[edge] (3,2.5) -- (3.5,2.5);
    \draw[edge] (3.5,1) -- (3,1.5);
    \draw[edge] (0,0) -- (0,.5);
    \draw[edge] (2,0) -- (2,.5);
    \draw[edge] (0,.5) -- (2,.5);
    \end{scope}
    \begin{scope}[xshift=8cm,yshift=0cm]
    \draw[thick, fill, radius=0.06] (0,0) circle node[below] {\small T. turgidum};
    \draw[thick, fill, radius=0.06] (2,-.5) circle node[below] {\small T. aestivum};
    \draw[thick, fill, radius=0.06] (0,.5) circle;
    \draw[thick, fill, radius=0.06] (2,0) circle;
    \draw[thick, fill, radius=0.06] (3,1.5) circle;
    \draw[revarc] (2,0) -- (3,1.5);
    \draw[thick, fill, radius=0.06] (3.5,1) circle node[below right] {\small A. tauschii};
    \draw[thick, fill, radius=0.06] (3,2.5) circle;
    \draw[arc] (3,2.5) -- (3,1.5);
    \draw[thick, fill, radius=0.06] (3.5,2) circle node[below right] {\small A. sharonensis};
    \draw[thick, fill, radius=0.06] (3,3.5) circle;
    \draw[arc] (3,3.5) -- (3,2.5);
    \draw[thick, fill, radius=0.06] (2,4.5) circle;
    \draw[arc] (2,4.5) -- (3,3.5);
    \draw[thick, fill, radius=0.06] (0,4.5) circle;
    \draw[thick, fill, radius=0.06] (1,5) circle;
    \draw[arc] (1,5) -- (0,4.5);
    \draw[arc] (1,5) -- (2,4.5);
    \draw[thick, fill, radius=0.06] (-1,3.5) circle;
    \draw[revarc] (-1,3.5) -- (0,4.5);
    \draw[thick, fill, radius=0.06] (-1.5,3) circle node[left] {\small T. monococcum};
    \draw[arc] (-1,3.5) -- (-1.5,3);
    \draw[thick, fill, radius=0.06] (-1,2.5) circle;
    \draw[arc] (-1,3.5) -- (-1,2.5);
    \draw[thick, fill, radius=0.06] (-1.5,2) circle node[below left] {\small T. uartu};
    \draw[revarc] (-1.5,2) -- (-1,2.5);
    \draw[thick, fill, radius=0.06] (-1,1.5) circle;
    \draw[thick, fill, radius=0.06] (1,3.5) circle;
    \draw[thick, fill, radius=0.06] (1.5,3) circle node[below] {\small A. speltoides};
    \draw[arc] (1,3.5) -- (1.5,3);
    \draw[arc] (1,3.5) -- (-1,1.5);
    \draw[revarc] (1,3.5) -- (2,4.5);
    \draw[revarc] (-1,1.5) -- (-1,2.5);
    \draw[arc] (-1,1.5) -- (0,.5);
    \draw[revarc] (3,3.5) -- (0,4.5);
    \draw[arc] (3,2.5) -- (3.5,2);
    \draw[revarc] (3.5,1) -- (3,1.5);
    \draw[revarc] (0,0) -- (0,.5);
    \draw[arc] (2,0) -- (2,-.5);
    \draw[arc] (0,.5) -- (2,0);
    \end{scope}
	\end{tikzpicture}
\caption{\label{fig:intro} An undirected phylogenetic network (left) and a directed phylogenetic network (right) \leoo{based on~\cite{blogwheat,wheat}}.
Note that the directed phylogenetic network can be obtained from the undirected phylogenetic network by adding a root vertex and orienting the edges.}
\end{figure}

\begin{figure}[h]
\centering
 \begin{tikzpicture}
	 \tikzset{edge/.style={thick}}
     \tikzset{arc/.style={-Latex,thick}}
	 \begin{scope}[xshift=0cm,yshift=0cm]
    \draw[arc] (1,0.4) -- (1,0);
	\draw[thick, fill, radius=0.06] (0,0) circle;
	\draw[thick, fill, radius=0.06] (-1,0) circle node[left] {$x$};
	\draw[edge] (-1,0) -- (0,0);
	\draw[fill] (1.9,-.1) rectangle (2.1,.1);
    \draw[edge] (0,0) -- (2,0);
    \draw[fill] (.9,-1.1) rectangle (1.1,-.9);
    \draw[edge] (0,0) -- (1,-1);
    \draw[edge] (2,0) -- (1,-1);
    \draw[thick, fill, radius=0.06] (3,-1) circle;
    \draw[edge] (1,-1) -- (3,-1);
    \draw[edge] (2,0) -- (3,-1);
    \draw[thick, fill, radius=0.06] (4,-1) circle node[right] {$y$};
    \draw[edge] (3,-1) -- (4,-1);
    \end{scope}
    \begin{scope}[xshift=6cm,yshift=0cm]
	\draw[thick, fill, radius=0.06] (0,0) circle;
	\draw[thick, fill, radius=0.06] (-1,0) circle node[left] {$x$};
	\draw[edge] (-1,0) -- (0,0);
	\draw[thick, fill, radius=0.06] (2,0) circle;
    \draw[edge] (0,0) -- (2,0);
    \draw[thick, fill, radius=0.06] (1,-1) circle;
    \draw[edge] (0,0) -- (1,-1);
    \draw[edge] (2,0) -- (1,-1);
    \draw[thick, fill, radius=0.06] (3,-1) circle;
    \draw[edge] (1,-1) -- (3,-1);
    \draw[edge] (2,0) -- (3,-1);
    \draw[thick, fill, radius=0.06] (4,-1) circle node[right] {$y$};
    \draw[edge] (3,-1) -- (4,-1);
    \end{scope}
	\end{tikzpicture}
\caption{\label{fig:notorientable} Left, an undirected \blue{binary} phylogenetic network with specified reticulation vertices (indicated by squares) and root location (indicated with an arrow) that has no orientation as a directed phylogenetic network. Right, the same undirected \blue{binary} phylogenetic network but with no information about the root and reticulation vertices. This \blue{latter undirected binary phylogenetic} network can be oriented as a \blue{binary} stack-free network, but not as a \blue{binary} tree-child network.}
\end{figure}

In the second part of the paper, we study \blue{the following question}. \blue{Given an} undirected binary phylogenetic network \blue{$N$}, can \blue{$N$} be oriented to become a directed \blue{binary} phylogenetic network of a \blue{given} class (with no information about the location of the root or the reticulation vertices). \blue{Again see} Figure~\ref{fig:notorientable} for an example. We give an algorithm \blue{for this task} that is fixed-parameter tractable (FPT), where the \emph{level} of \blue{$N$} is the parameter \blue{(see Algorithm~\ref{alg:C_orientation})}. \blue{The level of $N$ is a measure of its tree-likeness. (A formal definition is given in the next section.)} The algorithm can be applied to a wide range of classes of \blue{directed binary phylogenetic} networks, including the \blue{well-studied} classes of tree-child, tree-based, reticulation-visible, and stack-free networks, \blue{as well as} the recently-introduced classes of valid networks~\cite{murakami2018reconstructing} and orchard networks~\cite{erdos2019orchards,janssen2018cpn} (see Section~\ref{sec:classesintro} for definitions). \shortened{We include the proof for the class tree-child as an example (see Section~\ref{sec:Classes}) since this is one of the most well-studied classes of phylogenetic networks. The proofs for the other classes, following a similar approach, can be found in Appendix~\ref{app}.} \latest{To obtain this algorithm, we first describe an FPT algorithm where the number of reticulation vertices is the parameter (see Algorithm~\ref{alg:binary_blob_orientation_TC_SF}). The final FPT algorithm (Algorithm~\ref{alg:C_orientation}, which relies on Algorithm~\ref{alg:binary_blob_orientation_TC_SF}) may scale better because it has the level as the parameter, which is always smaller or equal to the reticulation number.} \blue{All of the algorithms in the paper have been} implemented and \blue{are} publicly available~\cite{implementation}.


To the best of our knowledge, the questions \blue{investigated} in this paper have not been studied \blue{previously}. \blue{To date}, most publications consider either directed or undirected phylogenetic networks, but do not study how they are related. Exceptions are a paper studying how to optimally root unrooted trees as to minimize their hybridization number~\cite{van2018unrooted} and papers about orienting split networks~\cite{nanuq,huson2005reconstruction,strimmer2000likelihood}. Also see~\cite{uprooted} which looks into the relationship between undirected phylogenetic networks and Buneman graphs. There is also a large body of literature on orienting graphs (see, e.g.,~\cite{Asahiro2016,alon1992colorings}), but such papers are not applicable to our situation because, for example, they do not require the orientation to be acyclic (one exception being~\cite{atallah1983graph} which is discussed later) or they do not have our degree restrictions. \remiee{Lastly, there are two papers that provide results on the orientability of genealogical phylogenetic networks. However, these only provide such results as sidenotes to their main purpose: rearranging networks~\cite{janssen2018exploring}, and characterizing \blue{undirected (unrooted)} tree-based networks~\cite{Francis2018}.}


\section{Preliminaries}\label{sec:Preliminaries}

\yuki{Throughout the paper,~$X$ denotes a non-empty finite set. Biologically speaking, $X$ can be \blue{viewed} as a set of extant taxa.} \leoo{An \emph{undirected  phylogenetic network} \blue{$N$} on $X$ is an undirected connected (simple) graph, in which no vertex has degree~$2$, and \blue{the set of} vertices of degree~$1$ (the \emph{leaves}) \blue{is $X$}. \blue{We say $N$ is} \emph{binary} if each non-leaf vertex has degree~$3$.} \yuki{An undirected phylogenetic network \leoo{with no cycles} \blue{is} an \emph{undirected phylogenetic tree}}. \yukiii{The \emph{reticulation number} of an undirected phylogenetic network is the number of edges that need to be removed to obtain, after suppressing degree-2 vertices, an undirected phylogenetic tree.}


A \emph{directed phylogenetic network} \blue{$N'$} on~$X$ is a directed acyclic graph \blue{with no parallel arcs} \remiee{in which exactly one vertex has in-degree~$0$ and this vertex has out-degree~$2$ (the {\em root}), no vertices have in-degree~$1$ and out-degree~$1$, and \blue{the set of} vertices of out-degree~$0$ is $X$ \blue{and all such vertices} have in-degree~$1$. \blue{The vertices of out-degree~$0$ are the {\em leaves} of $N'$.} \blue{We say $N'$ is} \emph{binary} if all non-root non-leaf vertices either have in-degree~$1$ and out-degree~$2$, or \latest{have} in-degree~$2$ and out-degree~$1$. Vertices with in-degree at least~$2$ are \emph{reticulations}, \blue{while} vertices with in-degree~$1$ are \emph{tree vertices}.} \blue{Arcs directed into a reticulation are called {\em reticulation arcs}.} \blue{Furthermore,} an arc \blue{of $N'$} is \emph{pendant} if it is incident to a leaf. If $(u,v)$ is an arc \blue{of $N'$}, then~$u$ is a \emph{parent} of~$v$, and~$v$ is a \emph{child} of~$u$. \yuki{A directed (binary) phylogenetic network with no reticulations is a \emph{directed (binary) phylogenetic tree.}}

\blue{To avoid ambiguity, when the need arises we will say a ``nonbinary phylogenetic network'' to mean a phylogenetic network that is not necessarily binary.}
Furthermore, we note that in the phylogenetics literature the terms \emph{rooted} and \emph{unrooted} \blue{phylogenetic} network are often used. However, since the location of the root does not \remiee{necessarily} imply the direction of all the arcs, we will use \emph{directed} and \emph{undirected} instead of \emph{rooted} and \emph{unrooted}, respectively.



\yukim{Two undirected \blue{phylogenetic} networks~$N$ and~$M$ on~$X$ are \emph{isomorphic} if there exists a bijection $f$ from the vertex set of $N$ to the vertex set of $M$ such that $f(x)=x$ for all $x\in X$, and such that~$\{u,v\}$ is an edge of~$N$ if and only if~$\{f(u), f(v)\}$ is an edge of~$M$.} Given an undirected \blue{phylogenetic} network $N$ on $X$ and a directed \blue{phylogenetic} network $N'$ on $X$, we say that $N$ is the \emph{underlying network} of $N'$ and \blue{that} $N'$ is an \leo{\emph{orientation}} of $N$ if the \blue{undirected phylogenetic} network \blue{obtained} from~$N'$ by replacing all directed arcs with undirected edges and suppressing \remiee{its degree-$2$ root} is isomorphic to~$N$. We say that~$N$ is \emph{orientable} if it has at least one orientation.

\remie{A \emph{biconnected component} of a directed or undirected \blue{phylogenetic} network is a maximal subgraph that cannot be disconnected by deleting a single vertex. A \remiee{biconnected component} is called a \emph{blob} if it contains at least three vertices.}
An \leoo{undirected \blue{phylogenetic} network} is \emph{level-$k$} if, by deleting at most~$k$ edges from each blob, \blue{the resulting graph is} a tree, \blue{that is, has no cycles}. A \leoo{directed \blue{phylogenetic} network} is \emph{level-$k$} if its underlying network is level-$k$. \latest{Hence, a directed binary phylogenetic network is level-$k$ if and only if each blob contains at most~$k$ reticulations.}

\blue{A graph is {\em mixed} if it contains} both undirected and directed edges. A \emph{partly-directed} \blue{phylogenetic} network is a mixed graph \blue{that is} obtained from an undirected \blue{phylogenetic} network by orienting a subset of its edges. \blue{An {\em orientation} of a partly-directed phylogenetic network $N$ on $X$ is a directed phylogenetic network on $X$ that is obtained from $N$ by inserting the root along a \may{directed or undirected edge, and orienting all undirected edges}.} A \emph{semi-directed} \blue{phylogenetic} network is a mixed graph obtained from a directed \blue{phylogenetic} network by unorienting all non-reticulation arcs and suppressing the root. \red{If the root, $\rho$ say, is incident with the arcs $(\rho, u)$ and $(\rho, v)$, where $u$ is a tree vertex and $v$ is a reticulation, then this process replaces $(\rho, u)$ and $(\rho, v)$ with the arc $(u, v)$. Note that, as the root has out-degree~$2$, it is not the parent of two reticulations.} \blue{Such networks are of interest because they are used in practical software}~\cite{SNAQ}. A semi-directed \blue{phylogenetic} network is a partly-directed \blue{phylogenetic} network but the converse is not true in general, see Figure~\ref{fig:notsemidirected}.

\begin{figure}
\centering
 \begin{tikzpicture}
	 \tikzset{edge/.style={thick}}
     \tikzset{arc/.style={-Latex,thick}}
	 \begin{scope}[xshift=0cm,yshift=0cm]
	\draw[thick, fill, radius=0.06] (0,0) circle;
	\draw[thick, fill, radius=0.06] (-.5,.5) circle node[above left] {$x$};
	\draw[edge] (-.5,.5) -- (0,0);
	\draw[thick, fill, radius=0.06] (2,0) circle;
    \draw[arc] (0,0) -- (2,0);
    \draw[thick, fill, radius=0.06] (0,-2) circle;
    \draw[edge] (0,0) -- (0,-2);
    \draw[thick, fill, radius=0.06] (1,-1) circle;
    \draw[arc] (1,-1) -- (0,-2);
    \draw[arc] (1,-1) -- (2,0);
    \draw[edge] (1,-1) -- (2,-2);
    \draw[edge] (1,-1) -- (1.5,-1.5);
    \draw[thick, fill, radius=0.06] (2,-2) circle;
    \draw[edge] (1,-2) -- (2,-2);
    \draw[arc] (1,-2) -- (0,-2);
    \draw[arc] (2,0) -- (2,-1);
    \draw[arc] (2,-2) -- (2,-1);
    \draw[thick, fill, radius=0.06] (2,-1) circle;
    \draw[thick, fill, radius=0.06] (2.5,-1) circle node[right] {$y$};
    \draw[edge] (2,-1) -- (2.5,-1);
    \draw[thick, fill, radius=0.06] (1,-2) circle;
    \draw[thick, fill, radius=0.06] (1,-2.5) circle node[below] {$z$};
	\draw[edge] (1,-2) -- (1,-2.5);
    \end{scope}
    \begin{scope}[xshift=5cm,yshift=0cm]
	\draw[thick, fill, radius=0.06] (0,0) circle;
	\draw[thick, fill, radius=0.06] (-.5,.5) circle node[above left] {$x$};
	\draw[edge] (-.5,.5) -- (0,0);
	\draw[thick, fill, radius=0.06] (2,0) circle;
    \draw[arc] (0,0) -- (2,0);
    \draw[thick, fill, radius=0.06] (0,-2) circle node[left] {$q$};
    \draw[arc] (0,0) -- (0,-2);
    \draw[thick, fill, radius=0.06] (1,-1) circle  node[above] {$r$};
    \draw[edge] (1,-1) -- (0,-2);
    \draw[arc] (1,-1) -- (2,0);
    \draw[edge] (1,-1) -- (2,-2);
    \draw[edge] (1,-1) -- (1.5,-1.5);
    \draw[thick, fill, radius=0.06] (2,-2) circle  node[right] {$s$};
    \draw[edge] (1,-2) -- (2,-2);
    \draw[arc] (1,-2) -- (0,-2);
    \draw[arc] (2,0) -- (2,-1);
    \draw[arc] (2,-2) -- (2,-1);
    \draw[thick, fill, radius=0.06] (2,-1) circle;
    \draw[thick, fill, radius=0.06] (2.5,-1) circle node[right] {$y$};
    \draw[edge] (2,-1) -- (2.5,-1);
    \draw[thick, fill, radius=0.06] (1,-2) circle node[above] {$p$};
    \draw[thick, fill, radius=0.06] (1,-2.5) circle node[below] {$z$};
	\draw[edge] (1,-2) -- (1,-2.5);
    \end{scope}
	\end{tikzpicture}
\caption{\label{fig:notsemidirected} \remiee{\blue{Left, a partly-directed phylogenetic network that} is semi-directed (it can be rooted \blue{along} the pendant edge \blue{incident with} $z$). \blue{Right, a partly-directed phylogenetic network that} is not semi-directed. If it were \blue{semi-directed}, \blue{then a directed phylogenetic network from which it is obtained} would have to be rooted \blue{along either} the pendant edge \blue{incident with} $x$ or one of the arcs incident to the neighbour of $x$; \red{otherwise, there is no directed path from the root to $x$}. This makes $(p,z)$ an arc, which implies that $p$ has the incoming arc $(s,p)$. For similar reasons, the orientation must include $(r,s)$ and $(q,r)$. \blue{But then}, together with $(p,q)$, these arcs form a directed cycle, a contradiction.}}
\end{figure}







We emphasize that we do not allow parallel edges or parallel arcs in (\blue{undirected and} directed) \blue{phylogenetic} networks. \blue{However,} replacing directed arcs of a directed \blue{phylogenetic} network by undirected edges and suppressing the root may create parallel edges. We do not consider this case explicitly because it can be dealt with easily. \blue{In particular,} if \blue{an undirected phylogenetic} network has more than one pair of parallel edges, it cannot be oriented; \blue{since} the oriented \blue{phylogenetic} network would contain \blue{either a} pair of parallel \blue{arcs} \remiee{or a directed cycle of length~$2$}. If there is exactly one pair of parallel edges, then, for the same reason, one of these edges needs to be subdivided with the root to obtain an orientation.



\blue{Lastly, for an (undirected) graph $G=(V, E)$, let $E'$ and $V'$ be subsets of $E$ and $V$, respectively. The graph obtained from $G$ by deleting each of the edges in $E'$ is denoted by $G\backslash E'$. Similarly, the graph obtained from $G$ by deleting each of the vertices in $V'$ is denoted by $G\backslash V'$. On the other hand, if $A$ and $B$ are sets, the set obtained from $B$ by deleting each of the elements in $A\cap B$ is denoted by $B-A$.}

\section{\remiee{Orienting an \latest{undirected} phylogenetic network given the root and the desired in-degrees}}\label{sec:OrientingWithConstraints}

\leoo{Suppose} that~$N$ is an undirected binary \blue{phylogenetic} network, with a designated edge~$e_\rho$, and~$R$ is a subset of the vertices \blue{of $N$}. Does there exist an orientation $N^r$ of $N$ whose set of reticulations is $R$ and whose root subdivides $e_{\rho}$? In this section, we characterize precisely when \blue{there exists such an orientation}. \blue{Furthermore, we} prove that if \blue{an orientation exists, then it} is unique, and \may{we} present a \remiee{linear}-time algorithm that finds~$N^r$.

\leoo{We start by discussing \blue{nonbinary phylogenetic} networks, which then allows us to treat binary \blue{phylogenetic} networks as a special case. In directed \blue{nonbinary phylogenetic} networks, vertices may have both their in-degree and out-degree greater than~$1$, \blue{in which case} knowing the locations of the root and the reticulations may not guarantee a unique orientation of the network~(see Figure~\ref{fig:nonbinary}). Therefore, in addition to knowing which vertices are reticulations, we also need to know their desired in-degrees.
See Section~\ref{sec:discussion} for a discussion on nonbinary networks in which reticulations are required to have out-degree~1.}

\begin{figure}[h]
\centering
 \begin{tikzpicture}
	 \tikzset{edge/.style={thick}}
     \tikzset{arc/.style={-Latex,thick}}
	 \begin{scope}[xshift=0cm,yshift=0cm]
	\draw[fill] (-.1,.9) rectangle (.1,1.1);
    \draw[fill] (.9,.9) rectangle (1.1,1.1);
    \draw[thick, fill, radius=0.06] (0,2) circle;
    \draw[thick, fill, radius=0.06] (1,2) circle;
    \draw[thick, fill, radius=0.06] (.5,3) circle;
    \draw[thick, fill, radius=0.06] (0,0) circle node[below] {$x$};
    \draw[thick, fill, radius=0.06] (1,0) circle node[below] {$y$};
    \draw[arc] (0,1) -- (0,0);
    \draw[arc] (1,1) -- (1,0);
    \draw[arc] (0,1) -- (1,1);
    \draw[arc] (1,2) -- (1,1);
    \draw[arc] (0,2) -- (0,1);
    \draw[arc] (1,2) -- (0,1);
    \draw[arc] (0,2) -- (1,1);
    \draw[arc] (.5,3) -- (1,2);
    \draw[arc] (.5,3) -- (0,2);
    \end{scope}
    \begin{scope}[xshift=5cm,yshift=0cm]
    \draw[fill] (-.1,.9) rectangle (.1,1.1);
    \draw[fill] (.9,.9) rectangle (1.1,1.1);
    \draw[thick, fill, radius=0.06] (0,2) circle;
    \draw[thick, fill, radius=0.06] (1,2) circle;
    \draw[thick, fill, radius=0.06] (.5,3) circle;
    \draw[thick, fill, radius=0.06] (0,0) circle node[below] {$x$};
    \draw[thick, fill, radius=0.06] (1,0) circle node[below] {$y$};
    \draw[arc] (0,1) -- (0,0);
    \draw[arc] (1,1) -- (1,0);
    \draw[arc] (1,1) -- (0,1);
    \draw[arc] (1,2) -- (1,1);
    \draw[arc] (0,2) -- (0,1);
    \draw[arc] (1,2) -- (0,1);
    \draw[arc] (0,2) -- (1,1);
    \draw[arc] (.5,3) -- (1,2);
    \draw[arc] (.5,3) -- (0,2);
    \end{scope}
	\end{tikzpicture}
\caption{\label{fig:nonbinary} Two \blue{non-isomorphic} directed \blue{phylogenetic} networks that are both orientations of the same undirected \blue{phylogenetic} network with the same \blue{root location and the same} set of reticulations.}
\end{figure}

\leoo{In what follows, let $N=(V,E,X)$ denote an undirected \blue{nonbinary phylogenetic} network on~$X$ with vertex set~$V$ and edge set~$E$. In addition, \blue{let} $e_\rho$ denote a designated edge \blue{of $N$} where we want to insert the root and, \blue{for all $v\in V$, let} $d^-_N(v)$ \blue{and $d_N(v)$} denote the \emph{desired in-degree} \blue{and the total degree} of $v$, where~$1\leq d^-_N(v) \leq d_N(v)$, \blue{respectively}. We say that $(N,e_\rho,d^-_N)$ is \emph{orientable} and that~$N^r$ is a \emph{orientation} of~$(N,e_\rho,d^-_N)$ if there exists an orientation~$N^r$ of~$N$ such that its root subdivides~$e_\rho$ and each~$v\in V$ has in-degree~$d^-_N(v)$ in~$N^r$.} 
\remiee{\may{Observe} that $(N,e_{\rho},d_N^-)$ is not orientable if $d_N^-(v)=d_N(v)$ for \blue{some non-leaf} vertex $v$ of $N$, or if $d_N^-(l)\neq 1$ for \blue{some} leaf $l$ of $N$.} \remiee{This leads to the following decision problem.}

\remiee{\fbox{
\parbox{0.8\linewidth}{
{\sc Constrained Orientation}\\
{\bf Input:} An undirected \blue{nonbinary phylogenetic} network~$N=(V,E,X)$, a \red{distinguished edge} $e_{\rho}\in E$, and a map $d_N^-:V\to\mathbb{N}$ assigning a desired in-degree to each vertex of $N$.\\
{\bf Output:} An orientation of $(N,e_{\rho},d_N^-)$ if it exists, and NO otherwise.}
}}



\subsection{\latest{Characterizing the orientability of undirected nonbinary phylogenetic networks}}
\label{subsec:NonbinaryCharacterization}





We start by introducing the notion of a degree cut, \leoo{which will be the key ingredient for characterizing orientability.}

\begin{definition} 
Let~$N = (V,E,X)$ be an undirected \blue{nonbinary phylogenetic} network with~$e_\rho\in E$ a distinguished edge, and let~$N_\rho = (V_\rho,E_\rho,X)$ 
be the graph obtained from~$N$ by subdividing~$e_\rho$ by a new vertex~$\rho$.
Given
the desired in-degree~$d^-_N(v)$ of each vertex~$v\in V$, a \emph{degree cut} for~$(N,e_\rho,d^-_N)$ is a pair~$(V',E')$ with~$V'\subseteq V$ and~$E'\subseteq E_\rho$ such that the following hold \blue{in $N_{\rho}$}:
\begin{itemize}
\item $E'$ is an edge cut of~$N_\rho$;
\item $\rho$ is not in the same connected component of \blue{$N_{\rho}\backslash E'$} as any~$v\in V'$;
\item each edge in $E'$ is incident to exactly one element of~$V'$; and
\item each vertex~$v\in V'$ is incident to \may{at least one and} at most $d^-_N(v)-1$ edges in~$E'$.
\end{itemize}
\end{definition}
\noindent The notion of a degree cut is illustrated in \blue{Figure~\ref{fig:retcut1}}. \red{Observe that if the desired in-degree of each vertex in $V$ is at most one, then $(N, e_{\rho}, d^-_N)$ has no degree cut.} \yukiii{We say that a degree cut~$(V',E')$ for $(N,e_\rho, d^-_N)$ is \emph{minimal} if for any edge~$e\in E'$, we have that~$(V',E'-\{e\})$ is not a degree cut for~$(N,e_\rho, d^-_N)$.}

\begin{figure}[h]
\centering
 \begin{tikzpicture}
	 \tikzset{edge/.style={thick}}
     \tikzset{arc/.style={-Latex,thick}}
     \tikzset{vertex/.style={thick, fill, radius=0.06}}
     \tikzset{uvertex/.style={thick, fill=white, radius=0.08}}
	 \begin{scope}[xshift=0cm,yshift=0cm]
    \draw[vertex] (1,0) circle node[above] {$\rho$};
	\draw[vertex] (0,0) circle;
	\draw[vertex] (-1,0) circle node[left] {$x$};
	\draw[edge] (-1,0) -- (0,0);
	\draw[edge] (0,0) -- (1,0);
    \draw[edge,dotted] (1,0) -- (2,0);
    \draw[vertex] (1,-1) circle;
    \draw[edge] (0,0) -- (1,-1);
    \draw[edge,dotted] (2,0) -- (1,-1);
    \draw[vertex] (3,-1) circle;
    \draw[edge] (2,0) -- (3,-1);
    \draw[vertex] (4,-1) circle node[right] {$y$};
    \draw[edge] (3,-1) -- (4,-1);
    \draw[edge,dotted] (2,-2) -- (1,-1);
    \draw[edge] (2,-2) -- (3,-1);
    \draw[edge] (2,-2) -- (2,0);
    \draw[uvertex] (2,-2) circle node[below right] {$2$};
    \draw[uvertex] (2,0) circle node[above right] {$3$};
    \end{scope}
	\end{tikzpicture}
\caption{\label{fig:retcut1} Illustration of a degree cut. Shown is the graph~$N_\rho$ obtained from an undirected \blue{phylogenetic} network~$N$ by subdividing an edge~$e_\rho$ by a new vertex~$\rho$. Each vertex~$v$ with $d_N^-(v)>1$, represented by an unfilled vertex, is labelled by~$d_N^-(v)$. A degree cut~$(V',E')$ for~$(N,e_\rho,d^-_N)$ is indicated by taking~$V'$ to be the set of unfilled vertices and~$E'$ to be the set of dashed edges.}
\end{figure}


We will show \latest{in Theorem~\ref{thm:nonbinary_characterization}} that the non-existence of a degree cut \blue{for} $(N,e_{\rho},d^-_N)$ \blue{together with a \latest{condition on the desired in-degrees} is equivalent to $(N,e_{\rho},d^-_N)$ being orientable. One direction \latest{of this theorem} \blue{is} established in Proposition~\ref{prop:nonbinary_necessary}\latest{(i) and (ii)}.}

\begin{proposition}\label{prop:nonbinary_necessary}
Let~$N=(V,E,X)$ be an undirected \blue{nonbinary phylogenetic} network,~$e_\rho\in E$ be a distinguished edge, and~$d^-_N(v)$ be the desired in-degree of each vertex~$v\in V$, with~$d^-_N(v)=1$ if~$v$ is a leaf and~$1\leq d^-_N(v)< d_N(v)$ otherwise. If $(N,e_\rho,d^-_N)$ is orientable, then \blue{each of} the following \blue{holds}:
\begin{enumerate}[label=\textup{(\roman*)}]
\item $(N,e_\rho,d^-_N)$ has no degree cut;\label{prop:nn1}
\item $\sum_{v\in V}d^-_N(v)=|E|+1$;
\item $N\backslash R$ is a forest, where $R$ is the set of vertices in $V$ with desired in-degree at least two.
\end{enumerate}
\end{proposition}

\begin{proof}
\blue{To prove (i),} suppose, for a contradiction, that $(N,e_\rho,d^-_N)$ has a degree cut~$(V',E')$. Consider an orientation~$N^r$ of $(N,e_\rho,d^-_N)$. In this orientation, each vertex~$v\in V'$ is incident to at most $d^-_N(v)-1$ arcs corresponding to edges in $E'$. Hence, each vertex in~$V'$ is incident to at least one incoming arc that does not correspond to an edge in $E'$. \yukimm{Let~$v\in V'$ be an arbitrarily chosen vertex, and let~$e$ be an incoming arc of~$v$ that does not correspond to an edge in~$E'$. Since there is a directed path from~$\rho$ to \blue{$v$ via $e$} in~$N^r$, and since~$(V',E')$ is a degree cut of~$(N,e_\rho,d^-_N)$, it must be the case that, \blue{prior to $e$, this path traverses an} arc that corresponds to an edge in~$E'$. In particular, this means that \blue{there is a directed path from} some other vertex in $V'$ to $v$. Observe that this property holds for all vertices in $V'$, that is, \blue{for each $v'\in V'$, there is a directed path from} some other vertex in $V'$ \blue{to $v'$}. Since $V'$ is finite, this implies that~$N^r$ contains a cycle, a contradiction.}

  
  

\blue{For (ii),} the total in-degree in an orientation is $\sum_{v\in V}d^-_N(v)$.
\blue{Since an} orientation has $|E|+1$ edges as edge~$e_\rho$ of~$N$ is subdivided by~$\rho$, \blue{it follows that} $\sum_{v\in V}d^-_N(v)=|E|+1$.

\blue{To prove (iii),} suppose $N\backslash R$ contains a cycle~$C=(v_1,v_2,\ldots,v_1)$. Then, in an orientation \blue{$N^r$} of $(N,e_\rho,d^-_N)$, each vertex of~$C$ has one incoming and at least two outgoing arcs. Without loss of generality, suppose that~$\{v_1,v_2\}$ is oriented from~$v_1$ to~$v_2$ \blue{in $N^r$}. Then all other edges incident to~$v_2$ are oriented away from~$v_2$ \blue{in $N^r$}, so~$\{v_2,v_3\}$ is oriented from~$v_2$ to~$v_3$. By repeating this argument, it follows that \blue{$N^r$ has} a directed cycle $(v_1,v_2,\ldots,v_1)$. \blue{This contradiction completes the proof of (iii) and the proposition}.
\end{proof}


We will show later in Corollary~\ref{cor:nonbinary_yuki} that \blue{(iii)} in Proposition~\ref{prop:nonbinary_necessary} is implied by \blue{(i) and (ii)}. We next prove a lemma which will be \blue{used} in several proofs. \remiee{See Figure~\ref{fig:trEdgeDeletionLemma} for an example.}

\begin{figure}
    \centering
    \resizebox{\columnwidth}{!}{
     \begin{tikzpicture}
	 \tikzset{edge/.style={thick}}
     \tikzset{arc/.style={-Latex,thick}}
     \tikzset{vertex/.style={thick, fill, radius=0.06}}
     \tikzset{uvertex/.style={thick, fill=white, radius=0.08}}
	 \begin{scope}[xshift=0cm,yshift=0cm]
    \draw[vertex] (1,0) circle node[above] {$\rho$};
	\draw[vertex] (-1,0) circle node[left] {$x$};
	\draw[edge] (-1,0) -- (0,0);
	\draw[edge] (0,0) -- (1,0);
    \draw[edge] (1,0) -- (2,0);
    \draw[edge] (0,0) -- (1,-1);
    \draw[edge, dashed] (2,0) -- (1,-1);
    \draw[vertex] (3,-1) circle;
    \draw[edge] (2,0) -- (3,-1);
    \draw[vertex] (4,-1) circle node[right] {$y$};
    \draw[vertex] (2.5,-1.5) circle;
    \draw[vertex] (3, -2) circle node [right] {$w$};
    \draw[edge] (2.5,-1.5) -- (3,-2);
    \draw[edge] (3,-1) -- (4,-1);
    \draw[edge] (1.5,-1.5) -- (1,-1);
    \draw[edge, dotted] (2,-2) -- (1.5,-1.5);
    \draw[vertex] (1.5,-1.5) circle;
    \draw[vertex] (1,-2) circle node[left] {$z$};
    \draw[edge] (1.5,-1.5) -- (1,-2);
    \draw[edge] (3,-1) -- (2.5,-1.5);
    \draw[edge, dashed] (2,-2) -- (2.5,-1.5);
    \draw[edge, dashed] (2,-2) -- (2,0);
    \draw[uvertex] (2,-2) circle node[below right] {$2$};
    \draw[uvertex] (0,0) circle node[above left] {$2$};
    \draw[uvertex] (1,-1) circle node[below left] {$2$};
    \draw[vertex] (2,0) circle;
    \end{scope}
    
    \begin{scope}[shift = {(6,0)}]
    \draw[vertex] (1,0) circle node[above] {$\rho$};
	\draw[vertex] (-1,0) circle node[left] {$x$};
	\draw[edge] (-1,0) -- (0,0);
	\draw[edge] (0,0) -- (1,0);
    \draw[edge] (1,0) -- (2,0);
    \draw[edge] (0,0) -- (1,-1);
    \draw[edge] (2,0) -- (1,-1);
    \draw[vertex] (3,-1) circle;
    \draw[edge] (2,0) -- (3,-1);
    \draw[vertex] (4,-1) circle node[right] {$y$};
    \draw[vertex] (3, -2) circle node [right] {$w$};
    \draw[edge] (3,-1) -- (3,-2);
    \draw[edge] (3,-1) -- (4,-1);
    \draw[edge] (1.5,-1.5) -- (1,-1);
    \draw[vertex] (1.5,-1.5) circle;
    \draw[vertex] (1,-2) circle node[left] {$z$};
    \draw[edge] (1,-2) -- (1.5,-1.5);
    \draw[edge] (1.5,-1.5) -- (2,0);
    \draw[uvertex] (0,0) circle node[above left] {$2$};
    \draw[uvertex] (1,-1) circle node[below left] {$2$};
    \draw[vertex] (2,0) circle;
    \end{scope}
    
    \begin{scope}[shift = {(12,0)}]
    \draw[vertex] (1,0) circle node[above] {$\rho$};
	\draw[vertex] (-1,0) circle node[left] {$x$};
	\draw[edge] (-1,0) -- (0,0);
	\draw[edge, dashed] (0,0) -- (1,0);
    \draw[edge] (1,0) -- (2,0);
    \draw[edge] (0,0) -- (1,-1);
    \draw[edge, dashed] (2,0) -- (1,-1);
    \draw[vertex] (3,-1) circle;
    \draw[edge] (2,0) -- (3,-1);
    \draw[vertex] (4,-1) circle node[right] {$y$};
    \draw[vertex] (2.5,-1.5) circle;
    \draw[vertex] (3, -2) circle node [right] {$w$};
    \draw[edge] (2.5,-1.5) -- (3,-2);
    \draw[edge] (3,-1) -- (4,-1);
    \draw[edge] (1,-2) -- (1,-1);
    \draw[vertex] (1,-2) circle node[left] {$z$};
    \draw[edge] (3,-1) -- (2.5,-1.5);
    \draw[edge] (2.5,-1.5) -- (2,0);
    \draw[uvertex] (0,0) circle node[above left] {$2$};
    \draw[uvertex] (1,-1) circle node[below left] {$2$};
    \draw[vertex] (2,0) circle;
    \end{scope}
    
	\end{tikzpicture}
	}
    \caption{\remiee{An illustration of Lemma~\ref{lem:nonbinary_yuki}. Left, the \blue{graph} $N_\rho$ obtained \blue{from an undirected nonbinary phylogenetic network $N$} by subdividing the edge $e_\rho$ by $\rho$. \blue{Each vertex with $d^-_N(v) > 1$, represented by an unfilled vertex, is labelled by $d^-_N(v)$.}
    \blue{The triple $(N, e_{\rho}, d^-_N)$} has no degree cut. \blue{If $V$ denotes the vertex set of $N$ and $R$ denotes the set of unfilled vertices,} the dashed edges are those \blue{edges with end-vertices in $V-R$ and $R$} that can be reached from~$e_\rho$ \blue{without traversing an unfilled vertex} (`$\{t,r\}$ edges' in the setting of Lemma~\ref{lem:nonbinary_yuki}). The dotted edge \blue{is an edge with end-vertices in $V-R$ and $R$} that cannot be reached from $e_\rho$ \blue{without traversing an unfilled vertex}. Middle, the \blue{graph $\latest{N'_\rho}$} obtained by \blue{deleting} the dashed edge that is incident to the neighbour of~$w$ from~$N_\rho$ \blue{and suppressing the resulting degree-two vertices}. Observe that \blue{$(N', e_{\rho}, d^-_{N'})$, as defined in Lemma~\ref{lem:nonbinary_yuki}}, has no degree cut. 
    Right, the graph \blue{$N''$} obtained by \blue{deleting} the dotted edge from~$N_\rho$ \blue{and suppressing the resulting degree-two vertices}. \blue{Here, $(N'', e_{\rho}, d^-_{N''})$} has a degree cut (the unfilled vertices \blue{together with} the dashed edges).
    }}
    \label{fig:trEdgeDeletionLemma}
\end{figure}


\begin{lemma}\label{lem:nonbinary_yuki}
Let $N=(V,E,X)$ be an undirected \blue{nonbinary phylogenetic} network, $e_\rho\in E$ be a distinguished edge, and $d^-_N(v)$ be the desired in-degree of each vertex~$v\in V$, with~$d^-_N(v)=1$ if~$v$ is a leaf and~$1\leq d^-_N(v)< d_N(v)$ otherwise.
Let $R \may{ = \{v\in V : d^-_N(v) \geq 2\}}$ denote the set of all vertices \may{of $N$} 
with desired in-degree at least two. Suppose that $(N,e_\rho,d^-_N)$ has no degree cut \red{and $R\neq\emptyset$}. Then \blue{the following hold:}
\begin{enumerate}[label=\textup{(\roman*)}]
\item There exists \blue{an} edge~$\{t,r\}\latest{\neq e_\rho}$ in \latest{$N$}, \red{where}~$t\in V-R$ \red{and} $r\in R$, \red{such that} there \blue{is} a path \blue{from} \latest{an endpoint of~$e_\rho$} \blue{to} $t$ not \blue{traversing} any vertex \blue{in} $R$.

\item \latest{For any such edge~$\{t,r\}$ in (i),} $(N',e'_\rho,d_{N'}^-)$ has no degree cut, where
\begin{enumerate}[label=\textup{(\Roman*)}]
\item $N'$ is the \blue{undirected nonbinary phylogenetic} network obtained from~$N$ by deleting~$\{t,r\}$ and suppressing \blue{any resulting degree-two vertices},

\item $e'_{\rho}=e_{\rho}$ unless $e_{\rho}=\{\latest{p,q}\}$ and, \latest{$p$ say,} is suppressed \red{(and so $q$ is not suppressed as $\{t, r\}\neq e_{\rho}$)}, in which case, $e'_{\rho}=\latest{\{q,s\}}$, \latest{where $s$ is the neighbour of~$p$ that is not in~$\{q,r,t\}$}, and

\item $d_{N'}^-$ is the desired in-degrees of the vertices of~$N'$ \blue{with}
\[ d_{N'}^-(v) = 
    \begin{cases}
      d^-_N(v)-1, & \text{if~$v=r$;}\\
      d^-_N(v), & \text{otherwise,}
    \end{cases}
\]
for all vertices~$v$ in~$N'$.
\end{enumerate}
\end{enumerate}
\end{lemma}

\begin{proof} \latest{Let $N_{\rho}$ be the graph obtained from $N$ by subdividing $e_{\rho}$ with a vertex $\rho$.}
If, \blue{in $N_{\rho}$}, both \blue{vertices adjacent to $\rho$} are in~$R$, then these vertices together with the two edges incident with $\rho$ form a degree cut \blue{for $(N, e_{\rho}, d^-_N)$}, a contradiction. It follows that at least one vertex adjacent to $\rho$ is not in $R$. Avoiding $\rho$, take a path from such a vertex to a vertex~$r\in R$, such that no other vertices of the path except~$r$ are in~$R$.  \may{To show that such a path exists, assume it does not. This is only possible when exactly one neighbour of~$\rho$ is in~$R$ and there is no path avoiding~$\rho$ between the neighbours of~$\rho$, i.e., the edges incident to~$\rho$ are cut-edges. In this case, the neighbour of~$\rho$ that is in~$R$ together with the edge between this vertex and~$\rho$ form a degree cut for $(N, e_{\rho}, d^-_N)$, a contradiction. Hence, there exists a path from a neighbour of~$\rho$ that is not in~$R$ to a vertex~$r\in R$, such that this path does not contain~$\rho$ and does not contain any vertices from~$R$ except~$r$.}
Then the last edge on this path is an edge~$\{t,r\}$ with~$t\in V-R$ and~$r\in R$ for which there \blue{is} a path \blue{from $\rho$ to} $t$ not using any vertex from~$R$.
\latest{Note that~$t\neq\rho$ since we started the path at a neighbour of~$\rho$. Also note that $r\neq\rho$ since~$r\in R$. Hence, the edge~$\{t,r\}$ is not incident to~$\rho$ and so it is an edge of~$N$. Since it is also an edge of~$N_\rho$, it is not equal to~$e_\rho$.} This establishes (i).

\blue{To prove (ii)}, consider any such edge~$\{t,r\}$, and \blue{let $P$ be a path in $N_{\rho}$ from $\rho$ to $t$ avoiding vertices in $R$.} Suppose $(N',e'_\rho,d^-_{N'})$ has a degree cut~$(V',E')$. \latest{Let~$N'_\rho$ be the graph obtained from $N'$ by subdividing $e'_{\rho}$ with a vertex $\rho$. Observe that~$N'_\rho$ can be obtained from~$N_\rho$ by deleting~$\{t,r\}$ and suppressing any resulting degree-2 vertices (except~$\rho$). Also note that we can obtain~$N'$ from~$N'_\rho$ by suppressing~$\rho$ and that~$e_\rho'$ is the edge created by suppressing~$\rho$. In this proof, we will work with~$N_\rho$ and $N'_\rho$ (rather than with~$N$ and~$N'$) because degree cuts may contain edges incident to~$\rho$.}

\blue{If $t$ is suppressed when obtaining} \latest{$N'_\rho$ from~$N_\rho$}, \blue{let $e_t=\{u, v\}$ denote the resulting edge \latest{in~$N'_\rho$}, where $u\in P$} \latest{(possibly~$u=\rho$)}. 
\blue{Similarly, if $r$ is suppressed when obtaining \latest{$N'_\rho$ from $N_\rho$}, let $e_r=\{u', v'\}$ denote the resulting edge \latest{in~$N'_\rho$ (possibly, $\rho\in\{u',v'\}$)}.}

Let~$S$ be the subgraph of $\latest{N'_\rho}\backslash E'$ consisting of all connected components containing at least one element of~$V'$, and let~$S_\rho$ be the subgraph of $\latest{N'_\rho}\backslash E'$ consisting of \blue{the remaining} connected components \blue{ of $\latest{N'_\rho}\backslash E'$}. \blue{Since $(V', E')$ is a degree cut of $(N', e'_{\rho}, d^-_{N'})$, the subgraph} $S_\rho$ contains~$\latest{\rho}$.
\blue{Furthermore, as} \blue{$P$ contains no vertices in $R$} \may{(so for all $p\in P$, we have $d^-_N(p)\leq 1$),} \red{it follows by the fourth property of a degree cut that} \may{no $p\in P$ is an element of $V'$. Hence, the path $P$ cannot contain any edges of $E'$, so any node on $P$ is in the component $S_{\rho}$ and, in particular,} either $t$ or, \blue{if $t$ is suppressed,} $u$ is contained in $S_{\rho}$. 

We \may{now derive a contradiction by} distinguishing three cases \blue{depending on $r$}.

\blue{For the first case, assume} that either $r$ or, if $r$ is suppressed, $e_r$ is contained in $S_{\rho}$. \blue{If either $t$ is not suppressed or $t$ is suppressed and $e_t\not\in E'$, then} $(V', E')$ is a degree cut of $(N, e_{\rho}, d^-_N)$, \blue{a contradiction}. \latest{Now suppose that} $t$ is suppressed and $e_t\latest{=\{u,v\}}\in E'$. \latest{Since~$u\in S_\rho$ we have~$v\in S$. Then,} $(V', (E'-\{e_t\})\cup \{\{t, v\}\})$ is a degree cut of $(N, e_{\rho}, d^-_N)$, \blue{a} contradiction.

\blue{For the second case, assume} that either $r$ or, if $r$ is suppressed, $e_r$ is contained in $S$.
\latest{If} either $t$ is not suppressed or $t$ is suppressed and $e_t\not\in E'$, then $(V'\cup \{r\}, E'\cup \{\{t, r\}\})$ is a degree cut of $(N, e_{\rho}, d^-_N)$, \blue{a contradiction}. \blue{Furthermore,} if $t$ is suppressed and $e_t\in E'$, \blue{then} $(V'\cup \{r\}, (E'-\{e_t\})\cup \{\{t, v\}, \{t, r\}\})$ is a degree cut of $(N, e_{\rho}, d^-_N)$, a contradiction.

\blue{For} the last case, \blue{assume} that $r$ is suppressed and $e_r\latest{=\{u',v'\}}\in E'$. Without loss of generality, \blue{say} $v'\in V'$. \blue{If either $t$ is not suppressed or $t$ is suppressed and $e_t\not\in E'$,} then $(V', (E'-\{e_r\})\cup \{\{r, v'\}\})$ is a degree cut of $(N, e_{\rho}, d^-_N)$, \blue{a contradiction}. If $t$ is suppressed and $e_t\latest{=\{u,v\}}\in E'$, \blue{then} \latest{we have, as before, that~$u\in S_\rho$ and $v\in S$. In this case,} $(V', (E'-\{e_t, e_r\})\cup \{\{t, v\}, \{r,v'\}\})$ is a degree cut of $(N, e_{\rho}, d^-_N)$. This last contradiction completes the proof of the lemma.
\end{proof}


We are now ready to prove the \blue{above-mentioned} characterization for when an undirected \blue{nonbinary phylogenetic} network \blue{has an orientation respecting a} given location for the root and in-degree of every vertex.

\begin{theorem}\label{thm:nonbinary_characterization}
Let~$N=(V,E,X)$ be an undirected \blue{nonbinary phylogenetic} network,~$e_\rho\in E$ be a distinguished edge, and~$d^-_N(v)$ be the desired in-degree of each vertex~$v\in V$, with~$d^-_N(v)=1$ if~$v$ is a leaf and~$1\leq d^-_N(v)< d_N(v)$ otherwise. Then $(N,e_\rho,d^-_N)$ is orientable if and only if $(N,e_\rho,d^-_N)$ has no degree cut \blue{and} $\sum_{v\in V}d^-_N(v)=|E|+1$.
\end{theorem}

\begin{proof}
If $(N,e_\rho,d^-_N)$ is orientable, then, \blue{by Proposition~\ref{prop:nonbinary_necessary}(i) and (ii)}, it has no degree cut \blue{and $\sum_{v\in V} d^-_N(v)=|E|+1$}. \blue{The proof of the converse is} by induction on $\sum_{v\in V} d^-_N(v)- |V|$. \blue{Note that $\sum_{v\in V} d^-_N(v)-|V|\ge 0$ as $d^-_N(v)\ge 1$ for all $v\in V$.} If~$\sum_{v\in V}d^-_N(v)-|V| = 0$, then every vertex in~$V$ has desired in-degree 1. By assumption,~$|V| = \sum_{v\in V}d^-_N(v) = |E| + 1$, and so~$N$ is an undirected phylogenetic tree, \blue{in} which \blue{case, $(N, e_{\rho}, d^-_N)$} is \blue{trivially} orientable.

\remiee{Now suppose \blue{that} $\sum_{v\in V} d^-_{N}(v)-|V|\geq 1$, and the \blue{converse} holds for any \blue{undirected nonbinary phylogenetic} network \blue{in which the sum of the given in-degree of each vertex minus the size of its vertex set is at most} $(\sum_{v\in V} d^-_{N}(v)-|V|)-1$.} Let $R$ denote the set of all vertices in $V$ with \remiee{desired} in-degree at least~2. \blue{Since $\sum_{v\in V} d^-_N(v)-|V|\ge 1$, it follows that $R$ is nonempty. Let $N_{\rho}$ be the graph obtained from $N$ by subdividing $e_{\rho}$ by $\rho$.} By Lemma~\ref{lem:nonbinary_yuki}, there exists \blue{an} edge~$\{t,r\}$ in \blue{$N_{\rho}$} with~$t\in V-R$ and~$r\in R$ for which there \blue{is} a path \blue{from $\rho$ to} $t$ not using any vertex from $R$.
In this case,~$t$ and~$r$ are both vertices of total degree at least~$2$, and so they cannot be leaves (i.e.,~$t$ and~$r$ must have required outdegree at least 1).
\blue{Set} $(N',e'_\rho,d_{N'}^-)$ \blue{to be} the same \blue{as its} namesake in the statement of Lemma~\ref{lem:nonbinary_yuki} \may{and let~$E'$ be the edge set of~$N'$.}
\blue{Recalling} that~$d_N(v)$ denotes the degree of a vertex~$v\in V$ \blue{and $\sum_{v\in V} d^-_N(v)=|E|+1$}, \blue{there are} four \blue{possibilities to consider} \blue{depending} on the degree of~$t$ and the degree of~$r$ in~$N$:
\begin{enumerate}[label=\textup{\red{$\bullet$}}]
    \item If~$d_N(t) = 3$ and~$d_N(r) = 3$, then both~$t$ and~$r$ are suppressed in \blue{obtaining} $N'$, and so
    $$\sum_{v\in V'}d^-_{N'}(v) = \sum_{v\in V}d^-_N(v)-3 = (|E| + 1) - 3 = (|E'|+3) - 2 = |E'| + 1.$$
    \item If~$d_N(t) = 3$ and~$d_N(r) >3$, then only~$t$ is suppressed in \blue{obtaining} $N'$, and so $$\sum_{v\in V'}d^-_{N'}(v) = \sum_{v\in V}d^-_N(v) - 2 = (|E| + 1) - 2 = (|E'|+2) - 1 = |E'| + 1.$$
    \item If~$d_N(t) >3$ and~$d_N(r) = 3$, then only~$r$ is suppressed in \blue{obtaining} $N'$, and so $$\sum_{v\in V'}d^-_{N'}(v) = \sum_{v\in V}d^-_N(v) - 2 = (|E| + 1) - 2 = (|E'|+2) - 1 = |E'| + 1.$$
    \item If~$d_N(t) >3$ and~$d_N(r) > 3$, then neither~$t$ nor~$r$ \blue{is} suppressed in \blue{obtaining} $N'$, and so
    $$\sum_{v\in V'}d^-_{N'}(v) = \sum_{v\in V}d^-_N(v) - 1 = (|E| + 1) - 1 = |E'| + 1.$$
\end{enumerate}
In all \blue{four possibilities}, $\sum_{v\in V'}d^-_{N'}(v) = |E'| + 1$. \blue{Furthermore, a routine check using the above calculations shows that, for all four possibilities,}
$\sum_{v\in V'}d^-_{N'}(v)-|V'| < \sum_{v\in V}d^-_N(v)-|V|$. \blue{By Lemma~\ref{lem:nonbinary_yuki},} $(N',e'_\rho,d^-_{N'})$ has no degree cut; we also have~$d^-_{N'}(v)=1$ if~$v$ is a leaf and~$1\leq d^-_{N'}(v)< d_{N'}(v)$ otherwise. It follows by \blue{the} induction \blue{assumption} that $(N', e'_{\rho}, d^-_{N'})$ is orientable.
Now consider such an orientation, $(N')^r$ say, and impose the same arc directions on~$N_\rho$ except for the edge~$\{t,r\}$. \blue{If $t$ is suppressed in obtaining $N'$, then $d_{N_{\rho}}(t)=3$, in which case, the two edges incident with $t$ that are not $\{t, r\}$ are oriented to respect the orientation of the corresponding edge in $(N')^r$.}
\blue{Analogously, the edges incident with $r$ that are not $\{t, r\}$ are orientated in a similar way if \red{$d_{N_{\rho}}(r)=3$}.}
Now \red{orient} $\{t,r\}$ from~$t$ to~$r$, \blue{and} let $N^r$ denote \blue{the resulting orientation of $N_{\rho}$}. \blue{It follows by construction that each vertex in $N^r$ has} the correct in-degrees.

It remains to show that~$N^r$ is an orientation of~$(N, e_{\rho}, d^-_N)$ by showing that~$N^r$ has no directed cycle. \blue{If} there exists such a cycle, then this \blue{directed} cycle uses the oriented edge~$(t,r)$ \blue{as $(N')^r$ has no directed cycle}. Hence \blue{$N^r$ has} a directed path~$P$ from~$r$ to~$t$. On the other hand, by the choice of~$t$, \blue{the directed graph $N^r$ has} a \blue{directed} path~$Q$ from~$\rho$ to~$t$ not using any vertex from~$R$. Since both~$P$ and~$Q$ end in~$t$, they must meet. Let~$v$ be the first vertex on~$Q$ meeting~$P$. Then~$v\neq\rho$ \blue{as $P$ starts at $r\neq \rho$ and} both arcs incident \blue{with} $\rho$ are directed away from~$\rho$. Therefore $v$ \blue{has in-degree \may{at least}~$2$}. \blue{But} $Q$ does not contain any vertices \blue{in} $R$. \blue{This contradiction} completes the proof of the theorem.
\end{proof}

\blue{A consequence of Theorem~\ref{thm:nonbinary_characterization} is} that Proposition~\ref{prop:nonbinary_necessary}(iii) is implied by \blue{Proposition~\ref{prop:nonbinary_necessary}(i) and (ii)}.

\begin{corollary}\label{cor:nonbinary_yuki}
Let~$N=(V,E,X)$ be an undirected \blue{nonbinary phylogenetic} network,~$e_\rho\in E$ be a distinguished edge, and~$d^-_N(v)$ be the desired in-degree of each vertex~$v\in V$, with~$d^-_N(v)=1$ if~$v$ is a leaf and~$1\leq d^-_N(v)< d_N(v)$ otherwise.
Let~$R$ denote the set of all vertices in~$V$ with \remiee{desired in-degree} at least~$2$.
If $(N,e_\rho,d^-_N)$ has no degree cut and $\sum_{v\in V} d^-_N(v)=|E|+1$, then $N\backslash R$ is a forest.
\end{corollary}

\begin{proof}
If $(N,e_\rho,d^-_N)$ has no degree cut and $\sum_{v\in V} d^-_N(v)=|E|+1$ then, by Theorem~\ref{thm:nonbinary_characterization}, $(N,e_\rho,d^-_N)$ is orientable. It \blue{now} follows \blue{by} Proposition~\ref{prop:nonbinary_necessary} that $N\backslash R$ is a forest.
\end{proof}

\subsection{Orientation algorithm}

In this section, we present a \blue{polynomial-time} algorithm for \blue{deciding if}, given an undirected \blue{nonbinary phylogenetic} network \blue{$N$}, \blue{there is} an orientation \blue{of $N$ respecting a} given location of the root and desired in-degree of each vertex, \blue{in which case, the algorithm returns such an orientation}. \blue{The algorithm is as follows.}

\begin{algorithm}[H]
\KwIn{An undirected \blue{nonbinary phylogenetic} network~$N = (V,E,X)$, an edge~$e_\rho\in E$, and the desired in-degree~$d^-_N(v)$ for each $v\in V$, with~$d^-_N(v)=1$ if~$v$ is a leaf and~$1\leq d^-_N(v)< d_N(v)$ otherwise.}
\KwOut{An orientation of~$(N,e_\rho,d^-_N)$ if it exists and NO otherwise.}
\uIf{$\sum_{v\in V} d^-_N(v)\neq |E|+1$}{\Return NO}
Subdivide~$e_\rho$ by a new vertex~$\rho$ and orient the two edges incident to~$\rho$ away from~$\rho$\;
\While{there exist an unoriented edge}{
\uIf{there is a vertex $v\in V$ with $d^-_N(v)$ incoming oriented edges and at least one incident unoriented edge}{orient all unoriented edges incident to~$v$ away from~$v$}
\uElse{\Return NO}
}
\Return the \blue{resulting} orientation
\caption{\remiee{{\sc Orientation Algorithm}$(N,e,d^-_N)$}\label{alg:nonbinary_orientation}}
\end{algorithm}

\begin{theorem}\label{thm:nonbinary_alg}
Let~$N=(V,E,X)$ be an undirected \blue{nonbinary phylogenetic} network,~$e_\rho\in E$ be a distinguished edge, and~$d^-_N(v)$ be the desired in-degree of each vertex~$v\in V$, with~$d^-_N(v)=1$ if~$v$ is a leaf and~$1\leq d^-_N(v)< d_N(v)$ otherwise.
Then Algorithm~\ref{alg:nonbinary_orientation} decides whether $(N,e_\rho,d^-_N)$ is orientable, \blue{in which case, it} finds an orientation in time $O(|E|)$. \yuki{Moreover, this orientation is the unique orientation of~$(N,e_\rho,d^-_N)$.}
\end{theorem}

\begin{proof}
\delete{We claim that Algorithm~\ref{alg:nonbinary_orientation} can be used for this purpose.}
By Proposition~\ref{prop:nonbinary_necessary}\blue{(ii)}, we may assume that $\sum_{v\in V} d^-_N(v)=|E|+1$. Let $N_{\rho}$ denote the graph obtained from $N$ by subdividing $e_{\rho}$ with $\rho$.
We say that a vertex of~$N_\rho$ is \emph{processed} by Algorithm~\ref{alg:nonbinary_orientation} when the algorithm orients its outgoing edges. Note that Algorithm~\ref{alg:nonbinary_orientation} only processes a vertex when it already has at least one incoming oriented edge, and when a vertex is processed all its remaining unoriented edges are oriented outwards.

First \blue{suppose} that there exists an orientation~$N^r$ of $(N, e_{\rho}, d^-_N)$. We will prove that Algorithm~\ref{alg:nonbinary_orientation} returns~$N^r$. To see this, we first show that \blue{if a} vertex of $N_\rho$ \blue{is} processed by Algorithm~\ref{alg:nonbinary_orientation}, \blue{then every} edge incident to \blue{this vertex obtains} the same orientation as in~$N^r$. \blue{Assume}, for a contradiction, that this is not the case, and let \blue{$v$} be the first vertex processed by Algorithm~\ref{alg:nonbinary_orientation} for which at least one of its incident edges is not oriented as in~$N^r$. Immediately before \blue{$v$} is processed, it has \blue{$d^-_N(v)$} incoming oriented edges and at least one incident unoriented edge. By the choice of \blue{$v$}, the incoming oriented edges \blue{of $v$} are oriented the same way as in~$N^r$ because the other \blue{end-vertices} of these edges have already been processed. Algorithm~\ref{alg:nonbinary_orientation} orients all other edges incident to \blue{$v$} away from \blue{$v$}. These edges are also oriented away from \blue{$v$} in $N^r$, since~$N^r$ is an orientation and \blue{$v$} is required to have in-degree \blue{$d^-_N(v)$}. This contradicts the assumption that at least one edge incident to \blue{$v$} does not have the same orientation as in~$N^r$. \blue{It follows} that if there exists an orientation~$N^r$ of $(N, e_{\rho}, d^-_N)$, \blue{then every} vertex processed by Algorithm~\ref{alg:nonbinary_orientation} \blue{has} all its incident edges \blue{assigned} the same orientation as in~$N^r$.
To prove that Algorithm~\ref{alg:nonbinary_orientation} returns~$N^r$, it remains to show that every non-leaf vertex is processed by the algorithm.

\blue{Assume} that Algorithm~\ref{alg:nonbinary_orientation} stops without having processed all non-leaf vertices. Let~$P$ be the set of vertices \blue{of $N_{\rho}$} that have been processed at this point. Let~$E'$ be the set of all edges \blue{of $N_{\rho}$} with exactly one \blue{end-vertex} in~$P$, and let~$V'$ be the set of \blue{all} vertices \blue{of $N_{\rho}$} not in~$P$ that are incident to an edge in~$E'$. \blue{Every} edge~$e \in E'$ is incident to one processed vertex $u\in P$ and one unprocessed vertex in~$V'$. By construction,~$e$ is oriented away from~$u$ and, \blue{by the previous argument}, \blue{$e$ has} the same orientation in $N^r$.

\blue{If $v\in V'$, then, as} every oriented edge is oriented in the same direction as in~$N^r$, \blue{we have that} $v$ is incident to at most~$d^-_N(v)$ incoming oriented edges. \blue{Also}, every edge in~$E'$ incident to~$v$ is \blue{oriented towards $v$}. If $v$ is incident to exactly~$d^-_N(v)$ edges in~$E'$, then~$v$ is processed by Algorithm~\ref{alg:nonbinary_orientation}, a contradiction. So $v$ is incident to fewer than~$d^-_N(v)$ edges in~$E'$. \blue{Since} $E'$ is  an edge cut of~$N_\rho$ such that~$\rho$ is not in the same connected component of~$N_{\rho}\backslash E'$ as any \blue{vertex in} $V'$, and each edge in~$E'$ is incident to exactly one element of~$V'$, \blue{it follows} that $(V',E')$ is a degree cut for~$(N, e_\rho, d^-_N)$, contradicting Proposition~\ref{prop:nonbinary_necessary}(i). \blue{This last contradiction implies that all non-leaf vertices of $N_{\rho}$ are processed.} Hence, if there exists an orientation~$N^r$ of $(N, e_{\rho}, d^-_N)$, Algorithm~\ref{alg:nonbinary_orientation} will return~$N^r$, and $N^r$ is the unique orientation of~$(N,e_\rho, d^-_N)$.

Now \blue{suppose} that Algorithm~\ref{alg:nonbinary_orientation} returns an orientation~$N^r$ of~$N_{\rho}$. We \blue{will} prove that~$N^r$ is an orientation of~$(N,e_\rho,d^-_N)$.
It \blue{suffices} to \blue{show} that all vertices of~$N^r$ have the correct in-degree and out-degree, and $N^r$ has no directed cycle.

\blue{Assume} that there exists some vertex~$u$ in~$N^r$ that does not have the correct in-degree and out-degree. Each vertex that is processed (as well as each leaf) always obtains the correct in-degree and out-degree. Hence $u$ has not been processed. Since all edges have been oriented \red{and edges are oriented away from a vertex only if that vertex is processed, it follows that $u$ has in-degree $d_N(v)$, and so~$u$ is not a leaf.
Thus, $d_N^-(u) < d_N(u)$ and so $u$ has in-degree at least $d^-_N(u)+1$. By a similar reasoning}, all vertices $v\in V$ have in-degree at least~$d^-_N(v)$. Hence, as~$\sum_{v \in V}d^-_N(v) = |E| +1$, the total in-degree \blue{of $N_{\rho}$} is at least $\sum_{v \in V}d^-_N(v) + 1 = |E| + 2$ . But this implies that the total number of edges in $N_\rho$ is at least $|E|+2$, a contradiction as $N_\rho$ has $|E|+1$ edges. \blue{Thus, every vertex of $N^r$ has the correct in-degree and out-degree.}

\blue{Now assume} that $N^r$ has a directed cycle. \blue{Since every vertex of $N^r$ has the correct in-degree and out-degree, every \latest{non-leaf} vertex has been processed.} Consider the vertex~$v$ of the cycle that is processed first. Let~$u$ be the neighbour of~$v$ on the cycle such that there is an oriented edge~$e$ from~$u$ to~$v$. As any oriented \blue{edge} incident to a vertex \blue{is} oriented away from that vertex when it is processed, $e$ must have been oriented before~$v$ was processed. But this implies that~$u$ was processed before~$v$, \blue{contradicting} our choice of~$v$. \blue{Thus $N^r$ has no directed cycles and it follows that}, if Algorithm~\ref{alg:nonbinary_orientation} returns an orientation of~$N$, it is an orientation of $(N, e_\rho, d^-_N)$.

\blue{To complete the proof of the theorem,} it remains to show that Algorithm~\ref{alg:nonbinary_orientation} runs in~$O(|E|)$ time. A naive implementation takes~$O(|V|^2)$ time, as there are~$O(|V|)$ vertices to process and it may take~$O(|V|)$ time to find the next vertex that can be processed and process it.
However, this running time can be improved by observing that any vertex \blue{$v$} (apart from the root~$\rho$) only becomes suitable for processing after \blue{it has $d^-_N(v)$} incoming \blue{oriented} edges. Thus it is enough to maintain a set \blue{$S$} of \blue{such} vertices and check, whenever an edge is oriented, whether an unprocessed \blue{end-vertex} of this edge should be added to \blue{$S$}. Then, instead of searching for a new vertex to process each time, we can simply take any vertex from the set \blue{$S$}. As each edge is oriented exactly once, the total time spent maintaining \blue{$S$} and orienting all edges is~$O(|E|)$.
\end{proof}



\noindent{\bf Partly-directed and semi-directed phylogenetic networks.}
\blue{We end this subsection with two consequences of Theorem~\ref{thm:nonbinary_characterization} and~Algorithm~\ref{alg:nonbinary_orientation} concerning partly-directed and semi-directed phylogenetic networks.} Recall that a partly-directed \blue{phylogenetic} network is a mixed graph obtained from an undirected \blue{phylogenetic} network by orienting some of \blue{its} edges. \yuki{Let $N=(V,E,A,X)$ \blue{be} a partly-directed \blue{phylogenetic} network on~$X$ with vertex set~$V$, undirected edge set~$E$, and directed edge set~$A$.} \blue{Let $e_{\rho}\in E$ and, for each $v\in V$, let $d^-_N(v)$ denote the} desired in-degree of $v$. \blue{We say that $(N, e_{\rho}, d^-_N)$ is {\em orientable} if there is an orientation of $N$ in which the root subdivides $e_{\rho}$ and, for each $v\in V$, the in-degree of $v$ is $d^-_N(v)$.} \blue{To decide if} $(N,e_\rho,d_N^-)$ is orientable, replace each arc \blue{of $N$} by an undirected edge and apply Algorithm~\ref{alg:nonbinary_orientation} to \blue{determine} whether there exists an orientation. If it exists, it is unique by Theorem~\ref{thm:nonbinary_alg}. Hence, we only need to check whether each arc in~$A$ is oriented the same way in the obtained orientation. \blue{Thus we have} the following corollary of Theorem~\ref{thm:nonbinary_alg}.

\begin{corollary}\label{cor:partlydirected}
Let~$N=(V,E,A,X)$ be a partly-directed \blue{nonbinary phylogenetic} network,~$e_\rho\in E$ and~$d_N^-(v)$ the desired in-degree \blue{of} each~$v\in V$. Then there exists a \leoo{linear-time} algorithm that decides whether $(N,e_\rho,d_N^-)$ is orientable and finds \blue{the unique} orientation if it exists.
\end{corollary}

We now consider semi-directed \blue{phylogenetic} networks. Recall that a semi-directed \blue{phylogenetic} network is a mixed graph \blue{that is} obtained from a directed \blue{phylogenetic} network by unorienting all \blue{non-reticulation} arcs and suppressing the root. \blue{We noted in Section~\ref{sec:Preliminaries} that} a partly-directed \blue{phylogenetic} network \blue{is not necessarily a semi-directed phylogenetic network.} \blue{Thus a natural question is whether} it is easy to decide \blue{if} a given partly-directed \blue{phylogenetic} network is semi-directed. \blue{Corollary~\ref{cor:partlydirected} allows us to answer this question positively.}

\blue{Let $N=(V, E, A, X)$ be a partly-directed nonbinary phylogenetic network on $X$. If there is a vertex of $N$ with exactly one incoming arc, then $N$ is not semi-directed, so we may assume that there are no such vertices. Let $R$ denote} the \blue{subset} of vertices \blue{of $N$} with \blue{at least two} incoming arcs. \blue{For each vertex $v\in V$}, define the desired in-degree $d_N^-(v)$ \blue{of $v$} as \blue{the number of arcs directed into $v$ if $v\in R$; otherwise, set $d^-_N(v)=1$ if $v\not\in R$.} For each choice of \blue{$e_{\rho}\in E$}, we apply Corollary~\ref{cor:partlydirected}. \blue{Then $N$} is semi-directed if and only if \blue{$(N, e_{\rho}, d^-_N)$ is orientable} for at least one choice of~$e_\rho$. \yuki{The running time is~$O(|E|^2)$, since \blue{there are} $|E|$ choices for~$e_\rho$ and Algorithm~\ref{alg:nonbinary_orientation} runs in~$O(|E|)$ time.} \blue{Hence we have} the following corollary.

\begin{corollary}\label{cor:semidirected}
Let~$N=(V,E,A,X)$ be a partly-directed \blue{nonbinary phylogenetic} network. Then we can decide in $O(|E|^2)$ time whether~$N$ is a semi-directed \blue{nonbinary phylogenetic} network.
\end{corollary}


\subsection{\latest{Characterizing the orientability of undirected binary phylogenetic networks}}

We now \red{consider} the special case of the decision problem {\sc Constrained Orientation} \red{for} \blue{undirected} binary \blue{phylogenetic} networks. Here, rather \latest{than} being given the desired in-degree of each vertex, we are simply given the set of desired reticulations as all such vertices have in-degree exactly two and all remaining vertices (except the root) have in-degree one.

\begin{definition}\label{def:retcut}
Let~$N = (V,E,X)$ be an undirected binary \blue{phylogenetic} network with~$e_\rho\in E$ a distinguished edge, and let~$N_\rho = (V_\rho,E_\rho,X)$ be the graph obtained from~$N$ by subdividing~$e_\rho$ by a new vertex $\rho$. Given the set of desired reticulations $R\subseteq V$, a \emph{reticulation cut} for~$(N,e_\rho,R)$ is a pair~$(R',E')$ with~$R'\subseteq R$ and~$E'\subseteq E_\rho$ such that the following hold \blue{in $N_{\rho}$}:
\begin{itemize}
\item $E'$ is an edge cut of~$N_\rho$;
\item $\rho$ is not in the same connected component of \blue{$N_{\rho}\backslash E'$} as any~$r\in R'$;
\item each edge in $E'$ is incident to exactly one element of~$R'$; and
\item $|R'|=|E'|$.
\end{itemize}
\end{definition}


\noindent Observe that, \red{if, in the definition of a degree cut, $N$ is binary, then, because of the fourth property of a degree cut, $V'$ is a subset of the set of vertices whose desired in-degree is two. Hence, the definition of} a reticulation cut \red{coincides with that of} a degree cut \red{when $N$ is binary}. \remiee{We say $(N,e_{\rho},R)$ is \emph{orientable} if $(N,e_{\rho},d^-_{N})$ is orientable, where $d^-_N(r)=2$ for all $r\in R$ and $d^-_N(v)=1$ for \blue{all $v\in V-R$}.} An example of a reticulation cut of the triple $(N,e_\rho,R)$ in Figure~\ref{fig:notorientable} is illustrated in Figure~\ref{fig:retcut}.


\begin{figure}[h]
\centering
 \begin{tikzpicture}
	 \tikzset{edge/.style={thick}}
     \tikzset{arc/.style={-Latex,thick}}
	 \begin{scope}[xshift=0cm,yshift=0cm]
    \draw[thick, fill, radius=0.06] (1,0) circle node[above] {$\rho$};
	\draw[thick, fill, radius=0.06] (0,0) circle;
	\draw[thick, fill, radius=0.06] (-1,0) circle node[left] {$x$};
	\draw[edge] (-1,0) -- (0,0);
	\draw[fill] (1.9,-.1) rectangle (2.1,.1);
    \draw[edge] (0,0) -- (1,0);
    \draw[edge,dotted] (1,0) -- (2,0);
    \draw[fill] (.9,-1.1) rectangle (1.1,-.9);
    \draw[edge,dotted] (0,0) -- (1,-1);
    \draw[edge] (2,0) -- (1,-1);
    \draw[thick, fill, radius=0.06] (3,-1) circle;
    \draw[edge] (1,-1) -- (3,-1);
    \draw[edge] (2,0) -- (3,-1);
    \draw[thick, fill, radius=0.06] (4,-1) circle node[right] {$y$};
    \draw[edge] (3,-1) -- (4,-1);
    \end{scope}
	\end{tikzpicture}
\caption{\label{fig:retcut} The \yukiii{graph}~$N_\rho$ obtained from the \blue{undirected binary phylogenetic} network in Figure~\ref{fig:notorientable} by subdividing~$e_\rho$ (the edge indicated with an arrow) by a new vertex~$\rho$. The set~$E'$ \blue{consisting of the} dotted edges and \blue{the set} $R'$ \blue{consisting} of the two square vertices \red{form} the reticulation cut~$(R',E')$.}
\end{figure}

\blue{The next proposition is a consequence of Proposition~\ref{prop:nonbinary_necessary}.}

\begin{proposition}\label{prop:necessary}
Let~$N=(V,E,X)$ be an undirected \leoo{binary} \blue{phylogenetic} network,~$e_\rho\in E$ and~$R\subseteq V$. If $(N,e_\rho,R)$ is orientable, then each of the following holds:
\begin{enumerate}[label=\textup{(\roman*)}]
\item $(N,e_\rho,R)$ has no reticulation cut;
\item $|R|=|E|-|V|+1$;
\item \blue{$N\backslash R$} is a forest.
\end{enumerate}
\end{proposition}

To illustrate Proposition~\ref{prop:necessary}, the example \blue{in} Figure~\ref{fig:notorientable} satisfies~(ii) and~(iii) but, \blue{as shown in} Figure~\ref{fig:retcut}, it does not satisfy~(i), and hence it is not orientable.

\blue{The next theorem is the \red{special case} of Theorem~\ref{thm:nonbinary_characterization} \red{when restricted to undirected binary phylogenetic networks}. It characterizes} when an undirected \leoo{binary} \blue{phylogenetic} network with given locations for the root and reticulations has an orientation. \blue{The correctness of this characterization follows from Theorem~\ref{thm:nonbinary_characterization}.}

\begin{theorem}\label{thm:characterization}
Let~$N=(V,E,X)$ be an undirected \leoo{binary} \red{phylogenetic} network,~$e_\rho\in E$ and~$R\subseteq V$. Then $(N,e_\rho,R)$ is orientable if and only if $(N,e_\rho,R)$ has no reticulation cut \blue{and} $|R|=|E|-|V|+1$.
\end{theorem}

\section{\remiee{Orientations within a specific subclass of \blue{directed} binary \blue{phylogenetic} networks}}\label{sec:fptalgs}

\markj{We \blue{now} turn our attention to deciding whether a given undirected \leoo{binary} \blue{phylogenetic} network \blue{has a $C$-orientation for} \blue{a given class $C$} of directed binary \blue{phylogenetic} networks. \blue{Unlike {\sc Constrained Orientation} we are given} no information about the location of the root or the reticulation vertices. Formally, given a class $C$ of directed binary \blue{phylogenetic} networks, the problem of interest is as follows:}

\markj{\fbox{
\parbox{0.8\linewidth}{
{\sc $C$-Orientation}\\
{\bf Input:} An undirected \leoo{binary} \may{phylogenetic} network~$N$.\\
{\bf Output:} A~$C$-orientation of~$N$ if it exists, and NO otherwise.}
}}

In this section, we present algorithms \blue{for} solving {\sc $C$-Orientation} for classes \blue{$C$} of directed binary \blue{phylogenetic} networks \blue{satisfying} certain properties. In \blue{the next section}, Section~\ref{sec:Classes}, we will show that \shortened{the class tree-child satisfies these properties. In Appendix~\ref{app} we use a similar approach to show the same holds for the classes} 
stack-free, 
tree-based, 
valid, 
orchard, 
and reticulation-visible networks \shortened{(all defined in Appendix~\ref{sec:classesintro})}. 

\leooo{This section is organised as follows. In Section~\ref{sec:RootingBlobs}, we describe an FPT algorithm for {\sc $C$-Orientation} that is parameterized by the reticulation number of \blue{$N$}. Subsequently, in Section~\ref{sec:PuttingBlobsTogether}, we extend this to an FPT algorithm for {\sc $C$-Orientation} but with the level \blue{of $N$} as the parameter. \remiee{These algorithms essentially guess the locations of the root and the reticulations, compute the unique corresponding orientation as in Section~\ref{sec:OrientingWithConstraints}, and determine whether it is within the required class. To get an FPT running time, \blue{$N$} needs to be reduced to a size which is dependent only on the reticulation number (or level) first. We will give \blue{such} a reduction for any class~\blue{$C$ of directed binary phylogenetic networks whose members}} satisfy three certain properties. Intuitively, these properties are as follows. First, membership of $C$ can be checked by considering each blob separately. \remiee{Second, if \blue{$N'$ is a directed binary phylogenetic} network in \blue{$C$ and} new leaves are attached to $N'$, then \blue{the resulting directed binary phylogenetic network is also} \blue{in $C$}. \blue{Lastly,} the third property is based on reducing ``chains'' (sequences of leaves whose neighbours form a path). The third property implies that if \blue{$N'$ is} a \blue{directed binary phylogenetic} network in \blue{$C$} and all chains of \blue{$N'$} are reduced to a certain constant length, then the \blue{resulting directed binary phylogenetic} network \blue{$N''$} is \blue{also} in \blue{$C$}. \blue{Additionally}, a \blue{particular}
relationship holds between the $C$-rooted edges of \blue{$N'$} and \blue{$N''$}.}} These three properties are formally defined in Definitions~\ref{def:blob-determined},~\ref{def:leaf-addable}, and~\ref{def:l-chain-reducible}, respectively.


\subsection{FPT algorithm parameterized by \blue{the} reticulation number}\label{sec:RootingBlobs}

For a class $C$ of directed \blue{binary phylogenetic} networks, we \leooo{\blue{begin} by describing a simple exponential-time algorithm, \blue{namely, Algorithm~\ref{alg:binary_blob_orientation}}, \blue{that} finds all edges of \blue{a given} undirected \blue{binary phylogenetic} network where the root can be inserted in order to obtain a $C$-orientation \blue{and, for all such edges,} one \blue{$C$-orientation}.} \blue{The FPT algorithm described later in this subsection uses Algorithm~\ref{alg:binary_blob_orientation} as a subroutine.}
\blue{Let $N$ be} an undirected \blue{binary phylogenetic} network, \blue{and let $e$ be an edge of $N$}. We say that $N$ can be \emph{$C$-\leooo{rooted} at $e$} if there is a $C$-orientation of~$N$ \blue{whose} root \blue{subdivides} $e$. \latest{If this is the case, we also say that~$e$ is a \emph{$C$-rooted edge} of~$N$.} 
\blue{If $e$ is incident to a leaf $l$ and $N$ can be $C$-rooted at $e$, we say that} $N$ can be \emph{$C$-\leooo{rooted} at $l$}.
For a set $X$ and a non-negative integer $n$, we let $\binom{X}{n}=\{Y\subseteq X: |Y|=n\}$ denote the set of size $n$ subsets of $X$.

\begin{algorithm}[H]
\KwIn{An undirected \remiee{binary} \blue{phylogenetic} network~$N=(V,E,X)$ with reticulation number~$k$.\deletee{, and a class \leooo{of directed networks}~$C$.}}
\KwOut{\remie{\leooo{The} set of \blue{$C$-rooted} edges \blue{of} $N$ \blue{and} a corresponding \blue{$C$-}orientation \blue{for each such edge}.}}
Set $L:=\emptyset$ for the root locations and orientations\;
\For{each edge $e$ of $N$}{
    \For{each guess $R\in 
    \binom{V}{k}$ of the \blue{set of} $k$ reticulation vertices}{
    \latest{Set $d_N^-(v)=2$ for each~$v\in R$ and $d_N^-(v)=1$ for each~$v\in V\setminus R$}\;
        Compute $N(e,R)=\text{{\sc Orientation Algorithm}}(N,e,\latest{d_N^-})$ \latest{(using Algorithm~\ref{alg:nonbinary_orientation})}\;
        \If{$N(e,R)$ is a $C$-orientation}{
            $L:=L\cup \{(e,N(e,R))\}$\;
            Quit the inner for-loop
        }
    }
}
\Return $L$
\caption{\blue{A} simple exponential-time \remiee{$C$-orientation} algorithm \remiee{for a class~$C$ of directed \blue{binary phylogenetic} networks}.\label{alg:binary_blob_orientation}}
\end{algorithm}

\noindent \remie{Note that Algorithm~\ref{alg:binary_blob_orientation} does not necessarily return all $C$-orientations of \blue{$N$}. Indeed, for each edge \blue{of $N$}, the inner loop quits \blue{(Line~7)} after one such orientation is found. To find the complete set of orientations, simply remove \blue{this} line.} \blue{The correctness of Algorithm~\ref{alg:binary_blob_orientation} and its running time is established in the next lemma.}

\begin{lemma}\label{lem:binary_blob_orientation}
\blue{Let $N=(V, E, X)$ be} an undirected binary \blue{phylogenetic} network with reticulation number $k$. \blue{Then} Algorithm~\ref{alg:binary_blob_orientation} applied to \blue{$N$} is correct and runs in $O(\red{n^{k+1}} (n+f_C(n,k)))$ time, where $n = |V|$  and $f_C(n,k)$ is the time-complexity of checking whether a directed binary \blue{phylogenetic} network with $n$ vertices and $k$ reticulations is in the class $C$ \blue{of directed binary phylogenetic networks}.
\end{lemma}

\begin{proof}
\blue{Let $e$ be an} edge of $N$. If \blue{$N$ can be $C$-rooted at $e$}, then there is a set \blue{$R$} of $k$ reticulations \blue{such that}, by Theorem~\ref{thm:nonbinary_alg}, \blue{{\sc Orientation Algorithm$(N, e, R)$} returns} a $C$-orientation of $N$ rooted \blue{along} $e$. \blue{Since} Algorithm~\ref{alg:binary_blob_orientation} \blue{checks} all possible locations for the $k$ reticulations, \blue{it} will find \blue{such a} $C$-orientation. \blue{The correctness of Algorithm~\ref{alg:binary_blob_orientation} now follows.}

\blue{For the} running-time, the outer loop runs $O(n)$ times, as the degree of every vertex \blue{of $N$} is at most three and so $|E|\le \frac{3}{2}n$.
The inner loop runs \remiee{at most} $\binom{n}{k}$ times.
Inside the inner loop, there are \blue{exactly} two parts that run in non-constant time. \blue{First, by} Theorem~\ref{thm:nonbinary_alg}, {\sc Orientation Algorithm} runs in $O(n)$ time and, \blue{second, by} definition, checking whether a directed \blue{binary phylogenetic} network with $n$ vertices and $k$ reticulations is in $C$ takes $O(f_C(n,k))$ time. \blue{These} combine to \blue{give} a \blue{total} running time of $O\left(\binom{n}{k} n (n+f_C(n,k))\right)$, \red{that is $O(n^{k+1}(n+f_C(n, k)))$}. 
\end{proof}


\leooo{To obtain an FPT algorithm \blue{for {\sc $C$-Orientation}}, we need to pose some restrictions on the class~$C$. The first of \blue{these restrictions} is described in Definition~\ref{def:blob-determined}. \blue{For} a blob~$B$ of a directed \blue{binary phylogenetic} network, the \blue{directed binary phylogenetic} network \emph{induced} by $B$ is obtained from~$B$ by \blue{adjoining, to} each vertex~$v$ \blue{of either} in-degree~$1$ and out-degree~$1$, or in-degree~$2$ and out-degree~$0$, a \blue{new} leaf $x$ and a \blue{new} arc $(v,x)$.}

\begin{definition}\label{def:blob-determined}
A class $C$ of directed \blue{binary phylogenetic} networks is {\em blob-determined} if \blue{the following property holds:} A directed \blue{binary phylogenetic} network $N$ is \blue{a member of} $C$ precisely if every \blue{network} induced \blue{by a} blob of $N$ is \blue{a member of} $C$.
\end{definition}

\red{Let $N$ be} an undirected (resp.\ directed) binary phylogenetic network \red{on $X$, and suppose that $e$ is a cut-edge (resp.\ cut-arc) of~$N$. A connected component of $N\backslash e$} that is an undirected (resp.\ directed) phylogenetic \red{tree on~$X'$, where $X'\subseteq X$, is called a {\em pendant phylogenetic subtree} of~$N$.} \blue{A pendant \red{phylogenetic} subtree is} \emph{trivial} if it consists of a single leaf; \blue{otherwise, it is {\em non-trivial}}. If a class $C$ of directed \blue{binary phylogenetic} networks is blob-determined, then, in deciding whether an undirected \blue{binary phylogenetic} network \blue{$N$} has a $C$-orientation, we may assume that $N$ has no non-trivial pendant \red{phylogenetic} subtrees. To see this, observe that if $N'$ is \blue{an undirected binary phylogenetic} network obtained \blue{from $N$} by replacing a pendant \red{phylogenetic} subtree \blue{with} a single leaf, \blue{say $l$}, then, \red{as $C$ is blob determined and, thus, the existence of a $C$-orientation depends only on the biconnected components of $N$, it follows that} $N'$ has a $C$-orientation if and only if $N$ has a $C$-orientation \red{(see Figure~\ref{fig:PendantTree})}. Moreover, if $e$ \blue{is the} pendant edge of~$N'$ \blue{incident with $l$}, then~$N'$ can be $C$-rooted at~$e$ if and only if~$N$ can be $C$-rooted at each edge of \blue{the} pendant \red{phylogenetic} subtree \blue{replaced by $l$} (\red{again}, see Figure~\ref{fig:PendantTree}). Hence, we will assume throughout the remainder of Section~\blue{\ref{sec:fptalgs}}, \blue{as well as} \blue{Section~\ref{sec:Classes}},
that \blue{if $N$ is an undirected binary phylogenetic network, then} \blue{$N$ has} no non-trivial pendant \red{phylogenetic} subtrees.

\latest{In addition, note that if~$N$ is an undirected binary phylogenetic network with reticulation number at most~$1$, then we can decide whether~$N$ can be $C$-rooted at an edge~$e$ by running Algorithm~\ref{alg:binary_blob_orientation}, with the running time being a polynomial in the number of vertices and the time needed to check membership of the class~$C$ (see Lemma~\ref{lem:binary_blob_orientation}). Therefore, we also assume throughout the remainder of Section~\blue{\ref{sec:fptalgs}}, \blue{as well as} \blue{Section~\ref{sec:Classes}},
that each undirected binary phylogenetic network has reticulation number at least~$2$.}

\begin{figure}
\centering
\begin{tikzpicture}[scale=.85]
	 \tikzset{edge/.style={thick}}
     \tikzset{arc/.style={-Latex,thick}}
	 \begin{scope}[xshift=0cm,yshift=0cm]
    \draw[thick, fill, radius=0.06] (1,0) circle node[right] {$l$};
    \draw[thick, fill, radius=0.06] (.5,0) circle node[above] {$\rho$};
    \draw[arc] (.5,0) -- (1,0);
    \draw[arc] (.5,0) -- (0,0);
    \draw[fill=lightgray] (-.9,0) ellipse (.9 and .9);
	\draw (0,-1.5) node {$N'$};
	\end{scope}
	\begin{scope}[xshift=3.5cm,yshift=0cm]
    \draw[thick, fill, radius=0.06] (1,0) circle;
    \draw[arc] (1,0) -- (0,0);
    \draw[fill=lightgray] (-.9,0) ellipse (.9 and .9);
    \draw[thick, fill, radius=0.06] (2,1) circle;
    \draw[thick, fill, radius=0.06] (2,-1) circle;
    \draw[thick, fill, radius=0.06] (3,1.5) circle node[right] {$x_1$};
    \draw[thick, fill, radius=0.06] (2.5,1.25) circle node[above left] {$\rho$};
    \draw[thick, fill, radius=0.06] (3,.5) circle node[right] {$x_2$};
    \draw[thick, fill, radius=0.06] (3,-.5) circle node[right] {$x_3$};
    \draw[thick, fill, radius=0.06] (3,-1.5) circle node[right] {$x_4$};
    \draw[arc] (2,1) -- (1,0);
    \draw[arc] (1,0) -- (2,-1);
    \draw[arc] (2,-1) -- (3,-.5);
    \draw[arc] (2,-1) -- (3,-1.5);
    \draw[arc] (2.5,1.25) -- (2,1);
    \draw[arc] (2,1) -- (3,.5);
    \draw[arc] (2.5,1.25) -- (3,1.5);
	\draw (.5,-1.5) node {$N$};
	\draw (-2,-1.9) node {(a)};
	\end{scope}
	\begin{scope}[xshift=10cm,yshift=0cm]
    \draw[thick, fill, radius=0.06] (1,0) circle  node[right] {$l$};
    \draw[arc] (0,0) -- (1,0);
    \draw[fill=lightgray] (-.9,0) ellipse (.9 and .9);
	\draw (0,-1.5) node {$N'$};
	\end{scope}
	\begin{scope}[xshift=13.5cm,yshift=0cm]
    \draw[thick, fill, radius=0.06] (1,0) circle;
    \draw[arc] (0,0) -- (1,0);
    \draw[fill=lightgray] (-.9,0) ellipse (.9 and .9);
    \draw[thick, fill, radius=0.06] (2,1);
    \draw[thick, fill, radius=0.06] (2,-1);
    \draw[thick, fill, radius=0.06] (3,1.5) circle node[right] {$x_1$};
    \draw[thick, fill, radius=0.06] (3,.5) circle node[right] {$x_2$};
    \draw[thick, fill, radius=0.06] (3,-.5) circle node[right] {$x_3$};
    \draw[thick, fill, radius=0.06] (3,-1.5) circle node[right] {$x_4$};
    \draw[arc] (1,0) -- (2,1);
    \draw[arc] (1,0) -- (2,-1);
    \draw[arc] (2,-1) -- (3,-.5);
    \draw[arc] (2,-1) -- (3,-1.5);
    \draw[arc] (2,1) -- (3,.5);
    \draw[arc] (2,1) -- (3,1.5);
	\draw (.5,-1.5) node {$N$};
	\draw (-2,-1.9) node {(b)};
	\end{scope}
	\end{tikzpicture}
\caption{\label{fig:PendantTree} \remiee{Orientating pendant \red{phylogenetic} subtrees. Two examples showing how orientations \blue{of an undirected binary phylogenetic} network \blue{$N$} \blue{and the undirected binary phylogenetic network $N'$} obtained from \blue{$N$ by replacing a} pendant \red{phylogenetic} subtree \blue{with a single leaf} can be derived from \blue{each other}. \blue{In the first example} (\latest{a}), the root \blue{is placed along an edge} in the pendant \red{phylogenetic} subtree of $N$ \blue{and, in the second example} (\latest{b}), \blue{the root is placed elsewhere}. Note that, \blue{for each example}, the orientation \blue{of the edges} within the grey circle \blue{are the same}. \blue{Thus}, for a blob-determined class $C$ of \blue{directed binary phylogenetic networks}, \blue{$N$} has a $C$-orientation if and only if \blue{$N'$} has a $C$-orientation.}}
\end{figure}



\leooo{To \blue{describe} the remaining two \blue{restrictions}, we need some \blue{additional} definitions.
\blue{Let $N$ be an undirected (resp.\ directed) phylogenetic} network.
\emph{Adding} a leaf to \blue{$N$} means that an edge, \blue{say} $\{u,v\}$ (resp.\ arc $(u,v)$), of \blue{$N$} is replaced by edges~$\{u,w\},\{w,v\},\{w,x\}$ (resp.\ arcs~$(u,w),(w,v),(w,x)$), \blue{where} $w$ \blue{is} a new vertex and $x$ \blue{is} a new leaf. \blue{The second restriction is described in Definition~\ref{def:leaf-addable}.}



\begin{definition}\label{def:leaf-addable}
\leooo{A class~$C$ of directed \blue{binary phylogenetic} networks is \emph{leaf-addable} if \blue{the following property holds:} If $N$ \blue{is a member of} $C$ \blue{and $N'$ is} obtained \blue{from $N$} by adding leaves, \blue{then $N'$} is \blue{a member of} $C$.}
\end{definition}

The \emph{generator} $G(N)$ of an undirected (resp.\ directed) \blue{binary phylogenetic} network $N$ is the \blue{undirected (resp.\ directed)} \leooo{multi-graph} obtained \blue{from $N$} by \blue{deleting} all \red{(trivial and non-trivial)} pendant \red{phylogenetic} subtrees \red{together with the edges (resp.\ arcs) joining the pendant phylogenetic subtrees to the rest of $N$}, and suppressing \blue{each of} the resulting vertices} \blue{of} degree~$2$ (resp.\ in-degree~$1$ and out-degree~$1$). \blue{Note that if $N$ is undirected, then, for the definition of $G(N)$, we additionally require} the reticulation number \blue{of $N$ to be} at least~$2$ \may{(which we assume already)}. \blue{Furthermore}, \blue{$G(N)$} may have parallel edges (resp.\ arcs), as well as undirected (resp.\ directed) loops. The \leoo{edges (resp.\ arcs) of} \blue{$G(N)$} are called \emph{sides}.

\latest{Let $N$ be an undirected binary phylogenetic network~$N$ and let \red{$s=\{u, v\}$} be a side of $G(N)$. \red{Let $P_s$ denote the} undirected path in~$N$ \red{starting at $u$ and ending at $v$} from which $s$ is obtained in the construction of $G(N)$ by suppressing degree-$2$ vertices.
A leaf $x$ \blue{of $N$} is said to be \emph{on} $s$, and $s$ is said to \emph{contain} $x$, if~$x$ is adjacent to an internal vertex of~$P_s$. \latest{Let~$n_s$ denote the number of leaves that are on side~$s$.} An edge of $N$ is \emph{on} $s$ if it is an edge of~$P_s$. \blue{If $P_s$ is the undirected path $u=u_0, e_0, u_1, e_1, \ldots, u_{n_s}, e_{n_s}, u_{n_s+1}=v$ and $c_i$ is the leaf adjacent to $u_i$ for all $i\in \{1, 2, \ldots, n_s\}$, then, \red{relative to $P_s$}, we say that the leaves $c_1, c_2, \ldots, c_{n_s}$ and the edges $e_0, e_1, \ldots, e_{n_s}$ of $N$ on $s$ are {\em ordered from $u$ to $v$}.} In addition, if~$e_\rho$} \blue{is} a distinguished edge \blue{in which} we want to insert the root, then $s$ is said to \emph{contain the root} if~$e_\rho$ is incident to an internal vertex of~$P_s$, \blue{that is} either~$e_\rho$ is on $s$ or \blue{$e_{\rho}$} is a pendant edge incident to an internal vertex of~$P_s$.

\latest{Similarly, if~$N$ is a directed binary phylogenetic network and~$s$ is a side of~$G(N)$, then~$P_s$ is the directed path in~$N$ from which~$s$ is obtained in the construction of~$G(N)$ by suppressing vertices of in-degree~$1$ and out-degree~$1$. A leaf $x$ \blue{of $N$} is said to be \emph{on} $s$, and $s$ is said to \emph{contain} $x$, if~$x$ is adjacent to an internal vertex of~$P_s$. Let~$n_s$ denote the number of leaves that are on side~$s$. An arc of $N$ is \emph{on} $s$ if it is an arc of~$P_s$. \blue{If $s=(u, v)$ is a side of $G(N)$, and $P_s$ is the directed path $u=u_0, e_0, u_1, e_1, \ldots, u_{n_s}, e_{n_s}, u_{n_s+1}=v$ and $c_i$ is the leaf adjacent to $u_i$ for all $i\in \{1, 2, \ldots, n_s\}$, then we say that the leaves $c_1, c_2, \ldots, c_{n_s}$ and the arcs $e_0, e_1, \ldots, e_{n_s}$ of $N$ on $s$ are {\em ordered from $u$ to $v$}.}}



Let $N$ be an undirected \blue{binary phylogenetic} network. Let $\ell$ be a non-negative integer, and let $s$ be a side of $G(N)$ that contains $n_s\geq\ell$ leaves \blue{of $N$}. Then the undirected \blue{binary phylogenetic} network obtained from $N$ by deleting $\blue{n_s}-\ell$ leaves \latest{that are} on $s$ and suppressing any resulting degree-$2$ vertices is \blue{said to be obtained from $N$ by} an {\em $\ell$-chain reduction on $s$}. More generally, an {\em $\ell$-chain reduction on $N$} consists of performing an $\ell$-chain reduction on each side of $G(N)$ \blue{containing} at least $\ell$ leaves.

\blue{The third restriction is described in Definition~\ref{def:l-chain-reducible}.}

\begin{definition}\label{def:l-chain-reducible}
\charles{Let $C$ be a class of directed \blue{binary phylogenetic} networks, and let $N$ be an undirected \blue{binary phylogenetic} network that is $C$-orientable. Let $N'$ be \blue{an} \blue{un}directed \blue{binary phylogenetic} network obtained from $N$ by an $\ell$-chain reduction \blue{on} $N$. Suppose that $s=\{u, v\}$ is a side of $G(N)$ that contains at least $\ell$ leaves \blue{of $N$}, and let \red{$P_s$ be the undirected path} $u=u_0, u_1, \ldots, \blue{u_{n_s}}, \blue{u_{n_s+1}}=v$ \red{of $N$ corresponding to $s$} ordered from $u$ to $v$. \blue{Viewing $s$ as a side of $G(N')$}, let $c'_1, c'_2, \ldots, c'_{\ell}$ denote the leaves \blue{of $N'$} on $s$ ordered from $u$ to $v$ \blue{and, for all $i\in\{1, 2, \ldots, \ell\}$, let $u'_i$ denote the unique vertex of $N'$ adjacent to $c'_i$}. We say that $N$ is {\em $\ell$-chain reducible \blue{along $s$}} if the following two properties hold:}
\begin{enumerate}[label=\textup{(\roman*)}]
\item If $N'$ can be $C$-rooted at \red{$\{u'_i, c'_i\}$} \latest{with $i\in \{1, 2, \ldots, \ell\}$}, then $N$ can be $C$-rooted at all edges incident to $u_j$ for all $j\in \{i, \blue{i+1}, \ldots, \blue{n_s}-(\ell-i)\}$.
\item If $N$ can be $C$-rooted at \blue{an edge} $e$ \blue{incident with} $u_j$ \latest{with $j\in \{1, 2, \ldots, \blue{n_s}\}$}, then~$N'$ can be $C$-rooted at \red{$\{u'_i, c'_i\}$} for some \blue{\latest{$i\in \{1, 2, \ldots, \ell\}$} satisfying} $j\in\{i, \blue{i+1},\ldots ,\blue{n_s}-(\ell-i)\}$.
\end{enumerate}
More generally, \blue{$N$ is {\em $\ell$-chain reducible} if $N$ is $\ell$-chain reducible along every side of $G(N)$ containing} at least \blue{$\ell$ leaves and the following property holds:}
\begin{enumerate}[label=\textup{(iii)}]
\item If $N$ can be $C$-rooted at an edge $e$ \blue{that is neither on a side $s$ containing at least $\ell$ leaves nor incident with a leaf on a side $s$ containing at least $\ell$ leaves}, then $N'$ can also be $C$-rooted at $e$.
\end{enumerate}
A class $C$ of directed \blue{binary phylogenetic} networks is {\em $\ell$-chain reducible} if \latest{every $C$-orientable undirected \blue{binary phylogenetic} network is $\ell$-chain reducible. This concludes Definition~\ref{def:l-chain-reducible}.}
\end{definition}

\noindent Properties (i) and (ii) in Definition~\ref{def:l-chain-reducible} are illustrated in Figure~\ref{fig:DefLChainRed}. \latest{Figure~\ref{fig:propiii} shows an example where Property (iii) is necessary.} \blue{Note that we can perform an $\ell$-chain reduction on any undirected binary phylogenetic network, but not every such network is $\ell$-chain reducible.}

\begin{figure}
    \centering
    \includegraphics[width=.8\textwidth]{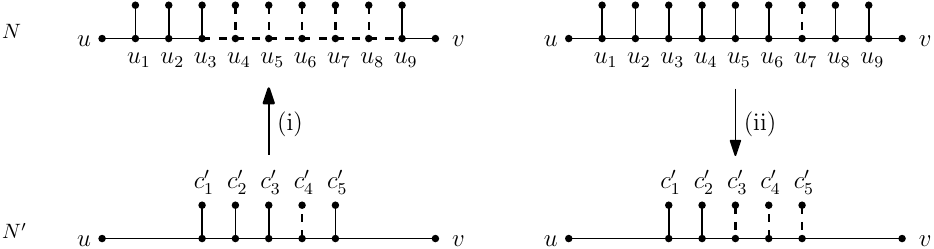}
   \caption{An illustration of \blue{an} $\ell$-chain reduction, \red{and (i) and (ii) of Definition~\ref{def:l-chain-reducible} applied to an undirected binary phylogenetic network $N$. The undirected binary phylogenetic network $N'$ has been obtained from $N$ by an $\ell$-chain reduction, where $\ell=5$. Here} the side $s=\{u,v\}$ of $G(N)$ is reduced from $n_s=9$ to $\ell=5$. \red{Now suppose that} $N$ is $5$-chain reducible. \blue{To illustrate (i) of Definition~\ref{def:l-chain-reducible}, if \latest{$N'$}} can be \blue{$C$-}rooted at leaf \red{$\{u'_4, c'_4\}$} (\blue{that is}, $i=4$), then $N$ can be \blue{$C$-}rooted at all edges \blue{incident with} $u_j$ for \blue{all} $j\in\{4, 5, \ldots, 8\}$ \red{(dashed edges)}. Furthermore, \blue{to illustrate (ii) of Definition~\ref{def:l-chain-reducible}}, if $N$ can be \blue{$C$-}rooted at the \blue{pendant} edge incident \blue{with} $u_7$, then $N'$ can be \blue{$C$-}rooted at \red{$\{u'_i, c'_i\}$} for some \red{$i\in \{1, 2, \ldots, 5\}$} \blue{satisfying} $7\in\{i, i+1,\ldots, i+4\}$, that is, for some $i\in \{3, 4, 5\}$ \red{(dashed edges)}.}
    \label{fig:DefLChainRed}
\end{figure}

\begin{figure}
    \centering
    \begin{tikzpicture}
	 \tikzset{edge/.style={thick}}
     \tikzset{arc/.style={-Latex,thick}}
	 \begin{scope}[xshift=0cm,yshift=0cm]
    \draw[thick, fill, radius=0.06] (0,0) circle;
	\draw[thick, fill, radius=0.06] (-.5,.5) circle;
	\draw[edge] (-.5,.5) -- (0,0);
	\draw (-.25,.25) node[above right] {$e$};
	\draw[thick, fill, radius=0.06] (2,0) circle;
    \draw[edge] (0,0) -- (2,0);
    \draw[thick, fill, radius=0.06] (0,-2) circle;
    \draw[edge] (0,0) -- (0,-2);
    \draw[thick, fill, radius=0.06] (1,-1) circle;
    \draw[edge] (1,-1) -- (0,-2);
    \draw[edge] (1,-1) -- (2,0);
    \draw[edge] (1,-1) -- (2,-2);
    \draw[edge] (1,-1) -- (1.5,-1.5);
    \draw[thick, fill, radius=0.06] (2,-2) circle;
    \draw[edge] (2,0) -- (2,-2);
    \draw[thick, fill, radius=0.06] (0,-3) circle;
	\draw[thick, fill, radius=0.06] (1,-3) circle;
	\draw[thick, fill, radius=0.06] (2,-3) circle;
    \draw[thick, fill, radius=0.06] (0,-3.5) circle;
	\draw[thick, fill, radius=0.06] (1,-3.5) circle;
	\draw[thick, fill, radius=0.06] (2,-3.5) circle;
	\draw[edge, dashed] (0,-2) -- (0,-3);
	\draw[edge, dashed] (0,-3) -- (0,-3.5);
	\draw[edge, dashed] (2,-2) -- (2,-3);
	\draw[edge, dashed] (2,-3) -- (2,-3.5);
	\draw[edge, dashed] (0,-3) -- (1,-3);
	\draw[edge, dashed] (1,-3) -- (2,-3);
	\draw[edge, dashed] (1,-3) -- (1,-3.5);
	\draw (1,-4) node {$N$};
    \end{scope}
    \begin{scope}[xshift=4cm,yshift=0cm]
    \draw[thick, fill, radius=0.06] (0,0) circle;
	\draw[thick, fill, radius=0.06] (-.5,.5) circle node[above] {$\rho$};
	\draw[thick, fill, radius=0.06] (-1,0) circle;
	\draw[arc] (-.5,.5) -- (0,0);
	\draw[arc] (-.5,.5) -- (-1,0);
	\draw[thick, fill, radius=0.06] (2,0) circle;
    \draw[arc] (0,0) -- (2,0);
    \draw[thick, fill, radius=0.06] (0,-2) circle;
    \draw[arc] (0,0) -- (0,-2);
    \draw[thick, fill, radius=0.06] (1,-1) circle;
    \draw[arc] (1,-1) -- (0,-2);
    \draw[arc] (2,0) -- (1,-1);
    \draw[arc] (1,-1) -- (2,-2);
    \draw[thick, fill, radius=0.06] (2,-2) circle;
    \draw[arc] (2,0) -- (2,-2);
    \draw[thick, fill, radius=0.06] (0,-3) circle;
	\draw[thick, fill, radius=0.06] (1,-3) circle;
	\draw[thick, fill, radius=0.06] (2,-3) circle;
    \draw[thick, fill, radius=0.06] (0,-3.5) circle;
	\draw[thick, fill, radius=0.06] (1,-3.5) circle;
	\draw[thick, fill, radius=0.06] (2,-3.5) circle;
	\draw[arc] (0,-2) -- (0,-3);
	\draw[arc] (0,-3) -- (0,-3.5);
	\draw[arc] (2,-2) -- (2,-3);
	\draw[arc] (2,-3) -- (2,-3.5);
	\draw[arc] (0,-3) -- (1,-3);
	\draw[arc] (2,-3) -- (1,-3);
	\draw[arc] (1,-3) -- (1,-3.5);
	\draw (1,-4) node {$N_r$};
    \end{scope}
    \begin{scope}[xshift=8cm,yshift=0cm]
    \draw[thick, fill, radius=0.06] (0,0) circle;
	\draw[thick, fill, radius=0.06] (-.5,.5) circle;
	\draw[edge] (-.5,.5) -- (0,0);
	\draw (-.25,.25) node[above right] {$e$};
	\draw[thick, fill, radius=0.06] (2,0) circle;
    \draw[edge] (0,0) -- (2,0);
    \draw[thick, fill, radius=0.06] (0,-2) circle;
    \draw[edge] (0,0) -- (0,-2);
    \draw[thick, fill, radius=0.06] (1,-1) circle;
    \draw[edge] (1,-1) -- (0,-2);
    \draw[edge] (1,-1) -- (2,0);
    \draw[edge] (1,-1) -- (2,-2);
    \draw[edge] (1,-1) -- (1.5,-1.5);
    \draw[thick, fill, radius=0.06] (2,-2) circle;
    \draw[edge] (2,0) -- (2,-2);
    \draw[thick, fill, radius=0.06] (0.5,-3) circle;
	\draw[thick, fill, radius=0.06] (1.5,-3) circle;
    \draw[thick, fill, radius=0.06] (0.5,-3.5) circle;
	\draw[thick, fill, radius=0.06] (1.5,-3.5) circle;
	\draw[edge] (0,-2) -- (.5,-3);
	\draw[edge, dashed] (.5,-3) -- (.5,-3.5);
	\draw[edge] (.5,-3) -- (1.5,-3);
	\draw[edge, dashed] (1.5,-3) -- (1.5,-3.5);
	\draw[edge] (2,-2) -- (1.5,-3);
	\draw (1,-4) node {$N'$};
    \end{scope}
	\end{tikzpicture}
   \caption{\latest{\red{An} example where (iii) of Definition~\ref{def:l-chain-reducible} is not satisfied. Suppose \red{that} $C$ is the class of stack-free networks (or reticulation-visible networks) and~$\ell=2$. Then \red{the undirected binary phylogenetic network} $N$ can be $C$-rooted at~$e$ since the directed network~$N_r$ is a $C$-orientation of $N$. However, a routine check shows that the \red{undirected binary phylogenetic} network~$N'$ obtained by performing an $\ell$-chain reduction on~$N$ cannot be $C$-rooted at~$e$.
   \red{Furthermore,} \red{both} (i) and (ii) of Definition~\ref{def:l-chain-reducible} \red{vacuously hold} since neither~$N$ nor~$N'$ can be $C$-rooted at any of the \red{dashed edges}.}}
    \label{fig:propiii}
\end{figure}

Let $C$ be an $\ell$-chain reducible, leaf-addable, blob-determined class of directed \blue{binary phylogenetic} networks.
We \red{next} describe \latest{an} FPT algorithm, namely, Algorithm~\ref{alg:binary_blob_orientation_TC_SF}, for $C$-\textsc{orientation} with the reticulation number \blue{of~$N$} as the parameter. \latest{In the description of Algorithm~\ref{alg:binary_blob_orientation_TC_SF}, recall that \emph{adding a leaf~$x$} to a directed network means that an arc~$(u,v)$ is subdivided with a new vertex,~$w$ say, to create the two arcs~$(u, w)$ and~$(w, v)$, and that leaf~$x$ is added with an arc~$(w, x)$ (so, in particular, the orientation of the added pendant arcs is determined). In particular, when we add back several leaves to form a chain, we repeat this operation sequentially for each leaf whilst respecting the ordering of the added leaves.

\red{As with} Algorithm~\ref{alg:binary_blob_orientation}, Algorithm~\ref{alg:binary_blob_orientation_TC_SF} finds all of the $C$-rooted edges of a given undirected binary phylogenetic network, \red{say $N$}, and, for all such edges, it also
finds a $C$-orientation.} \red{Loosely speaking, Algorithm~\ref{alg:binary_blob_orientation_TC_SF} starts by performing an $\ell$-chain reduction on $N$ to produce an undirected binary phylogenetic network $N^\ell$, and then, using Algorithm~\ref{alg:binary_blob_orientation}, finds all the $C$-rooted edges of $N^\ell$ as well as a $C$-orientation of $N^\ell$ for each such edge (Lines~1--2). For each $C$-rooted edge $e$ of $N^\ell$, the algorithm then iteratively finds several $C$-rooted edges of $N$ ``linked'' to $e$ via Definition~\ref{def:l-chain-reducible}. \may{It essentially does this by re-attaching the leaves that were removed in the $\ell$-chain reduction (after optionally first removing the leaf edge where the root is located and relocating the root to the resulting degree-2 node). It thus also} provides a corresponding $C$-orientation. 

Noting that $G(N)=G(N^\ell)$, let $s$ be the side of $G(N^\ell)$ that contains either $e$ if $e$ is not pendant or the leaf incident to $e$ if $e$ is pendant, and let $n_s$ be the number of leaves of $N$ on $s$. How this iterative process proceeds depends on whether (i) $n_s < \ell$ (Lines~7--10\may{; uses Definition~\ref{def:l-chain-reducible}(iii)}), (ii) $n_s\ge \ell$ and $e$ is a pendant edge of $N^\ell$ (Lines~11--30\may{; uses Definition~\ref{def:l-chain-reducible}(i)}), or (iii) $n_s\ge \ell$ and $e$ is not a pendant edge of $N^\ell$ (Lines~31--33\may{; we argue that we do not need to consider this case explicitly}).
Most of the work is in (ii) where Algorithm~\ref{alg:binary_blob_orientation_TC_SF} initially handles pendant edges of $N$ linked to $e$ (Lines~16--22) and then handles non-pendant edges of $N$ linked to $e$ (Lines~23--29). The fact that this process finds all $C$-rooted edges of $N$ as well as a corresponding $C$-orientation \may{of~$N$} for each such edge is established in Lemma~\ref{lem:fptalgretic}.}

\begin{algorithm}
\KwIn{An undirected \remiee{binary} \blue{phylogenetic} network~$N$ \blue{with} \may{reticulation number~$k\geq 2$ and} no non-trivial pendant \red{phylogenetic} subtrees.\deletee{, and an $l$-chain reducible, \leooo{leaf-addable, blob-determined} class~$C$.}}
\KwOut{The set of $C$-rooted edges \blue{of} $N$ and \leooo{a corresponding \blue{$C$-}orientation \may{of~$N$} for each such edge.}}
    Construct \latest{an} undirected binary \blue{phylogenetic} network~$N^\ell$ \remiee{by performing an $\ell$-chain reduction on $N$}\label{line:1}\; 
    \leooo{Find the set of $C$-rooted edges \blue{of} $N^\ell$ and a corresponding \blue{$C$-}orientation $N^\ell_e$ for each such edge~$e$} using Algorithm~\ref{alg:binary_blob_orientation}\;
    Set $L:=\emptyset$ for the root locations and orientations\;
    \For{each $C$-rooted edge $e$ \blue{of $N^{\ell}$}}{
    Let~$s=\{u,v\}$ be the side of $G(N^{\blue{\ell}})$ that contains either the leaf incident to~$e$ if~$e$ is pendant, or $e$ \red{itself} if~$e$ is not pendant\;
    Let~$n_s$ be the number of leaves \blue{of $N$} on~$s$\; 
        \If{$n_s < \ell$}{
            Extend~$N^\ell_e$
            to a $C$-orientation $N_{e}$ of $N$ by adding back the leaves deleted in the reduction \red{(in Line~\ref{line:1})} \latest{at their original location}\; 
            Set $L:=L\cup\{(e,N_{e})\}$\;
        }
        \If{$n_s\ge \ell$ and $e$ is a pendant edge, \blue{say $\{u'_i, c'_i\}$}, \blue{of $N^{\ell}$}}{
        Let $c'_1, c'_2, \ldots, c'_\ell$ be the leaves \blue{of} $N^{\ell}$ on~$s$ ordered from~$u$ to~$v$\;
        \blue{Let $e'_0, e'_1, \ldots, e'_{\ell}$ be the edges of $N^{\ell}$ on $s$ ordered from $u$ to $v$}\;
        Let~$c_1, c_2, \ldots, c_{n_s}$ be the leaves \blue{of} $N$ on~$s$ ordered from~$u$ to~$v$\;
        Let~$e_0, e_1,\ldots, e_{n_s}$ be the edges \blue{of} $N$ on~$s$ ordered from~$u$ to~$v$\;
        
        \For{\red{each} $j\in\{i, i+1, \ldots, \blue{n_s}-(\ell-i)\}$}{
        Let~$f$ be the pendant edge \blue{of} $N$ incident to $c_j$\;
        Modify~$N^\ell_e$
            to a $C$-orientation $N_{f}$ of $N$ \blue{as follows}. \blue{First, add back $(j-1)-(i-1)$ leaves \latest{to an arbitrary arc} on the directed path from $u'_i$ to $u$ and add back $(n_s-j)-(\ell-i)$ leaves \latest{to an arbitrary arc} on the directed path from $u'_i$ to $v$. \latest{Then, (re)label} the leaves ordered from $u'_i$ to $u$ as $c_{j-1}, c_{j-2}, \ldots, c_1$, \latest{(re)label the leaves ordered from $u'_i$ to $v$ as} $c_{j+1}, c_{j+2}, \ldots, c_{n_s}$ \latest{and relabel} the leaf adjacent to $u'_i$ as $c_j$. Now extend the resulting orientation by adding back the remaining leaves deleted in the reduction \red{(in Line~\ref{line:1})} \latest{at their original location}}\;
            \If{$L$ does not contain a pair with~$f$ as \latest{the} first element yet}{Set $L=L\cup\{(f,N_{f})\}$\;}
        }
        
        \For{\red{each} $j\in\{i-1, i, \ldots, n_s-(\ell-i)\}$}{
        Let~$f=e_j$\;
        Modify~$N^\ell_e$
            to a $C$-orientation $N_f$ of~$N$ \blue{as follows.} \blue{First delete $c'_i$ and the root, relocating the root to $u'_i$.} \blue{Second, add back $(j-1)-(i-1)$ leaves \latest{to an arbitrary arc} on the directed path from $u'_i$ (the new root) to $u$ and add back $(n_s-(j-1))-(\ell-i))$ leaves \latest{to an arbitrary arc} on the directed path from $u'_i$ to $v$, (re)labelling the leaves ordered from $u'_i$ to $u$ and from $u'_i$ to $v$ as $c_{j-1}, c_{j-2}, \ldots, c_1$ and $c_j, c_{j+1}, \ldots, c_{n_s}$, respectively. Now extend the resulting orientation by adding back the remaining leaves deleted in the reduction \red{(in Line~\ref{line:1})} \latest{at their original location}}\;
            \If{$L$ does not contain a pair with~$f$ as \latest{the} first element yet}{Set $L=L\cup\{(f,N_{f})\}$\;}
        }
        
        }
        \If{\blue{$n_s\ge \ell$ and $e$ is not a pendant edge of $N^{\ell}$}}{
        \red{Do nothing as $e$ is incident with a pendant $C$-rooted edge of $N^{\ell}$, and any corresponding $C$-orientation of $N$ is constructed in Lines 23--29;}
        }
    }
    \Return $L$\;
\caption{\leooo{\blue{An} FPT algorithm for $C$-\textsc{orientation} with the reticulation number of $N$ as the parameter, where~$C$ is an $\ell$-chain reducible, leaf-addable, and blob-determined class of directed \blue{binary phylogenetic} networks}.\label{alg:binary_blob_orientation_TC_SF}}
\end{algorithm}




\begin{lemma}
\blue{Let $N=(V, E, X)$ be an undirected binary phylogenetic network with reticulation number $k$, where $k\ge 2$.} \remiee{Then} Algorithm~\ref{alg:binary_blob_orientation_TC_SF} \blue{applied to $N$} is correct and runs in time $$O((8\ell(k-1))^{k+1}(\ell(k-1)+f_C(8\ell(k-1),k)) + \ell(k-1)n^2)
= O(g(k,\ell)+\ell(k-1)n^2),$$ where $n=|V|$, \blue{$f_C(8\ell(k-1), k)$ is the time complexity of checking whether a directed binary phylogenetic network \blue{with $8\ell(k-1)$ vertices and $k$ reticulations} is in the \marknew{$\ell$-chain reducible, leaf-addable, blob-determined} class $C$ of directed binary phylogenetic networks,}
and~$g$ is \remiee{a function of $k$ and $\ell$ independent of~$n$.}\label{lem:fptalgretic}
\end{lemma}

\begin{proof} To establish the lemma, we use the same notation as in Algorithm~\ref{alg:binary_blob_orientation_TC_SF}. \leooo{To prove correctness, we first show that the algorithm correctly infers $C$-rooted edges \blue{of} $N$ from the $C$-rooted edges \blue{of} $N^\ell$.}
\blue{Let $e$ be a $C$-rooted edge of $N^{\ell}$, and let $s$ be the side of $G(N^{\ell})$ containing either the leaf incident to $e$ if $e$ is pendant, or $e$ if $e$ is not pendant. If}  $n_s < \ell$, \blue{then $e$ is an edge of $N$ and, as} $C$ is leaf-addable, it follows that the algorithm correctly \blue{concludes that} $N$ can be $C$-rooted at $e$. \blue{On the other hand, if $n_s\ge \ell$, then, as $C$ is} $\ell$-chain reducible, \blue{it follows} by Property~(i) \blue{of Definition~\ref{def:l-chain-reducible} that each of the edges of $N$ inferred by} Algorithm~\ref{alg:binary_blob_orientation_TC_SF} is a $C$-rooted edge \blue{of} $N$ on side~$s$.


We now \blue{show} that Algorithm~\ref{alg:binary_blob_orientation_TC_SF} finds all $C$-rooted edges \blue{of} $N$. Suppose that~$N$ can be $C$-rooted at edge~$e_\rho$, and let~$s_\rho$ be the side of $G(N)$ that \blue{contains} either \blue{the} leaf incident to~$e_\rho$ if \blue{$e_{\rho}$} is pendant, or $e_\rho$ if \blue{$e_{\rho}$} is not pendant. \blue{First} suppose that~$s_{\rho}$ contains fewer than $\ell$ leaves of $N$. \blue{Then $e_{\rho}$ is an edge of $N^{\ell}$ and,} \remiee{by Property~(iii) of \blue{Definition~\ref{def:l-chain-reducible}}, $e_{\rho}$ \blue{is a $C$-rooted edge of} $N^\ell$. Thus, \blue{as Algorithm~\ref{alg:binary_blob_orientation}} \may{finds all $C$-rooted edges of $N^{\ell}$ with corresponding orientations}, the algorithm correctly finds $e_{\rho}$ \blue{and, \may{because $C$ is leaf-addable}, a corresponding $C$-}orientation of $N$.}



\blue{Now} suppose that~$s_\rho$ contains at least $\ell$ leaves of $N$. We \blue{consider} two cases \blue{depending on whether or not $e_{\rho}$ is a pendant edge of $N$}. \blue{If} $e_\rho$ is pendant, \blue{then} $e_\rho$ is incident to a leaf, say~$c_j$, \blue{of $N$}. \remiee{By Property~(ii) of \blue{Definition~\ref{def:l-chain-reducible}}, $N^\ell$ can be $C$-rooted at $c_i'$ \blue{for some $i$ satisfying} $j\in\{i,i+1\ldots ,n_s-(\ell-i)\}$. \blue{Since Algorithm~\ref{alg:binary_blob_orientation} \may{finds all $C$-rooted edges of $N^{\ell}$ with corresponding orientations}}, the algorithm will \blue{establish that $N^{\ell}$ can be $C$-rooted at $c'_i$ and also} find a \blue{corresponding $C$-}orientation \blue{of $N^{\ell}$}. \blue{It follows that Algorithm~\ref{alg:binary_blob_orientation_TC_SF} correctly finds that} $e_{\rho}$ \blue{is a $C$-rooted edge of $N$ and, it is easily checked}, as $C$ is leaf-addable, \blue{a corresponding $C$-orientation of $N$ in \red{Line~18}}.} \blue{If $e_{\rho}$ is not a pendant edge of $N$, then $e_{\rho}$ is incident to a vertex, say $u_j$ which is adjacent to $c_j$, of $N$. By Property~(ii) of Definition~\ref{def:l-chain-reducible}, $N^{\ell}$ can be $C$-rooted at $c'_i$ for some $i$ satisfying $j\in\{i, i+1, \ldots, n_s-(\ell-i)\}$. Thus, as Algorithm~\ref{alg:binary_blob_orientation} \may{finds all $C$-rooted edges of $N^{\ell}$ with corresponding orientations}, the algorithm establishes that $c'_i$ is a $C$-rooted edge of $N^{\ell}$ and also finds a corresponding $C$-orientation. \red{Therefore, by the argument in the first paragraph of the proof,} Algorithm~\ref{alg:binary_blob_orientation_TC_SF} correctly finds that $e_p$ is a $C$-rooted edge of $N$ and, it is easily checked, as $C$ is blob-determined and leaf-addable, it finds a corresponding $C$-orientation of $N$ in \red{Line~25}.} \blue{Hence} Algorithm~\ref{alg:binary_blob_orientation_TC_SF} correctly finds all $C$-rooted edges \blue{of} $N$ \blue{as well as a corresponding $C$-}orientation of $N$.

\may{Note that all $C$-rooted edges are indeed found in Lines~7--30, so the case of Line~31 can indeed be ignored in the algorithm.}

For the running time, note that \blue{Algorithm~\ref{alg:binary_blob_orientation_TC_SF} consists of} three separate \blue{parts}: the \blue{$\ell$-chain} reduction \blue{on $N$} to \blue{get} $N^\ell$ by \blue{deleting} leaves (Line~1); the application of Algorithm~\ref{alg:binary_blob_orientation} to find the $C$-rooted edges \blue{of} $N^\ell$ \blue{and a corresponding $C$ orientation for each such edge} (Line~2); and the inference of the $C$-rooted edges \blue{of} $N$ \blue{as well as the corresponding $C$-orientations} (Lines~3--\red{34}). It is clear that the reduction \blue{in Line~1} can be executed in $O(n^2)$ time.


Next we turn to the running time of applying Algorithm~\ref{alg:binary_blob_orientation} to~$N^\ell$. As each side of the generator of~$N^\ell$ contains at most~$\ell$ leaves, the number of vertices and edges of~$N^\ell$ are bounded by a function of~$k$ and~$\ell$. This makes the running time of Algorithm~\ref{alg:binary_blob_orientation} a function of~$\ell$ and~$k$. To be more concrete, first observe that,  \revL{since $G(N)$ is cubic, $3|V(G(N))| = 2|E(G(N))|$. Combining this with $k=|E(G(N))| - |V(G(N))| + 1$ \latest{(which follows from the definition of the reticulation number)} gives $|E(G(N))|=3(k-1)$ and $|V(G(N))|=2(k-1)$.} 
Hence, $|V(N^{\blue{\ell}})|\leq 2(k-1) + 6\ell(k-1) \leq 8\ell(k-1)$ and so, by Lemma~\ref{lem:binary_blob_orientation},
the running time \blue{of the second part} is
$$O((8\ell(k-1))^{k+1}(\ell(k-1)+f_C(8\ell(k-1),k))).$$

For the last part, \latest{$G(N)$ can be found in $O(n)$ time by deleting all leaves and suppressing their neighbours.} \blue{As $G(N)$ and $G(N^{\ell})$ are isomorphic}, each side of \blue{$G(N^{\ell})$} has at most \blue{$\ell+1$} edges, \blue{and so} $N^\ell$ has at most $3(k-1)(2\ell+1)$ edges. For each $C$-rooted edge \blue{of} $N^\ell$, we modify a \blue{$C$-}orientation of~$N^\ell$ at most~$2n$ times, each time taking $O(n)$ time. Hence, the running time of this part is $O(\ell(k-1)n^2)$.
\blue{Taken altogether}, the \blue{total} running time of Algorithm~\ref{alg:binary_blob_orientation_TC_SF} is 
$$O((8\ell(k-1))^{k+1}(\ell(k-1)+f_C(8\ell(k-1),k)) + \ell(k-1)n^2).$$
\end{proof}

The next theorem is an immediate consequence of Lemma~\ref{lem:fptalgretic}.

\begin{theorem}\label{thm:fpt-retic}
Let $C$ be an $\ell$-chain reducible, \leooo{leaf-addable, blob-determined} class of directed \blue{binary phylogenetic} networks. \blue{If $f_C(8\ell(k-1), k)$ \latest{(described in Lemma~\ref{lem:fptalgretic})} is a computable function}, then {\sc $C$-orientation} is FPT with the reticulation number of the undirected binary \blue{phylogenetic} network as the parameter.
\end{theorem}

\leo{\blue{In Section~\ref{sec:PuttingBlobsTogether}, we extend Algorithm~\ref{alg:binary_blob_orientation_TC_SF} to an FPT} algorithm \blue{for $C$-orientation, where} the level of \blue{$N$} is the parameter. Before doing this, we conclude this subsection with the following sufficient condition for a network to be $C$-orientable.}

\begin{proposition}
\blue{Let $N$ be an} \leooo{undirected} \blue{binary phylogenetic} network with at least~$\ell$ leaves on each side of \blue{$G(N)$}, and let $C$ be an $\ell$-chain reducible, leaf-addable class of directed \blue{binary phylogenetic} networks. If, by adding leaves, \blue{there is an} undirected \blue{binary phylogenetic} network \blue{that is} $C$-orientable, then \blue{$N$} is $C$-orientable.
\end{proposition}

\begin{proof}
Let $N$ be an arbitrary undirected \blue{binary phylogenetic} network with \leooo{at least}~$\ell$ leaves on each side of its generator. Suppose we can add leaves to~$N$ to obtain an undirected \blue{binary phylogenetic} network $N'$ that is $C$-orientable. Since $C$ is $\ell$-chain reducible, \blue{it follows by Properties~(ii) and~(iii) of Definition~\ref{def:l-chain-reducible} that applying an $\ell$-chain reduction to} $N'$ \blue{gives} a \blue{directed binary phylogenetic} network $N^{\ell}$ that is $C$-orientable. \leooo{Since $N$ can \blue{be} obtained \blue{from $N^{\ell}$} by adding leaves \blue{and} $C$ \blue{is} leaf-addable, $N$ is $C$-orientable.}
\end{proof}


\subsection{FPT algorithm parameterized by \blue{the} level}\label{sec:PuttingBlobsTogether}

Using Algorithms~\ref{alg:binary_blob_orientation} and~\ref{alg:binary_blob_orientation_TC_SF}, in this section we establish an FPT algorithm for {\sc $C$-orientation}, where the level of the undirected \blue{binary phylogenetic} network $N$ is the parameter. The main idea is to \red{orient} each blob \blue{of $N$} and to combine these orientations \blue{into} an orientation of \blue{$N$}. For this second step, we first need the following definitions.

Let $N$ be an undirected \blue{binary phylogenetic} network and let $B$ be a blob of $N$. The \blue{undirected binary phylogenetic} network \emph{induced} by $B$ is obtained from $B$ by \blue{adjoining to} each degree-$2$ vertex $u$ a \blue{new} leaf~$x$ and a \blue{new} edge~$\{u, x\}$. Furthermore, \blue{for} a blob-determined class $C$ of directed \blue{binary phylogenetic} networks, \blue{if $N$ is a member of $C$}, we say that~$B$ can be $C$-\emph{rooted} at a cut-edge~$e=\{u,v\}$ of~$N$ with~$u\in B$ and~$v\notin B$ if the \blue{undirected binary phylogenetic} network induced by~$B$ can be $C$-rooted at the \charles{pendant} \blue{edge} incident to~$u$.

\blue{Allowing for} bi-directed \blue{edges}, let~$N^o_C$ be the \blue{mixed} graph obtained from $N$ by directing each cut-edge~$e$ \blue{of $N$} incident to a blob~$B$ away from $B$ if~$B$ cannot be $C$-rooted at~$e$. \blue{Note that} if a cut-edge $e$ \charles{joins} two blobs \blue{of $N$} and neither \blue{blob} can be $C$-rooted at $e$, then this cut-edge becomes bi-directed. Define $T_C(N)$ to be \may{obtained} from $N_C^o$ by contracting \blue{every} undirected edge \blue{of $N^o_C$}. \latest{Note that \may{(the underlying graph of)}~$T_C(N)$ \may{is a tree} as all edges in the blobs of~$N$ are undirected in~$N^o_C$ and therefore contracted \may{(and a graph without blobs is a tree)}. Also note that~$T_C(N)$ is not a phylogenetic tree, but a tree in the usual graph-theoretic sense and that all its edges are directed or bidirected.}

\blue{Let $C$ be an $\ell$-chain reducible, leaf-addable, blob-determined class of directed binary phylogenetic networks.} The FPT algorithm  for {\sc $C$-orientation} with the level of \blue{$N$} as the parameter is described \blue{as} Algorithm~\ref{alg:C_orientation}. The main idea behind the algorithm is captured by the following proposition. A \emph{rooted tree} is a directed tree with a single vertex \blue{of} in-degree $0$, called the \emph{root}, in which all arcs are directed away from the root. Note that a rooted tree may consist of a single vertex.

\begin{proposition}\label{prop:RootBlobDetermined}
Let $C$ be a blob-determined class of directed \blue{binary phylogenetic} networks, and let $N$ be an undirected \blue{binary phylogenetic} network. Then $N$ has a $C$-orientation if and only if, \blue{for} each blob \blue{$B$} of $N$, \blue{the undirected binary phylogenetic} network induced by \blue{$B$} has a $C$-orientation, and $T_C(N)$ is a rooted tree.
\end{proposition}

\begin{proof}
First assume that $N$ has a $C$-orientation \blue{$N'$}, \blue{and let $B$ be a blob of $N$}. Since $C$ is blob-determined, \blue{$N'$} induces a $C$-orientation of \blue{the undirected binary phylogenetic} network \blue{$N_B$} induced by \blue{$B$}. \blue{Now let} $\{u, v\}$ be a cut-edge of $N$, where $u$ is a vertex of $B$. \blue{If} $\{u, v\}$ is directed away from $B$ in \blue{$N^o_C$}, then~$B$ cannot be $C$-rooted at $\{u, v\}$, \blue{and so, as $N'$ is a $C$-orientation of $N$, it follows that $\{u, v\}$ is not directed towards $u$ in $N'$}. \blue{Thus $\{u, v\}$ is directed away from $u$ in $N'$, that is},
$\{u, v\}$ is directed away from $B$ in \blue{$N'$}. \blue{Therefore if an edge is orientated in $N^o_C$, then the orientation of that edge is in agreement with its orientation in $N'$. (In particular, it follows that no edge is bi-directed in $N^o_C$.)}
Therefore, by contracting the \blue{arcs} of \blue{$N'$} \blue{for which the corresponding edges of} $N^o_C$ have no orientation, we obtain $T_C(N)$. \blue{Since} $T_C(N)$ is obtained from a directed \blue{binary phylogenetic} network by contracting arcs, and $T_C(N)$ is a tree, it \blue{follows that $T_C(N)$ is} a rooted tree.

To prove the converse, assume that the \blue{undirected binary phylogenetic} network induced by each blob of~$N$ has a $C$-orientation and that $T_C(N)$ forms a rooted tree. Let $K$ be the subgraph of $N$ that \blue{contracts} to the root of $T_C(N)$. Then either \blue{(i)} $K$ consists of \blue{a single} blob \blue{$B$ of $N$}, or \blue{(ii)} \blue{$K$} contains at least one cut-edge \blue{of $N$}. \blue{Depending on whether (i) or (ii) holds}, we \blue{next show} that there exists a $C$-orientation of $N$ where the root is located either \blue{on an edge of} $B$, or on a cut edge $e$ of $K$.

\blue{If (i) holds, then, as $T_C(N)$ is a rooted tree (obtained from $N_C^o$ by contracting all undirected edges), all of} the cut-edges \blue{of $N$} incident to \blue{a vertex of} $B$ are oriented away from $B$ in $N_C^o$. \blue{Therefore, as} the \blue{undirected binary phylogenetic} network $N_B$ induced by $B$ is $C$-orientable, there exists an edge $e_\rho$ of $B$ at which $N_B$ can be $C$-rooted. We \blue{now} find a \blue{$C$-}orientation of $N$ as follows. Subdivide $e_\rho$ by inserting the root, and orient the edges in $B$ the same way as \blue{they are orientated} in $N_B$. \blue{Orienting} all cut-edges \blue{of $N$} away from the root, each blob $B'\neq B$ of $N$ now has \blue{exactly one} incoming cut-arc, \blue{say} $(u,v)$. Since $T_C(N)$ is a rooted tree, the \blue{undirected binary phylogenetic} network induced by $B'$ can be $C$-rooted at the cut-edge incident to $v$. \red{Orienting} the edges of $B'$ \blue{(and all other such blobs of $N$)} accordingly, gives a $C$-orientation of~$N$. \blue{For (ii)}, we subdivide \blue{the cut-edge} $e$ \blue{of $K$} by the root and proceed in the same way as for \blue{(i)}\may{, starting by orienting all
cut-arcs away from the root.}
\end{proof}

\begin{algorithm}[t]
\KwIn{An undirected binary \blue{phylogenetic} network $N$ \blue{with no non-trivial pendant \red{phylogenetic} subtrees}. \deletee{and a $\ell$-chain reducible, leaf-addable blob-determined class~$C$.}}
\KwOut{A $C$-orientation of $N$ if it exists, and NO otherwise.}
\blue{Find the set of blobs of $N$}\;
\For{each blob $B$ of $N$}{
\latest{Apply Algorithm~\ref{alg:binary_blob_orientation_TC_SF} to the \blue{undirected binary phylogenetic} network~$N_B$ induced by~$B$ and let $L_B$ be the returned set of pairs $(e,B_e)$ consisting of $C$-rooted edge~$e$ and corresponding orientation~$B_e$ of~$N_B$}\;

    \If{$L_B=\emptyset$}{
        \Return NO\;
    }
}
\blue{Construct} $N^o_C$ \latest{from~$N$} by orienting
\latest{each cut-edge~$e$} of $N$ incident to a vertex \latest{of a blob} $B$ away from $B$ \blue{if} \latest{there is no pair in~$L_B$ with~$e$ as first element (possibly orienting edges in two directions)\;}
\blue{Construct} $T_C(N)$ \latest{from $N^o_C$} by contracting all non-oriented edges in $N^o_C$\label{line:contract}\;
\uIf{$T_C(N)$ is a rooted tree}{
    \remie{Determine the subgraph $K$ of $N$ that \latest{is contracted, in Line~\ref{line:contract},} to the root of $T_C(N)$\;}
    \If{$K$ consists of \blue{a single} blob $B$ \blue{of $N$}}{
        \remie{Pick an arbitrary \latest{element $(e,B_e)\in L_B$} and orient~$B$ \latest{in~$N$ according to~$B_e$}, calling the root vertex~$\rho$\;}
    }
   \If{$K$ contains a cut-edge}{
       Subdivide \latest{an arbitrary cut-edge}~$e$ by the root~$\rho$\;
  }
    Orient all cut-edges of \blue{$N$} away from~$\rho$\;
    \For{each unoriented blob~$B'$ \latest{of~$N$}}{
    Find the cut-arc~$(u,v)$ entering~$B'$\;
    Let~$\{v,x\}$ be the cut-edge incident to~$v$ in the network induced by~$B'$\;
    Find a pair $(\{v,x\},\latest{B'_{\{v,x\}}})\in L_{B'}$\;
    Orient the edges of~$B'$ \latest{in $N$ as in~$B'_{{\{v,x\}}}$}\;
    }
    \Return the oriented network~\latest{$N$}\;
}
\Else{
    \Return NO\;
}
\caption{\blue{An} FPT algorithm for $C$-\textsc{orientation} with the level of $N$ as the parameter, where $C$ is an $\ell$-chain reducible, leaf-addable, and blob-determined class of directed \blue{binary phylogenetic} networks.
\label{alg:C_orientation}
}
\end{algorithm}


\blue{The correctness of Algorithm~\ref{alg:C_orientation} and its running time is established in the next lemma.}


\begin{lemma}
Let $N=(V, E, X)$ be an undirected \blue{binary phylogenetic} network. Then Algorithm~\ref{alg:C_orientation} applied to $N$ is correct and runs in time
\remie{$O( g(L,\ell)n+\ell(L-1)n^3))$} \blue{if $C$ is an} $\ell$-chain reducible, leaf-addable, and blob-determined class of directed \blue{binary phylogenetic} networks,
where $n=|V|$, $L$ is the level of $N$,
and $g$ is \blue{a} function of~$L$ and~$\ell$ independent of $n$.
\label{lem:fptalglevel}
\end{lemma}

\begin{proof}
\blue{The} correctness of \blue{Algorithm~\ref{alg:C_orientation} is essentially given in the proof of} Proposition~\ref{prop:RootBlobDetermined}, \blue{and so it is omitted}.
\remie{For the running time, first note that \revL{all blobs can be found in $O(n^3)$ time by checking for each edge whether it is a cut-edge. The rest of the} algorithm consists of two parts. The first part \blue{consists} of \blue{finding} all \blue{$C$-rooted edges} of the \blue{undirected binary phylogenetic} networks induced by the blobs \blue{of $N$} and a \blue{corresponding} \blue{$C$-}orientation \blue{for each such edge} (Lines~1--8), \blue{while} the second part consists of \blue{constructing} $N^o_C$ \blue{and $T_C(N)$} and, \blue{provided $T_C(N)$ is a rooted tree}, finding \blue{an} orientation of the cut-edges and blob edges of~$N$ (Lines~9--29).} 


\blue{By Lemma~\ref{lem:fptalgretic}}, for a blob $B$ with $n_B$ vertices and $k_B$ reticulations, running Algorithm~\ref{alg:binary_blob_orientation_TC_SF} on \blue{the undirected binary phylogenetic network induced by} $B$ takes $O(g(k_B, \ell)+\ell(k_B-1)n^2_B)$ time. Since \blue{$N$ has} at most $n$ blobs \latest{and~$k_B\leq L$ by the definition of level}, the first part of Algorithm~\ref{alg:C_orientation} runs in time
$$O(g(L, \ell)n+\ell(L-1)n^3).$$

\remie{\blue{For} the second part of Algorithm~\ref{alg:C_orientation}, we \blue{initially construct $N^o_C$ and} $T_C(N)$. \blue{Orientating the cut-edges of $N$ incident to blob vertices} to obtain $N^o_C$ \blue{and then} contracting the \blue{unorientated} edges of $N_C^o$ to obtain $T_C(N)$ takes $O(n^2)$ time.} \remie{\blue{Once this is completed}, the second part of the algorithm requires only one pass through \blue{$N$} to orient \blue{its} edges, as we may independently pick an orientation for each blob $B$ from the set of orientations $L_B$ with the correct root-edge (finding such orientation in the set may take $O(n)$ time). Hence, the second part of the algorithm only takes $O(n^2)$ time.} \blue{This completes the proof of the lemma.}
\end{proof}

\noindent {\bf Remark.} \blue{If $C$ is not necessarily $\ell$-chain reducible and leaf-addable, but is a blob-determined class of directed binary phylogenetic networks, then we can adapt Algorithm~\ref{alg:C_orientation} by replacing Line~3 with the following to obtain an algorithm for deciding if an undirected binary phylogenetic network $N$ has a $C$-orientation:} \\
Let $L_B$ be the output of Algorithm~\ref{alg:binary_blob_orientation} applied to the undirected binary phylogenetic network induced by $B$. \\
\blue{Taking the same approach as the proof of Lemma~\ref{lem:fptalglevel}, the running time of this adaption is} $O(\binom{n}{L} n^2(n+f_C(n, L)))$, where $n$ is the number of vertices of $N$, and $L$ is the level of $N$.

The next theorem is an immediate consequence of Lemma~\ref{lem:fptalglevel}.

\begin{theorem}
Let $C$ be an $\ell$-chain reducible, leaf-addable, blob-determined class of directed \blue{binary phylogenetic} networks, \may{for any fixed~$\ell$}. \blue{If $g(L, \ell)$ is a computable function}, then \blue{Algorithm~\ref{alg:C_orientation}} is an FPT algorithm for {\sc $C$-Orientation}, where the level~\may{$L$} of the \blue{inputted} undirected binary \blue{phylogenetic} network is the parameter.
\label{thm:fpt}
\end{theorem}


\section{Specific classes}
\label{sec:Classes}

\shortened{A directed phylogenetic network is said to be \emph{tree-child} if each non-leaf vertex has a child that is a tree vertex.} The main result of this section \blue{is the following} theorem \blue{which establishes that {\sc $C$-Orientation}} is FPT \shortened{when~$C$ is the class of binary tree-child networks. The same technique can be applied to many other known classes, see Appendix~\ref{app}.}


\begin{theorem}\label{the:combining_tc}
\shortened{Algorithm~\ref{alg:C_orientation} is an FPT algorithm for deciding whether an undirected binary phylogenetic network~$N$ has a tree-child orientation, where the level of~$N$ is the parameter.}
\end{theorem}

The proof of Theorem~\ref{the:combining} \blue{relies on} combining Theorem~\ref{thm:fpt} with Lemmas~\ref{lem:tcsfblob}--\ref{visiblelast}. These lemmas \blue{collectively} show that \blue{each} of the \blue{directed binary phylogenetic} classes \blue{in the statement of Theorem~\ref{the:combining}} is \blue{$\ell$}-chain reducible, leaf-addable, and blob-determined. \blue{To establish the $\ell$-chain reducible property for each of these classes, we will show that each such class satisfies a variant of this property. This variant, called rooted $\ell$-chain reducible, is described in Section~\ref{rootedchain}, where we also show that if a directed binary phylogenetic class $C$ is rooted $\ell$-chain reducible}, leaf-addable, and blob-determined, then $C$ is $\ell$-chain reducible. \blue{Sections~\ref{tree-child} \may{through}~\ref{reticulation-visible} contain the statements of Lemmas~\ref{lem:tcsfblob}--\ref{visiblelast} and their proofs.} \blue{Lastly, note that we could have used Algorithm~\ref{alg:binary_blob_orientation_TC_SF} instead of Algorithm~\ref{alg:C_orientation} \marknew{(with the reticulation number of~$N$ as the parameter)} in the statement of Theorem~\ref{the:combining}.}

\subsection{Rooted $\ell$-chain reduction}
\label{rootedchain}

\blue{We begin by defining the operation of rooted $\ell$-chain reduction.} \latest{Note that this operation is defined on undirected binary phylogenetic networks, but with a specified pendant edge~$e_{\rho}$ which will be used as the root location. We also remark that we will use the term ``rooted $\ell$-chain reduction'' to refer to a network obtained by this operation, as well as to refer to the operation \red{itself}. Recall that we assume throughout Sections~\ref{sec:fptalgs} and~\ref{sec:Classes} that networks have no nontrivial pendant \red{phylogenetic} subtrees and that they have reticulation number at least~$2$.}


\begin{definition}
Let $N$ be an undirected \blue{binary phylogenetic} network, let \blue{$\ell$} be a \blue{non-negative} integer, and let $e_{\rho}$ be a pendant edge of $N$. Furthermore, let $s_{\rho}\blue{=\{u, v\}}$ be the side of $G(N)$ containing $e_{\rho}$, \red{and let $P_{s_\rho}$ denote the undirected path of $N$ corresponding to $s_{\rho}$ \may{between} $u$ \may{and} $v$. We call} \latest{an} undirected \blue{binary phylogenetic} network obtained from $N$ by applying the following three operations a {\em rooted \blue{$\ell$}-chain reduction \red{(from $u$, with respect to~$e_\rho$)}} \blue{on} $N$:
\begin{enumerate}[label=\textup{(\roman*)}]
\item for each side~$s$ of~$G(N)$ other than~$s_{\rho}$ that contains at least \blue{$\ell$} leaves, delete \blue{$n_s-\ell$} leaves on $s$, where \blue{$n_s$} is the number of leaves \red{on} \blue{$s$};
\item delete all leaves on $s_{\rho}$ that are adjacent to an internal vertex of \blue{$P_{s_{\rho}}$} between $u$ and \blue{the end vertex of} $e_\rho$ \blue{on \red{$P_{s_\rho}$}}; and
\item if there are at least \blue{$\ell$} leaves adjacent to an internal vertex of \blue{$P_{s_{\rho}}$} between \blue{the end vertex of} $e_{\rho}$ \blue{on $P_{s_{\rho}}$} and $v$, then delete all but $\blue{\ell}-1$ of these leaves; \latest{otherwise, if there are at most $\ell-1$ such leaves, do nothing}.
\end{enumerate}
\end{definition}



\begin{definition}
\latest{\blue{A class $C$ of directed binary phylogenetic networks} is {\em \blue{rooted $\ell$}-chain reducible} if the following property holds: \red{Let $N$ be} an undirected binary phylogenetic network
and \red{let} $e_{\rho}$ be a pendant edge of $N$. \red{Let $s_\rho=\{u, v\}$ be the side of $G(N)$ containing $e_\rho$ and let $N_u$ and $N_v$
be rooted \may{$\ell$}-chain reductions from $u$ and $v$, respectively, with respect to $e_\rho$.} \red{If} $N$ can be $C$-rooted \red{at $e_\rho$}, \blue{then} \red{at least one of $N_u$ and $N_v$} can be $C$-rooted at $e_{\rho}$.}
\end{definition}


An example \blue{illustrating} these definitions is given in Figure~\ref{fig:ChainReductionOneSide}, \blue{where} $C$ \blue{is} the class of \blue{binary} tree-child networks. \blue{Note that, we will eventually} show that \blue{the class of binary tree-child networks is} \remiee{rooted $3$-chain reducible} (Lemma~\ref{lem:3chainreducibleTCSF}).

\begin{figure}
\centering
\includegraphics[width=.7\textwidth]{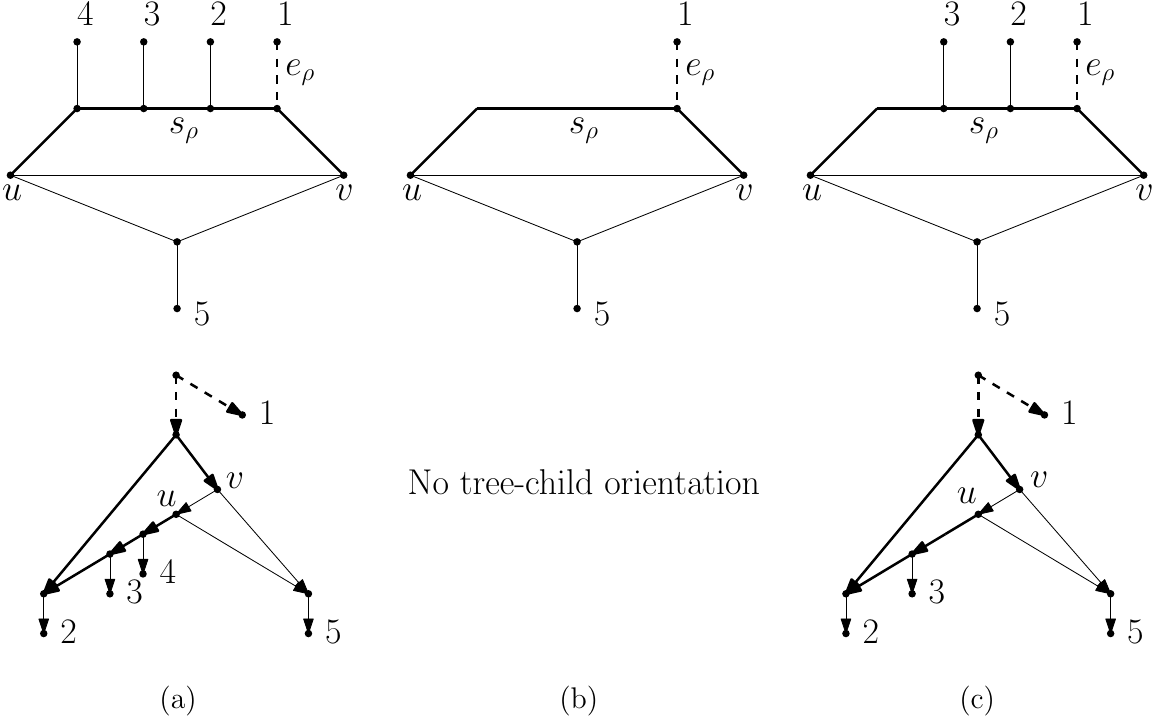}
\caption{\label{fig:ChainReductionOneSide} \red{An} example of rooted $3$-chain reduction\red{s}. Subfigure (a) shows an undirected binary phylogenetic network $N$ that has a tree-child orientation rooted at edge $e_{\rho}$ as shown below it. The side of the generator $G(N)$ that contains the root is denoted~$s_\rho=\{u,v\}$. As the class of tree-child networks is rooted $3$-chain reducible, \red{with respect to $e_\rho$}, a rooted $3$-chain reduction \red{on} $N$ \red{from at least one of~$u$ and~$v$ results} in an undirected binary phylogenetic network that can be tree-child rooted at $e_{\rho}$. Subfigure~(b) \red{shows a} rooted $3$-chain reduction on $N$ \red{from $u$, but, as indicated in (b), it cannot be tree-child rooted at $e_\rho$}.
However, \red{as shown in Subfigure~(c), a rooted $3$-chain reduction from $v$ results in an undirected binary phylogenetic network that can be tree-child rooted at $e_{\rho}$.}}
\end{figure}

\blue{The next two lemmas will be used to show that if a class of directed binary phylogenetic networks is rooted $\ell$-chain reducible, leaf-addable, and blob-determined, then it is $\ell$-chain reducible.}

\begin{lemma}\label{lem:CombinedRootInferenceSimple}
Let $C$ be a \remiee{rooted} \blue{$\ell$}-chain reducible, leaf-addable class of directed \blue{binary phylogenetic} networks, and let $N$ be an undirected binary phylogenetic network that is $C$-orientable. \latest{Let $N'$ be the undirected binary phylogenetic network obtained from $N$ by an $\ell$-chain reduction on~$N$}. \remiee{\blue{Suppose that} $s=\{u, v\}$ \blue{is} a side of \blue{$G(N)$} \leooo{that contains at least \blue{$\ell$} leaves}}. Let $c_1, c_2, \ldots, c_{\blue{n_s}}$ denote the leaves \blue{of $N$ on $s$} ordered from $u$ to $v$, and let $c'_1, c'_2, \ldots, c'_{\blue{\ell}}$ denote the leaves of $N'$ \leooo{on} $s$ ordered \blue{from $u$ to $v$}. \blue{Then each of the following hold:}
\begin{enumerate}[label=\textup{(\roman*)}]
\item If $i\in \{1, 2, \ldots, \blue{\ell}\}$ \blue{and} $N'$ can be $C$-rooted at $c'_i$, then $N$ can be $C$-rooted at $c_j$ for all $j\in \{i, \blue{i+1}, \ldots, n_{\blue{s}}-(\blue{\ell}-i)\}$.
\item If $j\in \{1, 2, \ldots, \blue{n_s}\}$ \blue{and} $N$ \leooo{can be $C$-rooted at $c_j$}, then $N'$ \leooo{can be $C$-rooted at} $c'_i$ for some $i$ \blue{satisfying} $j\in\{i, \blue{i+1}, \ldots, n_{\blue{s}}-(\blue{\ell}-i)\}$.
\end{enumerate}
\end{lemma}

\begin{proof}
\blue{For (i)}, suppose that $N'$ can be $C$-rooted at $c'_i$, \blue{where $i\in \{1, 2, \ldots, \ell\}$}, and \blue{let $N^{\ell}_i$ be a $C$-orientation of $N'$ rooted at $c'_i$}. \leooo{Let $j\in \{i, i+1, \ldots, n_{\blue{s}}-(\blue{\ell}-i)\}$.} \blue{Now construct an orientation $N_j$ of $N$ from $N^{\ell}_i$ as follows. First,} add \blue{back} $j-i$ leaves \blue{on the directed path from} the neighbour \blue{$u'_i$} of $c'_i$ to $u$ and \blue{add back $(n_s-j)-(\ell-i)$} 
leaves \blue{on the directed path from $u'_i$ to} $v$ \blue{relabelling the leaves ordered from $u'_i$ to $u$ and $u'_i$ to $v$ as $c_{j-1}, c_{j-2}, \ldots, c_1$ and $c_{j+1}, c_{j+2}, \ldots, c_{n_s}$, respectively, and relabelling the leaf adjacent to $u'_i$ as $c_j$. \latest{Note that, as $j\leq n_s-(\ell-i)$, it follows that $(n_s-j)-(\ell-i)\ge 0$.} Now extend the resulting orientation by adding back the remaining leaves deleted in the reduction \latest{at their original location}.} \blue{This gives $N_j$, an orientation of $N$ rooted at $c_j$.} Since $C$ is leaf-addable \blue{and $N^{\ell}_i$ is a $C$-orientation}, \blue{it follows that $N_j$ is a $C$-orientation of $N$} rooted at $c_j$. \blue{This establishes (i).}

\blue{To prove (ii)}, suppose that $N$ can be $C$-rooted at $c_j$, \blue{where $j\in \{1, 2, \ldots, n_s\}$}. Since $C$ is \remiee{rooted} \blue{$\ell$}-chain reducible, \blue{there is a rooted $\ell$-chain reduction} $N''$ \blue{on $N$ with respect to the edge incident with $c_j$} that can be rooted at $c_j$. \blue{Without loss of generality}, we may assume that \blue{in this reduction we} deleted all \blue{the} leaves on~$s$ between \blue{the neighbour $u_j$ of $c_j$} and $u$, and if there are at least \blue{$\ell$} leaves \blue{on $s$} between \blue{$u_j$} and $v$, we deleted all but $\blue{\ell}-1$ of these leaves.


First assume that $j\geq n_{\blue{s}} - (\blue{\ell}-1)$, and let $j=n_{\blue{s}}-t$, where $t\le \blue{\ell}-1$. \blue{In this case}, no leaves \blue{of $N$ on $s$} are deleted between \blue{$u_j$} and $v$ to obtain $N''$.
\blue{Thus} $N''$ \blue{has} exactly $t+1$ leaves on $s$ and \may{(by definition of $c_j$ and of rooted $\ell$-chain reduction)} $N''$ can be $C$-rooted at $c_j$ the \blue{first} leaf on $s$ \blue{ordered from $u$ to $v$}. \blue{Let $N^{\ell}_j$ denote a $C$-orientation of $N''$ rooted at $c_j$. We next construct an orientation $N^{\ell}_i$ of $N'$ from $N^{\ell}_j$ as follows.} Add \blue{back} $\blue{\ell}-(t+1)$ leaves \blue{on the directed path from $u_j$ to} $u$, so that we have exactly \blue{$\ell$} leaves on $s$, and relabel the leaves \blue{ordered from $u_j$ to $u$ and from $u_j$ to $v$ as $c'_{i-1}, c'_{i-2}, \ldots, c'_1$ and $c'_{i+1}, c'_{i+2}, \ldots, c'_{\ell}$, respectively, and relabel $c_j$ as $c'_i$.} \blue{This gives $N^{\ell}_i$. Since $G(N^{\ell}_i)$ \latest{is isomorphic to} $G(N')$ and each side $s$ of $G(N^{\ell}_i)$ and $G(N')$ contains the same number of leaves, it follows that}, up to relabelling \blue{the leaves on each side, $N^{\ell}_i$ is an orientation of $N'$.} \blue{Thus, as $C$ is leaf-addable, $N'$ has a $C$-orientation} rooted at $c'_i$, \blue{where} $i = \blue{\ell} - t = \blue{\ell} - (n_{\blue{s}}-j)$. Since $n_{\blue{s}} -(\blue{\ell} - i) = j$, we have $j\in \{i, \blue{i+1}, \dots, n_s - (\ell - i)\}$ \latest{as required}.

Now assume that $j < n_{\blue{s}} - (\blue{\ell} - 1)$. \blue{Then}, \may{by applying the rooted $\ell$-chain reduction, we delete all leaves \blue{of $N$} on $s$ between $u$ and \blue{$u_j$} while keeping the network $C$-rootable at the leaf-edge incident to $u_j$. Hence,} we have that \leooo{$N''$ can be $C$-rooted at the \blue{first} leaf on $s$ \blue{ordered from $u$ to $v$}. Moreover, as $ j < n_{\blue{s}} - \latest{(\blue{\ell} - 1)}$, side~$s$ \blue{of} $N''$ contains exactly \blue{$\ell$} leaves. Therefore, \blue{as $G(N')$ \latest{is isomorphic to} $G(N'')$}, up to relabelling \blue{the leaves on each side}, $N'$ is \blue{isomorphic} to $N''$. Hence, $N'$ can be $C$-rooted at $c'_1$. Since $j\geq i$ and $n_{\blue{s}}-(\blue{\ell}-i) \may{\geq} n_{\blue{s}}-(\blue{\ell}-1)>j$ \latest{as $i\ge 1$}, we have $j \in \{i, \may{i+1}, \dots, n_{\blue{s}} - (\blue{\ell} - i)\}$, \latest{again, as required}.} \blue{This completes the proof of (ii) and the lemma.}
\end{proof}

The next lemma \blue{is the non-pendant edge analogue of Lemma~\ref{lem:CombinedRootInferenceSimple}}.


\begin{lemma}\label{lem:CombinedRootInferenceOtherEdges}
Let $C$ be a \remiee{rooted} \blue{$\ell$}-chain reducible, leaf-addable, blob-determined class of directed \blue{binary phylogenetic} networks, and let $N$ be an undirected \blue{binary phylogenetic} network that is $C$-orientable. \latest{Let $N'$ be the undirected binary phylogenetic network obtained from $N$ by an $\ell$-chain reduction on $N$.} Suppose that $s=\{u, v\}$ is a side of \blue{$G(N)$} that contains at least \blue{$\ell$} leaves. Let $e_0, e_1, \ldots, e_{\latest{n_s}}$ denote the edges \blue{of $N$} on~$s$ ordered from $u$ to $v$, and let $c'_1, c'_2, \ldots, c'_l$ denote the leaves \blue{of} $N'$ on $s$ ordered from $u$ to $v$. \blue{Then each of the following hold:}
\begin{enumerate}[label=\textup{(\roman*)}]
\item If $i\in \{1, 2, \ldots, \blue{\ell}\}$ \blue{and} $N'$ can be $C$-rooted at $c'_i$, then $N$ can be $C$-rooted at $e_j$ for all $j\in \{i-1, i, \ldots, n_{\blue{s}}-(\blue{\ell}-i)\}$.
\item If $j\in \{0, 1, \ldots, n_{\blue{s}}\}$ \blue{and} $N$ can be $C$-rooted at $e_j$, then $N'$ can be $C$-rooted at $c'_i$ for some $i$ \blue{satisfying} $j\in\{i-1, i, \ldots, n_{\blue{s}}-(\blue{\ell}-i)\}$.
\end{enumerate}
\end{lemma}

\begin{proof}
\blue{Let $e_j$ be an edge of $N$ on $\latest{s}$, and} \leooo{let $N_j$ be the \blue{undirected binary phylogenetic} network obtained from $N$ by subdividing $e_j$ \blue{with a new vertex} and \blue{adjoining} a new leaf \blue{to this vertex via} a \blue{new} edge. Let $c_1, c_2, \ldots, c_{n_{\blue{s}}+1}$ denote the leaves \blue{of} $N_j$ on $s$ ordered from $u$ to $v$. \blue{Thus} $c_{j+1}$ is the new leaf. \blue{Note that $G(N)=G(N_j)$.} Since $C$ is blob-determined, $N_j$ can be $C$-rooted at $c_{j+1}$ if and only if $N$ can be $C$-rooted at $e_j$.} \leoooo{Let $N_j'$ be \blue{the undirected binary phylogenetic} network obtained from $N_j$ by \blue{an $\ell$-chain reduction} \latest{on $N_j$}.}

\blue{For the proof of} (i), assume that~$N'$ can be $C$-rooted at $c'_i$, where $i\in \{1, 2, \ldots, \blue{\ell}\}$. \latest{L}et $j\in\{i-1, i, \ldots, n_{\blue{s}}-(\blue{\ell}-i)\}$. Up to relabelling leaves, $N_j'$ is \blue{isomorphic} to $N'$, \blue{and so} $N_j'$ can be $C$-rooted at the $i$-th leaf on $s$ \blue{ordered from $u$ to $v$}. Therefore, as $N_j$ has $n_{\blue{s}}+1$ leaves on side $s$, it follows by Lemma~\ref{lem:CombinedRootInferenceSimple}(i) applied to $N_j$ that~$N_j$ can be $C$-rooted at~$c_{j'}$ for all~$j'\in\{i, i+1, \ldots, \blue{(n_s-(\ell-i))+1}\}$. In particular, as $j+1\in\{i, i+1, \ldots, \blue{(n_s-(\ell-i)\latest{)}+1}\}$, \blue{we have that} $N_j$ can be $C$-rooted at $c_{j+1}$. \blue{Thus} $N$ can be $C$-rooted at $e_j$.

\blue{To prove} (ii), assume that $N$ can be $C$-rooted at $e_j$, where $j\in \{0, 1, \ldots, n_{\blue{s}}\}$. Then $N_j$ can be $C$-rooted at $c_{j+1}$. \latest{To see this, take a $C$-orientation of $N$ rooted at $e_j$, and \red{orient} the edges of $N_j$, except the pendant edge incident with $c_{j+1}$, in the same direction as the corresponding edges of the $C$-orientation of $N$. Now subdivide the edge incident with $c_{j+1}$ by a vertex~$w$ and \red{orient} the two edges incident with~$w$ away from it. The resulting directed binary phylogenetic network is a $C$-orientation of $N_j$ rooted at $c_{j+1}$.}

By Lemma~\ref{lem:CombinedRootInferenceSimple}(ii), $N_j'$ can be $C$-rooted at~$c'_i$ for some $i$ \blue{satisfying} $j+1\in\{i, i+1, \ldots, \blue{(n_s-(\ell-i))+1}\}$. Since $N_j'$ is isomorphic to $N'$ up to relabelling leaves, $N'$ can be $C$-rooted at $c'_i$ for some $i$ \blue{satisfying} $j\in\{i-1, i, \ldots, n_{\blue{s}}-(\blue{\ell}-i)\}$. \blue{This completes the proof of (ii) and the lemma.}
\end{proof}

\blue{A consequence of the last two lemmas is the following proposition.}

\begin{proposition}\label{prop:3conditionslChainRed}
\remiee{Let $C$ be a rooted \blue{$\ell$}-chain reducible, leaf-addable, blob-determined class of directed binary \blue{phylogenetic} networks. Then $C$ is \blue{$\ell$}-chain reducible.}
\end{proposition}

\begin{proof}
\latest{To see that $C$ is $\ell$-chain reducible, observe that} \remiee{properties \blue{(i) and (ii)} of \blue{$\ell$}-chain reducibility follow directly from Lemmas~\blue{\ref{lem:CombinedRootInferenceSimple} and~\ref{lem:CombinedRootInferenceOtherEdges}}, \blue{while} property (iii) \blue{of $\ell$-chain reducibility is a consequence of $C$ being leaf-addable}.}
\end{proof}

\subsection{\shortened{Tree-child} networks}\label{tree-child}

In this section, we establish Theorem~\ref{the:combining_tc}.
Recall that a \blue{directed binary phylogenetic} network \blue{$N$} is tree-child if every non-leaf vertex has a child that is a tree vertex. \blue{Equivalently, $N$} is tree-child if and only if \blue{$N$ has no {\em stack reticulations}, two reticulations one of which is the parent of the other, and no {\em sibling reticulations}, two reticulations sharing a common parent} (see~\cite{sem16}). \blue{This equivalence will be used throughout this subsection.}

\blue{Let $N$ be a directed binary phylogenetic network.} Since adding leaves \blue{to $N$} cannot create \blue{any} stack or \blue{sibling reticulations}, it \blue{follows} that the class of \blue{binary} tree-child networks is leaf-addable. The next two lemmas show that this class is also blob-determined and rooted $3$-chain reducible, and thus, \blue{by Proposition~\ref{prop:3conditionslChainRed}}, $3$-chain reducible.

\begin{lemma}
\shortened{The class of binary tree-child networks} is blob-determined.
\label{lem:tcsfblob}
\end{lemma}

\begin{proof}
\blue{Let $N$ be a} directed \blue{binary phylogenetic} network. If \blue{$N$} is tree-child, then \blue{$N$ has no stack and no sibling reticulations, and so every directed binary phylogenetic} network induced by a blob \blue{of $N$} is tree-child.

Now suppose that all \blue{directed binary phylogenetic} networks induced by a blob of $N$ are tree-child. Then each \blue{such} network \blue{has no stack and no sibling reticulations}. \latest{Observe that} every reticulation of $N$ is contained in a blob of $N$ \latest{and that both parents of a reticulation are in the same blob \may{(because there are paths from the root to each parent)}. If~$N$ contains sibling reticulations, then \may{(for similar reasons)} their common parent must be in the same blob as each of the reticulations, and so the network induced by this blob would also contain sibling reticulations, a contradiction. Similarly,~$N$ cannot contain stacks. Thus} $N$ is tree-child. \blue{This completes the proof of the lemma.}
\end{proof}

\blue{By Lemma~\ref{lem:tcsfblob}}, the class of \blue{binary} tree-child
networks is blob-determined. \blue{Therefore}, as explained in Section~\ref{sec:RootingBlobs}, \blue{in the process of deciding if an} undirected \blue{binary phylogenetic} network $N$ \blue{has a tree-child orientation}, we may assume that $N$ has \blue{no} non-trivial pendant \red{phylogenetic} subtrees. The analogous assumption holds \shortened{for other classes}.

\begin{lemma}\label{lem:3chainreducibleTCSF}
The class of \blue{binary} \shortened{tree-child networks is} \remiee{rooted} $3$-chain reducible.
\end{lemma}

\begin{proof}

Let $N$ be an undirected \blue{binary phylogenetic} network that can be tree-child rooted at a pendant edge $e_\rho\blue{=\{v_{\rho}, x_{\rho}\}}$, \blue{where} $x_{\rho}$ is a leaf, \blue{and let} $N_d$ be a \blue{tree-child orientation of $N$ rooted at $e_{\rho}$}. Note that $v_{\rho}$ is a tree vertex in $N_d$. \latest{Recall that we assume that~$N$, and therefore~$N_d$, has reticulation number at least~$2$, so~$G(N_d)$ is well defined and has two sides $s_1$ and $s_2$} leaving $v_{\rho}$. \blue{We next construct a directed binary phylogenetic network $N'_d$ from $N_d$ as follows}. \blue{First}, for each side $s$ of $G(N_d)$ \blue{that} \blue{is neither $s_1$ nor $s_2$ and} contains at least two leaves, delete all except one of the leaves \blue{of} $N_d$ on~$s$ and suppress the resulting vertices of in-degree one and out-degree one. \blue{At this stage of the construction, it is easily seen that the resulting directed binary phylogenetic} network remains tree-child, as no stack \blue{and no sibling reticulations} \blue{have been} created. \blue{Continuing the construction}, delete all leaves \blue{of} $N_d$ that are on either $s_1$ or $s_2$ \blue{of $G(N_d)$}, \blue{and} suppress the resulting vertices of in-degree one and out-degree one. \blue{This gives $N'_d$.} \blue{Like the first part, the second part of the construction also preserves the property of being tree-child.} To see this, observe that at most one of $s_1$ and $s_2$ \blue{has} a reticulation \blue{as an end-vertex}; otherwise, $N$ has a reticulation cut, contradicting Proposition~\ref{prop:necessary}. Hence, \blue{$N'_d$ has no sibling reticulations}.
Moreover, as the root \blue{of $N_d$} is not a reticulation,
\blue{it follows that $N'_d$ has no stack reticulations}. Hence $N_d'$ is tree-child.

Let $N'$ \blue{denote the} underlying undirected \blue{binary phylogenetic} network \blue{of $N'_d$}, \blue{and} let $s_\rho$ be the side of~$G(N')$ containing \blue{$x_{\rho}$}. We \blue{next show} that each side $s\neq s_\rho$ of $G(N')$ contains at most three leaves \blue{of} $N'$. To see this, $s$ corresponds to at most two sides of $G(N'_d)$; \blue{if $s$ corresponds to exactly} two sides, \blue{then these sides} meet at a reticulation \blue{of $N'_d$} with a leaf as \blue{a} child (see Figure~\ref{fig:SideCorrespondence}). \blue{Thus $s$ contains at most three leaves of $N'$.} \blue{On the other hand}, the side $s_\rho$ \blue{of $G(N')$} corresponds to at most three sides of $G(N'_d)$. Namely, $s_1$, $s_2$, and a third side $s_3$ if \remie{an internal vertex of $s_{\rho}$ \blue{corresponds to} a reticulation $r$ of $N'_d$ (see Figure~\ref{fig:SideCorrespondence}). \blue{If a third side $s_3$ exists, then $r$ is the parent of} a leaf \blue{of $N'_d$}. 
}
Sides $s_1$ and $s_2$ \blue{of $G(N'_d)$} contain no leaves \blue{of $N'_d$} and, if it exists, $s_3$ contains at most one leaf \blue{of $N'_d$}. In addition, the leaf~$x_\rho$ is on the side $s_\rho$ \blue{of $G(N')$} and, if \blue{$s_3$ exists and} $r$ has a child that is a leaf, then this leaf is also on the side $s_\rho$ \blue{of $G(N')$}. Hence the side $s_\rho$ also contains at most three leaves \blue{of $N'$} in total, \blue{where} $x_\rho$ \blue{is either} the first or the last leaf when the leaves \blue{of $N'$} on $s_{\rho}$ are ordered.

Let $N^r$ be \blue{the undirected binary phylogenetic network} obtained from $N$ by applying a \remiee{rooted $3$-chain reduction} with respect to~$e_\rho$.
\blue{Since $G(N^r)=G(N')$, it follows that} $N^r$ can be obtained \blue{from $N'$} by adding leaves and, \blue{if necessary}, relabelling leaves. \blue{Therefore, as} $N'$ can be tree-child rooted at~$e_\rho$ and the class of tree-child networks is leaf-addable,~$N^r$ can be tree-child rooted at~$e_\rho$. \blue{Hence $N$ is $3$-chain reducible with respect to $e_{\rho}$. It now follows that the class of binary tree-child networks is $3$-chain reducible.}
\end{proof}

\begin{figure}
    \centering
    \begin{tikzpicture}
	 \tikzset{edge/.style={thick}}
     \tikzset{arc/.style={-Latex,thick}}
	 \begin{scope}[xshift=-.5cm,yshift=0cm]
    \draw[thick, fill, radius=0.06] (0,0) circle node[above] {$u$};
    \draw[ thick, fill, radius=0.06] (0,-1) circle;
    \draw[ thick, fill, radius=0.06] (0,-2) circle;
    \draw[ thick, fill, radius=0.06] (0,-3) circle;
    \draw[ thick, fill, radius=0.06] (0,-4) circle circle node[below] {$v$};;
    \draw[ thick, fill, radius=0.06] (-.5,-1.5) circle;
    \draw[ thick, fill, radius=0.06] (-.5,-2.5) circle;
    \draw[ thick, fill, radius=0.06] (-.5,-3.5) circle;
    \draw[arc] (0,-1) -- (-.5,-1.5);
    \draw[arc] (0,-2) -- (-.5,-2.5);
    \draw[arc] (0,-3) -- (-.5,-3.5);
    \draw[arc] (0,0) -- (0,-1);
    \draw[arc] (0,-1) -- (0,-2);
    \draw[arc] (0,-2) -- (0,-3);
    \draw[arc] (0,-3) -- (0,-4);
	\draw (0,-5) node {(a)};
    \end{scope}
    \begin{scope}[xshift=1.5cm,yshift=0cm]
    \draw[thick, fill, radius=0.06] (0,0) circle node[above] {$u$};
    \draw[thick, fill, radius=0.06] (.4,-1) circle;
    \draw[thick, fill, radius=0.06] (0,-1.5) circle;
    \draw[arc] (.4,-1) -- (0,-1.5);
    \draw[thick, fill, radius=0.06] (.8,-2) circle;
    \draw[thick, fill, radius=0.06] (.4,-2.5) circle;
    \draw[arc] (.8,-2) -- (.4,-2.5);
    \draw[thick, fill, radius=0.06] (1.2,-3) circle;
    \draw[thick, fill, radius=0.06] (1.6,-2) circle;
    \draw[thick, fill, radius=0.06] (2,-2.5) circle;
    \draw[arc] (1.6,-2) -- (2,-2.5);
    \draw[thick, fill, radius=0.06] (2,-1) circle;
    \draw[thick, fill, radius=0.06] (2.4,-1.5) circle;
    \draw[arc] (2,-1) -- (2.4,-1.5);
    \draw[thick, fill, radius=0.06] (2.4,0) circle node[above] {$v$};
    \draw[thick, fill, radius=0.06] (1.2,-3.5) circle;
    \draw[arc] (0,0) -- (.4,-1);
    \draw[arc] (.4,-1) -- (.8,-2);
    \draw[arc] (.8,-2) -- (1.2,-3);
    \draw[arc] (2.4,0) -- (2,-1);
    \draw[arc] (2,-1) -- (1.6,-2);
    \draw[arc] (1.6,-2) -- (1.2,-3);
    \draw[arc] (1.2,-3) -- (1.2,-3.5);
	\draw (1.2,-5) node {(b)};
    \end{scope}
    \begin{scope}[xshift=6cm,yshift=0cm]
    \draw[thick, fill, radius=0.06] (0,0) circle node[above] {$\rho$};
    \draw[thick, fill, radius=0.06] (-.3,-1) circle;
    \draw[thick, fill, radius=0.06] (-.8,-1.5) circle;
    \draw[arc] (-.3,-1) -- (-.8,-1.5);
    \draw[thick, fill, radius=0.06] (-.6,-2) circle;
    \draw[thick, fill, radius=0.06] (-1.1,-2.5) circle;
    \draw[arc] (-.6,-2) -- (-1.1,-2.5);
    \draw[thick, fill, radius=0.06] (-.9,-3) circle node[below] {$u$};
    \draw[thick, fill, radius=0.06] (.3,-1) circle;
    \draw[thick, fill, radius=0.06] (.8,-1.5) circle;
    \draw[arc] (.3,-1) -- (.8,-1.5);
    \draw[thick, fill, radius=0.06] (.6,-2) circle;
    \draw[thick, fill, radius=0.06] (1.1,-2.5) circle;
    \draw[arc] (.6,-2) -- (1.1,-2.5);
    \draw[thick, fill, radius=0.06] (.9,-3) circle node[below] {$v$};
    \draw[arc] (0,0) -- (-.3,-1);
    \draw[arc] (-.3,-1) -- (-.6,-2);
    \draw[arc] (-.6,-2) -- (-.9,-3);
    \draw[arc] (0,0) -- (.3,-1);
    \draw[arc] (.3,-1) -- (.6,-2);
    \draw[arc] (.6,-2) -- (.9,-3);
    \draw (0,-5) node {(c)};
    \end{scope}
    \begin{scope}[xshift=11cm,yshift=0cm]
    \draw[thick, fill, radius=0.06] (0,0) circle node[above] {$\rho$};
    \draw[thick, fill, radius=0.06] (-.3,-1) circle;
    \draw[thick, fill, radius=0.06] (-.8,-1.5) circle;
    \draw[arc] (-.3,-1) -- (-.8,-1.5);
    \draw[thick, fill, radius=0.06] (-.9,-3) circle;
    \draw[thick, fill, radius=0.06] (-.9,-3.5) circle;
    \draw[arc] (-.9,-3) -- (-.9,-3.5);
    \draw[thick, fill, radius=0.06] (.3,-1) circle;
    \draw[thick, fill, radius=0.06] (.8,-1.5) circle;
    \draw[arc] (.3,-1) -- (.8,-1.5);
    \draw[thick, fill, radius=0.06] (.6,-2) circle;
    \draw[thick, fill, radius=0.06] (1.1,-2.5) circle;
    \draw[arc] (.6,-2) -- (1.1,-2.5);
    \draw[thick, fill, radius=0.06] (.9,-3) circle node[below] {$v$};
    \draw[arc] (0,0) -- (-.3,-1);
    \draw[arc] (-.3,-1) -- (-.9,-3);
    \draw[arc] (0,0) -- (.3,-1);
    \draw[arc] (.3,-1) -- (.6,-2);
    \draw[arc] (.6,-2) -- (.9,-3);
    \draw[thick, fill, radius=0.06] (-2.1,0) circle node[above] {$u$};
    \draw[thick, fill, radius=0.06] (-1.7,-1) circle;
    \draw[arc] (-2.1,0) -- (-1.7,-1);
    \draw[thick, fill, radius=0.06] (-2.1,-1.5) circle;
    \draw[arc] (-1.7,-1) -- (-2.1,-1.5);
    \draw[thick, fill, radius=0.06] (-1.3,-2) circle;
    \draw[arc] (-1.7,-1) -- (-1.3,-2);
    \draw[thick, fill, radius=0.06] (-1.7,-2.5) circle;
    \draw[arc] (-1.3,-2) -- (-1.7,-2.5);
    \draw[arc] (-1.3,-2) -- (-.9,-3);
    \draw (-.6,-5) node {(d)};
    \end{scope}
    \end{tikzpicture}
    \caption{\label{fig:SideCorrespondence}\remiee{\blue{The (generic)} correspondence of \blue{the} undirected sides of \blue{the generator of} an undirected \blue{binary phylogenetic} network $N$ to \blue{the} directed sides of \blue{the generator of} an \blue{orientation} $N'$ of $N$. If a side $\{u,v\}$ \blue{of $G(N)$} does not contain the root edge, this side corresponds to either (a) one side \blue{of $G(N')$} or (b) two sides \blue{of $G(N')$} separated by a reticulation with a leaf child. If \blue{the side $\{u, v\}$ of $G(N)$} does contain the root $\rho$, \blue{this side} corresponds to either (c) the two sides of $G(N')$ incident with $\rho$, or to three sides \blue{of $G(N')$ as shown in} (d). In this figure, all degree-one vertices are leaves, except the ones labeled $u$ or $v$.}}
\end{figure}

\blue{By} Lemmas~\ref{lem:tcsfblob} and~\ref{lem:3chainreducibleTCSF}, the
algorithms \blue{of} Section~\ref{sec:fptalgs} are applicable to the class of \blue{binary} tree-child networks. \blue{Thus, by Theorem~\ref{thm:fpt-retic}, Theorem~\ref{the:combining} holds provided
$g(L, \ell)$ is a computable function.} We \blue{end} this subsection with the following lemma, which shows that this is indeed the case.

\begin{lemma}\label{lem:CheckNetworkTC/SF}
Let $C$ be \blue{the} class of \blue{binary} tree-child \blue{networks, and let $N$ be a directed binary phylogenetic} network. Then deciding if $N$ is in $C$ takes $O(n)$ time, \blue{where $n$ is the number of vertices in $N$}.
\end{lemma}

\begin{proof}
\blue{To} check whether \blue{$N$ is tree-child, we simply need to check that no reticulation is in a stack or in a pair of sibling reticulations}. \blue{Since this} only \blue{requires} checking the \blue{(local)} neighbourhood of each vertex, which is of size at most three \blue{as $N$ is} binary, this check can be executed in linear time.
\end{proof}

\section{Discussion}\label{sec:discussion}

We have answered several foundational questions regarding the \blue{orientation of} undirected phylogenetic networks. We have also shown that some of our results \blue{apply} to partly-directed phylogenetic networks. Nevertheless, many interesting questions remain open.

Our results do not apply directly to some of the phylogenetic networks published in the biological literature. \blue{The reason for this is that these phylogenetic networks fall outside the framework of} our definition.
It would be interesting to consider modifications of the definition \blue{given here} that allow for the study of such networks from a mathematical point of view. For example, the \blue{phylogenetic} network of grape \blue{cultivars} in~\cite[Fig. 3]{grapes} contains several interesting complications. Firstly, it can be directly observed that any orientation of this \blue{phylogenetic} network needs to have multiple roots (this can, for example, be concluded from the part of the network containing Muscat of Alexandria, Muscat Hamburg, and Trollinger). Secondly, as well as undirected and directed edges, the \blue{phylogenetic} network contains dotted edges \blue{joining} pairs of cultivars which are siblings or equivalent. \blue{Other} examples \blue{include} the \blue{phylogenetic} network of bears in~\cite[Fig. 4]{bears} and the \blue{phylogenetic} network \blue{of the evolutionary history of Europeans} in~\cite[Fig. 1]{humans}. The \blue{first of these phylogenetic} networks contains bidirected arcs, which we have not taken into account in this paper, \blue{while} the \blue{second} has dotted edges indicating that the direction is either unclear (corresponding to our undirected edges) or bidirectional.

\blue{More explicit (computational) questions are the following.} Given an undirected binary phylogenetic network \blue{$N$}, the problem of deciding if \blue{$N$} has a tree-based orientation is NP-complete \shortened{(Section~\ref{sec:ubtbn})}. Although we have shown that \blue{the analogous decision problems for the classes of binary} tree-child and \blue{binary} stack-free networks are fixed-parameter tractable with respect to the level \blue{of $N$}, it remains open whether \blue{these} problems are polynomial-time solvable. We expect both \blue{decision problems} to be NP-complete, but have not found a proof. \leoo{A related question \blue{concerns} \blue{undirected} nonbinary \blue{phylogenetic} networks. It is common in the literature \blue{for directed nonbinary phylogenetic networks} to have the restriction that each \blue{reticulation} has exactly one outgoing arc. Calling such \blue{phylogenetic} networks \blue{with this restriction} \emph{funneled}, an open question is whether one can decide in polynomial time \blue{if} a given undirected nonbinary \blue{phylogenetic} network has a funneled orientation. This is not always the case, as can be seen (with some effort) from the example shown in Figure~\ref{fig:nofunneled}.}

\begin{figure}[h]
\centering
 \begin{tikzpicture}
	 \tikzset{edge/.style={thick}}
     \tikzset{arc/.style={-Latex,thick}}
	 \begin{scope}[xshift=0cm,yshift=0cm]
    \draw[thick, fill, radius=0.06] (0,0) circle;
    \draw[thick, fill, radius=0.06] (0,1) circle;
    \draw[thick, fill, radius=0.06] (1,0) circle;
    \draw[thick, fill, radius=0.06] (1,1) circle;
    
    \draw[thick, fill,
     radius=0.06] (-.5,-.5) circle node[below] {$x_1$};
     \draw[thick, fill,
     radius=0.06] (1.5,-.5) circle node[below] {$x_2$};
     \draw[thick, fill,
     radius=0.06] (1.5,1.5) circle node[above] {$x_3$};
     \draw[thick, fill,
     radius=0.06] (-.5,1.5) circle node[above] {$x_4$};
     
    \draw[edge] (0,0) -- (0,1);
    \draw[edge] (0,0) -- (1,0);
    \draw[edge] (0,1) -- (1,1);
    \draw[edge] (1,0) -- (1,1);
    \draw[edge] (0,0) -- (1,1);
    \draw[edge] (0,1) -- (1,0);
    \draw[edge] (0,0) -- (-.5,-.5);
    \draw[edge] (0,1) -- (-.5,1.5);
    \draw[edge] (1,0) -- (1.5,-.5);
    \draw[edge] (1,1) -- (1.5,1.5);
    
    \end{scope}

	\end{tikzpicture}
\caption{\label{fig:nofunneled} An undirected nonbinary \blue{phylogenetic} network that has no funneled orientation.}
\end{figure}

\remiee{Another question is whether our results generalize to \blue{directed phylogenetic} networks with a root of out-degree~1 or out-degree greater than 2. This only makes sense when we are allowed to root an \blue{undirected phylogenetic} network at an existing \blue{vertex}, instead of at an edge as we have assumed in this paper. Note that, for \blue{directed binary phylogenetic} networks that are blob determined, rooting at an edge is equivalent to adding a leaf to that edge, and rooting \blue{along} the resulting pendant edge. Similarly, rooting at a \blue{vertex} is equivalent to attaching a leaf to the \blue{vertex} via a new edge, and rooting \blue{along this new edge}. For these reasons, we expect our results to generalize.} \shortened{Finally, it would} be interesting to find out whether the results in Section~\ref{sec:Classes} generalize to partly-directed \blue{phylogenetic} networks.

\bibliographystyle{alpha}
\bibliography{sample}

\newpage
\appendix

\section{More classes}\label{app}

\subsection{\remiee{\charles{Some} classes of directed binary \blue{phylogenetic} networks}}\label{sec:classesintro}

For a class~$C$ of directed \blue{phylogenetic} networks, we say that an undirected \blue{phylogenetic} network $N$ is \emph{$C$-orientable} 
if $N$ is the underlying \blue{phylogenetic} network of some directed \blue{phylogenetic} network~$N'$ in~$C$. \leo{If this is the case,~$N'$ is called a \emph{$C$-orientation} 
of~$N$.} We consider the following prominent classes of directed \leoo{binary} \blue{phylogenetic} networks, see also Figure~\ref{fig:classes}. \red{A vertex $v$ of a directed binary phylogenetic network $N$ is {\em visible} if there is a leaf $\ell\in X$ such that every path from the root of $N$ to $\ell$ traverses $v$.}

\begin{figure}[h]
\centering
 \begin{tikzpicture}
	 \begin{scope}[xshift=0cm,yshift=0cm,yscale=.8]
     \draw (0,0) ellipse (8 and 3.3); 
     \draw (0,0) ellipse (6 and 2.2); 
     \draw (0,0) ellipse (4 and 1); 
     \draw (0,0) ellipse (1.1 and .4); 
     \node at (0,0) {tree-child};
     \node at (3,0) {valid};
     \node at (5,0) {stack-free};
     \node at (7,0) {tree-based};
     \draw[blue] (0,-1.1) ellipse (1.9 and 1.8); 
     \node[blue] at (1.95,-2.6) {orchard};
     \draw[blue] (1.25,-2.45) -- (1.35,-2.6);
     \draw[red] (-2,0) ellipse (3.2 and 1.5); 
     \draw[red] (-0.4,1.3) -- (-0.25,1.45);
     \node[red] at (1.2,1.5) {reticulation-visible};
    \end{scope}
	\end{tikzpicture}
\caption{\label{fig:classes} An overview of the directed \leoo{binary} 
\blue{phylogenetic} network classes \remiee{considered in this paper}. Note that there are eleven regions and each of them can be shown to be nonempty. \leoo{Also note that, in the nonbinary case, the landscape looks a bit different.}}
\end{figure}

Let~$N$ be a directed \leoo{binary} \blue{phylogenetic} network \blue{on $X$}. Then
\begin{itemize}
\item $N$ is \emph{tree-child} if every non-leaf vertex has a child that is a tree vertex. 
\red{It is straightforward to show that $N$ is tree-child if and only if every vertex is visible.} For example, the directed \blue{phylogenetic} network in Figure~\ref{fig:intro} is tree-child.

\item $N$ is \emph{stack-free} if no reticulation has a reticulation as a child. Clearly, \blue{the class of} tree-child \blue{networks} is \blue{contained in the class of} stack-free \blue{networks}. See Figure~\ref{fig:example3} for an example of \blue{an undirected binary phylogenetic} network that is orientable, but not stack-free orientable (and hence not tree-child orientable).

\item $N$ is \emph{tree-based} if it can be obtained from a \leoo{directed} \blue{binary phylogenetic} tree, \blue{called a {\em base-tree}}, in which the root may have out-degree~$1$ or~$2$, by subdividing arcs of the tree (any number of times) and adding arcs \blue{(called {\em linking arcs})} between the subdividing vertices~\cite{FrancisSteel2015,jetten2018nonbinary}. The class \leoo{of binary tree-based networks} contains the class \leoo{of binary stack-free networks~\cite{zhang2016tree}.} For an example of an undirected \blue{binary phylogenetic} network \blue{that is orientable but has} no tree-based orientation, \red{consider the undirected binary phylogenetic network $N$ in} Figure~\ref{fig:nottreebased}. \red{It is easily checked that $N$ is orientable. To see that $N$ is not tree-based orientable, suppose $N$ has such an orientation. Now $N$ has, as subgraphs, two vertex-disjoint subdivisions of the Petersen graph, and at least one of these subdivisions does not contain the root edge. Thus, ignoring the orientation of the arcs, a base-tree of this orientation would realise \may{(together with an extra edge)} a Hamiltonian cycle of the Petersen graph. But it is well known that the Petersen graph is not Hamiltonian.}

\item $N$ is \emph{valid} if deleting any reticulation arc and suppressing its end-\blue{vertices} results in a directed \blue{binary phylogenetic} network. \red{The class of valid networks is contained in the class of stack-free networks and contains the class of tree-child networks}~\cite{murakami2018reconstructing}.
    
\item $N$ is \emph{reticulation-visible} (or \emph{stable}) if \red{every reticulation is visible, in which case $N$ is stack-free (see~\cite{sem18})}.
\end{itemize}
\latest{We list a few relations between classes of directed binary phylogenetic networks. The class of level-$1$ networks is contained in the class of tree-child networks. To see this, \red{first} note that a \red{directed binary phylogenetic network is} tree-child \red{precisely if it has no} vertex \red{whose} two children are both reticulations, \red{and no} reticulation \red{whose only child is also a reticulation} (see~\cite{sem16}). \red{It is easily checked that, if a directed binary phylogenetic network has such a pair of reticulations \may{(that have a common parent or are adjacent)}, then this pair \may{is in the same biconnected component}. Therefore all level-$1$ networks are tree-child.} Furthermore, the class of level-$2$ networks is not contained in the class of stack-free networks (consider a network with exactly two reticulations, one of which is the child of the other) and the class of level-$3$ networks is not contained in the class of tree-based networks~(see e.g.~\cite[Figure~3]{jetten2018nonbinary}).}

\begin{figure}[h]
\centering
 \begin{tikzpicture}
	 \tikzset{edge/.style={thick}}
     \tikzset{arc/.style={-Latex,thick}}
	 \begin{scope}[xshift=0cm,yshift=0cm]
	\draw[thick, fill, radius=0.06] (0,0) circle;
    \draw[thick, fill, radius=0.06] (2,0) circle;
    \draw[thick, fill, radius=0.06] (1,-1) circle;
    \draw[thick, fill, radius=0.06] (0,-2) circle;
    \draw[thick, fill, radius=0.06] (1,-2) circle;
    \draw[thick, fill, radius=0.06] (2,-1) circle;
    \draw[thick, fill, radius=0.06] (2.5,-1.5) circle node[right] {$x$};
    \draw[thick, fill, radius=0.06] (1.5,-2.5) circle node[below] {$y$};
    \draw[edge] (0,0) -- (2,0);
    \draw[edge] (0,0) -- (0,-2);
    \draw[edge] (0,0) -- (1,-1);
    \draw[edge] (2,0) -- (1,-1);
    \draw[edge] (0,-2) -- (1,-2);
    \draw[edge] (2,0) -- (2,-1);
    \draw[edge] (2,-1) -- (1,-2);
    \draw[edge] (0,-2) -- (1,-1);
    \draw[edge] (2,-1) -- (2.5,-1.5);
    \draw[edge] (1,-2) -- (1.5,-2.5);
    \end{scope}
    \begin{scope}[xshift=5cm,yshift=0cm]
	\draw[thick, fill, radius=0.06] (0,0) circle;
    \draw[thick, fill, radius=0.06] (2,0) circle;
    \draw[thick, fill, radius=0.06] (1,-1) circle;
    \draw[thick, fill, radius=0.06] (0,-2) circle;
    \draw[thick, fill, radius=0.06] (1,-2) circle;
    \draw[thick, fill, radius=0.06] (2,-1) circle;
    \draw[thick, fill, radius=0.06] (2.5,-1.5) circle;
    \draw[thick, fill, radius=0.06] (3,-2) circle;
    \draw[thick, fill, radius=0.06] (2,-3) circle;
    \draw[thick, fill, radius=0.06] (1.5,-2.5) circle;
    \draw[thick, fill, radius=0.06] (2.5,-2.5) circle;
    \draw[thick, fill, radius=0.06] (1,-3) circle node[below] {$x_1$};
    \draw[thick, fill, radius=0.06] (2,-3.5) circle node[below] {$x_2$};
    \draw[thick, fill, radius=0.06] (3,-3) circle node[below right] {$x_3$};
    \draw[thick, fill, radius=0.06] (3.5,-2) circle node[right] {$x_4$};
    \draw[thick, fill, radius=0.06] (3,-1) circle node[right] {$x_5$};
    \draw[edge] (0,0) -- (2,0);
    \draw[edge] (0,0) -- (0,-2);
    \draw[edge] (0,0) -- (1,-1);
    \draw[edge] (2,0) -- (1,-1);
    \draw[edge] (0,-2) -- (1,-2);
    \draw[edge] (2,0) -- (2,-1);
    \draw[edge] (2,-1) -- (1,-2);
    \draw[edge] (0,-2) -- (1,-1);
    \draw[edge] (2,-1) -- (3,-2);
    \draw[edge] (1,-2) -- (2,-3);
    \draw[edge] (3,-2) -- (2,-3);
    \draw[edge] (1.5,-2.5) -- (1,-3);
    \draw[edge] (2,-3) -- (2,-3.5);
    \draw[edge] (2.5,-2.5) -- (3,-3);
    \draw[edge] (3,-2) -- (3.5,-2);
    \draw[edge] (2.5,-1.5) -- (3,-1);
    \end{scope}
	\end{tikzpicture}
\caption{\label{fig:example3} Two undirected \blue{binary phylogenetic} networks that are orientable \remiee{(where the root can be inserted into any edge)} but not stack-free orientable. The network to the right can be extended to any number of leaves.}
\end{figure}

\begin{figure}
\centering
 \begin{tikzpicture}
	 \tikzset{edge/.style={thick}}
     \tikzset{arc/.style={-Latex,thick}}
	 \begin{scope}[xshift=0cm,yshift=0cm]
	\draw[thick, fill, radius=0.06] (0,0) circle;
    \draw[thick, fill, radius=0.06] (.75,0) circle;
    \draw[edge] (0,0) -- (.75,0);
    \draw[thick, fill, radius=0.06] (1,1.5) circle;
    \draw[edge] (0,0) -- (1,1.5);
    \draw[thick, fill, radius=0.06] (1.5,1) circle;
    \draw[edge] (1,1.5) -- (1.5,1);
    \draw[thick, fill, radius=0.06] (1,-1.5) circle;
    \draw[edge] (0,0) -- (1,-1.5);
    \draw[thick, fill, radius=0.06] (1.5,-1) circle;
    \draw[edge] (1,-1.5) -- (1.5,-1);
    \draw[thick, fill, radius=0.06] (3,1.5) circle;
    \draw[edge] (1,1.5) -- (3,1.5);
    \draw[thick, fill, radius=0.06] (2.5,1) circle;
    \draw[edge] (3,1.5) -- (2.5,1);
    \draw[thick, fill, radius=0.06] (3,-1.5) circle;
    \draw[edge] (1,-1.5) -- (3,-1.5);
    \draw[thick, fill, radius=0.06] (2.5,-1) circle;
    \draw[edge] (3,-1.5) -- (2.5,-1);
    \draw[thick, fill, radius=0.06] (3.5,.5) circle;
    \draw[edge] (3,1.5) -- (3.5,.5);
    \draw[thick, fill, radius=0.06] (3.5,-.5) circle;
    \draw[edge] (3,-1.5) -- (3.5,-.5);
    \draw[edge] (3.5,.5) -- (3.5,-.5);
    \draw[edge] (1.5,1) -- (2.5,-1);
    \draw[edge] (1.5,1) -- (1.5,-1);
    \draw[edge] (.75,0) -- (2.5,1);
    \draw[edge] (.75,0) -- (2.5,-1);
    \draw[edge] (2.5,1) -- (1.5,-1);
    \draw[thick, fill, radius=0.06] (4,-.5) circle node[below] {$x$};
    \draw[edge] (3.5,-.5) -- (4,-.5);
    \draw[thick, fill, radius=0.06] (5,-.5) circle node[below] {$y$};
    \draw[edge] (5,-.5) -- (5.5,-.5);
    \draw[thick, fill, radius=0.06] (5.5,-.5) circle;
    \draw[thick, fill, radius=0.06] (5.5,.5) circle;
    \draw[edge] (3.5,.5) -- (5.5,.5);
    \draw[edge] (5.5,-.5) -- (5.5,.5);
    \draw[thick, fill, radius=0.06] (6,1.5) circle;
    \draw[edge] (5.5,.5) -- (6,1.5);
    \draw[thick, fill, radius=0.06] (6,-1.5) circle;
    \draw[edge] (5.5,-.5) -- (6,-1.5);
    \draw[thick, fill, radius=0.06] (6.5,1) circle;
    \draw[edge] (6,1.5) -- (6.5,1);
    \draw[thick, fill, radius=0.06] (6.5,-1) circle;
    \draw[edge] (6,-1.5) -- (6.5,-1);
    \draw[thick, fill, radius=0.06] (8,1.5) circle;
    \draw[edge] (6,1.5) -- (8,1.5);
    \draw[thick, fill, radius=0.06] (7.5,1) circle;
    \draw[edge] (8,1.5) -- (7.5,1);
    \draw[thick, fill, radius=0.06] (8,-1.5) circle;
    \draw[edge] (6,-1.5) -- (8,-1.5);
    \draw[thick, fill, radius=0.06] (7.5,-1) circle;
    \draw[edge] (8,-1.5) -- (7.5,-1);
    \draw[thick, fill, radius=0.06] (9,0) circle;
    \draw[edge] (8,1.5) -- (9,0);
    \draw[edge] (8,-1.5) -- (9,0);
    \draw[thick, fill, radius=0.06] (8.25,0) circle;
    \draw[edge] (8.25,0) -- (9,0);
    \draw[edge] (8.25,0) -- (6.5,1);
    \draw[edge] (8.25,0) -- (6.5,-1);
    \draw[edge] (6.5,-1) -- (7.5,1);
    \draw[edge] (7.5,-1) -- (6.5,1);
    \draw[edge] (7.5,-1) -- (7.5,1);
    \end{scope}
	\end{tikzpicture}
\caption{\label{fig:nottreebased} An undirected \blue{binary phylogenetic} network \blue{$N$} that is orientable but not tree-based orientable.}
\end{figure}




A further class of directed \leoo{binary}
\blue{phylogenetic} networks is defined based on the notion of cherry picking. \blue{Let $N$ be a directed binary phylogenetic network.} A \emph{cherry} \blue{of $N$} is a pair of leaves with the same parent. \emph{Picking a cherry} means deleting one leaf of the cherry and suppressing its parent \latest{(where \emph{suppressing} a vertex~$v$ with exactly one parent~$u$ and exactly one child~$w$ means that we delete~$v$ and add an arc~$(u,w)$)}. A \emph{reticulated cherry} \blue{of $N$} is a pair of leaves connected by an (underlying) undirected path
with two internal vertices, exactly one of which is a reticulation. \emph{Picking a reticulated cherry} consists of deleting the middle arc of this path and suppressing its endpoints. \blue{We say $N$} is \emph{\leoo{an orchard network}} if it can be reduced to a single \leo{cherry} 
by \yukiii{repeatedly} picking cherries and reticulated cherries. \may{In particular, all tree-child networks are orchard networks~\cite{bordewich2016determining}.} \blue{If $N$ is an orchard network and $N'$ is obtained from $N$ by picking either a cherry or a reticulated cherry, then $N'$ is an orchard network. This leads to a} linear-time \blue{algorithm for deciding} whether an arbitrary directed binary network is \leoo{an orchard network}~\cite{erdos2019orchards,janssen2018cpn}. \blue{It is also used in the proof of the following proposition which shows that} all orchard networks are tree-based.



\begin{proposition}
\blue{The class of binary orchard networks is contained in the class of binary tree-based networks.}
\end{proposition}

\begin{proof}
\blue{Let $N$ be an binary orchard network.} The proof is by induction on the number \blue{$k$} of arcs of \blue{$N$}. \blue{If $k=2$, then $N$ consists of a pair of leaves adjoined to the root, and so $N$ is tree-based. Now suppose that $k\ge 3$ and that} all \blue{binary} orchard networks with \blue{at most $k-1$ reticulation arcs} are tree-based. \blue{Since $N$ is orchard, it} contains \blue{either} a cherry or a reticulated cherry. Let~$N'$ be \blue{the binary orchard network} obtained from~$N$ by picking a cherry or reticulated cherry. Since~$N'$ is orchard, it follows by the induction \blue{assumption} that $N'$ can be obtained from a \blue{directed binary phylogenetic} tree~$T'$ by subdividing arcs of~$T'$ and adding \blue{linking arcs} between the subdividing vertices.

First suppose that~$N'$ is obtained from~$N$ by picking a cherry consisting of leaves~$x$ and $y$, and say that leaf~$y$ is deleted. Let~$T$ be the \blue{directed binary phylogenetic} tree obtained from~$T'$ by adding leaf~$y$ as a sibling of~$x$, i.e., subdividing the pendant arc \blue{directed into} $x$ with a \blue{vertex} $p$ and adding the arc $(p,y)$. Then~$N$ can be obtained from~$T$ by subdividing arcs of $T$ and adding linking arcs between the subdividing vertices in exactly the same way as~$N'$ is obtained from $T'$. Now suppose that $N'$ is obtained from $N$ by picking a reticulated cherry \blue{in which} the arc~$(u,v)$ \blue{is deleted}. Then $N$ can be obtained from $T'$ by subdividing arcs of~$T'$ and adding linking arcs between the subdividing vertices \remiee{in the same way as~$N'$ is obtained from~$T'$ except that the arc~$(u,v)$ is \blue{also} added as} a linking arc. The result now follows.
\end{proof}

\shortened{In the remainder of this section, we prove the following.}

\begin{theorem}\label{the:combining}
Let $C$ \blue{be} a class of directed binary \blue{phylogenetic} networks. If $C$ is tree-child, stack-free, tree-based, reticulation-visible, valid, or orchard, then
\shortened{Algorithm~\ref{alg:C_orientation} is an FPT algorithm for deciding whether an undirected binary phylogenetic network~$N$ has a $C$-orientation, where the level of~$N$ is the parameter.}
\end{theorem}

\shortened{The proof for stack-free is completely analogous to the proof for tree-child, which has been dealt with in Theorem~\ref{the:combining_tc}. The remaining classes are treated in the following subsections~\ref{treebasedsection}-\ref{reticulation-visible}.}

\subsection{Tree-based networks}\label{treebasedsection}


\blue{In this section, we establish Theorem~\ref{the:combining} for the class of tree-based networks.} There are several characterizations of tree-based networks, and we will use the following due to \cite{zhang2016tree} (also see \blue{\cite{hayamizu2018, jetten2018nonbinary}).}
\blue{Let $N$ be a directed binary phylogenetic network.} A \blue{{\em chain of sibling reticulations}} \blue{of $N$} is a sequence of (not necessarily distinct) vertices
$$u_1,\, v_1,\, u_2,\, \ldots,\, u_k,\, v_k,\, u_{k+1}$$
such that, for all $i\in \{1, 2, \ldots, k\}$, the vertex $v_i$ is a \blue{reticulation and a} child of $u_i$ and~$u_{i+1}$. \blue{Observe that $u_2, u_3, \ldots, u_k$ are necessarily tree vertices.} \blue{If, in addition,} $u_1$ and $u_{k+1}$ are both reticulations, \blue{we say this chain is {\em terminating}}.
The \blue{above-mentioned characterization is that} \blue{$N$} is tree-based if and only if \blue{$N$ has no terminating chain of sibling reticulations}.

\blue{Let $N$ be a directed binary phylogenetic network. If $N$} is tree-based, then, as adding leaves to \blue{$N$} cannot create a \blue{terminating chain of sibling reticulations}, it \blue{follows} that the class of \blue{binary} tree-based networks is leaf-addable. Moreover, \blue{by a similar argument to that proving Lemma~\ref{lem:tcsfblob}}, this class is blob-determined because \blue{if $N$ has} a \blue{terminating chain of sibling reticulations}, it must be contained in a blob \blue{of $N$}. The next lemma shows that the class of \blue{binary} tree-based networks is rooted $2$-chain reducible.





\begin{lemma}\label{lem:TBChainRed}
The class of \blue{binary} tree-based networks is \remiee{rooted} $2$-chain reducible.
\end{lemma}

\begin{proof}
Let~$N$ be an undirected \blue{binary phylogenetic} network that can be tree-based rooted at a pendant edge $e_\rho\blue{=\{v_{\rho}, x_{\rho}\}}$, where $x_{\rho}$ is a leaf, and let $N_d$ be a \blue{tree-based orientation of $N$ rooted at $e_{\rho}$}. Note that $v_{\rho}$ is a tree vertex \blue{in $N_d$}. Let $s_1$ and $s_2$ \blue{denote the two sides of $G(N_d)$ leaving $v_{\rho}$}. \blue{We next construct a directed binary phylogenetic network $N'_d$ from $N_d$ based on $G(N)$.}

Let $s_\rho$ be the side of $G(N)$ containing the root, \red{and let $P_{s_\rho}$ be an undirected path of $N$ corresponding to $s_\rho$}.
First consider the sides $s\neq s_\rho$ of $G(N)$. If $s$ corresponds to a single side of~$G(N_d)$, delete all except one of the leaves \blue{of $N_d$} that are on~$s$, and suppress any resulting vertices of in-degree one and out-degree one. This creates no \blue{terminating chain of sibling reticulations}.
If $s$ corresponds to two sides of $G(N_d)$ meeting at a reticulation~$r$ with a leaf-child~$x_r$, then delete all leaves \blue{of $N_d$} on $s$ except for $x_r$ and one other leaf, and suppress any resulting vertices of in-degree one and out-degree one. This again ensures that no \blue{terminating chain of sibling reticulations} are introduced.

Now consider $s_\rho$,\remie{which corresponds to at most three sides of~$G(N_d)$, namely $s_1$, $s_2$ and, if \remie{an internal vertex of \blue{$P_{s_{\rho}}$ corresponds to} a reticulation~$r$ in $N_d$ \leoooo{adjacent to a leaf}, }
a third side $s_3$ (see Figure~\ref{fig:SideCorrespondence}). If there is no such third side, then delete all leaves \blue{of $N_d$} on \blue{$s_1$ and $s_2$}, and suppress any resulting vertices of in-degree one and out-degree one. If there is such a third side $s_3$, delete all leaves \blue{of $N_d$} on \blue{$s_1$, $s_2$, and $s_3$}, and suppress any resulting vertices of in-degree one and out-degree one. \blue{Note that, by definition}, $x_{\rho}$ and the child of $r$ \blue{are not deleted}. Again, no \blue{terminating chain of sibling reticulations} are created as one end-vertex of \blue{either $s_1$ or $s_2$ is} a tree vertex of $N_d$ whose parent is also a tree vertex.} \blue{This gives $N'_d$.}
Taking the same approach as that which concluded the proof of Lemma~\ref{lem:3chainreducibleTCSF}, it now follows that the class of \blue{binary} tree-based networks is \latest{rooted} $2$-chain reducible.
\end{proof}

\blue{Since deciding if a directed binary phylogenetic network with $n$ vertices is tree-based takes $O(n)$ \revL{time~\cite{zhang2016tree,hayamizu2018}}, it now follows that Theorem~\ref{the:combining} holds} for the class of tree-based networks.

\subsection{Orchard networks}


\blue{In this section, we show that Theorem~\ref{the:combining} holds} for the class of binary orchard networks. We begin by showing that this class is \blue{rooted $3$-chain reducible, leaf-addable, and blob-determined}.

\begin{lemma}\label{lem:orchardLeafAddable}
The class of \blue{binary} orchard networks is leaf-addable.
\end{lemma}

\begin{proof}
Let $N$ be a \blue{binary} orchard network, and suppose $N'$ is obtained from $N$ by adding a single leaf $z$. Since $N$ is orchard, \blue{$N$} can be reduced to a single cherry by a sequence of operation\blue{s} each of which picks either a cherry or a reticulated cherry. \blue{Now} apply the same operations to $N'$ until this is no longer possible, that is \remie{until the next cherry or reticulated cherry to pick is not present.}
There are two \blue{possibilities to consider.}


The first possibility is that the next operation is to pick a cherry, $\{x, y\}$ say, but, without loss of generality, $\{y, z\}$ is a cherry. In this possibility, we \blue{next} pick $\{y, z\}$ by deleting $z$, and then continue the sequence of operations as for $N$, thereby reducing $N'$ to a single cherry, \blue{and so $N'$ is orchard}.

The second possibility is that the next operation is to pick a reticulated cherry, say $\{x, y\}$, where the parent of $x$ is a reticulation $v$. Here, the parent of $y$ is a tree vertex, $u$ say, \blue{and, when we apply the sequence of operations to $N$}, $u$ is a parent of $v$. \blue{In this possibility, the reason $\{x, y\}$ is not a} \remie{reticulated cherry is that either $\{x, z\}$ or $\{y, z\}$ is a cherry, or $\{x, z\}$ is a reticulated cherry in which the parent of $z$ is a child of $u$ and a parent of $v$. In the first \blue{instance}, we next pick $\{x, z\}$ or $\{y, z\}$, \blue{depending on which is a cherry}, by deleting $z$. In the \blue{second instance},} we next pick the reticulated cherry $\{x, z\}$, and then \blue{pick} the \blue{resulting} cherry $\{y, z\}$ by deleting~$z$. \blue{In both instances}, we then continue the sequence of operations as for $N$, thereby reducing $N'$ to a single cherry, \blue{and so again $N'$ is orchard}. The lemma now follows.
\end{proof}

\begin{lemma}\label{lem:orchardblobdetermined}
The class of \blue{binary} orchard networks is blob-determined.
\end{lemma}

\begin{proof}
Let $N$ be \blue{a binary} orchard \blue{network}, and \blue{let} $B$ \blue{be} a blob of $N$. \blue{Since} \blue{$N$} can be reduced \blue{to a single cherry} by picking cherries and reticulated cherries in any order~\cite{erdos2019orchards,janssen2018cpn} \blue{and, if $r$ is a reticulation of $N$, then $r$ and both of its parents are in the same blob}, we can \blue{apply a a sequence of cherry-picking operations} not affecting~$B$ until each out-going cut-arc of $B$ is pendant. Then, by the same reasoning, we can \blue{continue the sequence} picking cherries and reticulated cherries only within $B$ until \blue{$B$} is reduced to a single cherry. \blue{It follows that} the second part of this sequence induces a sequence of cherry-picking operations that reduces the \blue{directed binary phylogenetic} network induced by $B$ to a single cherry. \blue{Thus} the \blue{directed binary phylogenetic} network induced by $B$ is orchard.

\blue{Conversely,} suppose that \blue{every directed binary phylogenetic} network induced by a blob of $N$ is orchard. \blue{Then} we can reduce $N$ \blue{to a single cherry by a sequence of cherry-picking operations} by \blue{systematically} reducing each \blue{pendant} blob. This shows that $N$ is orchard, \blue{and completes the proof of the lemma}.
\end{proof}


\begin{lemma}
The class of \blue{binary} orchard networks is \remiee{rooted} $3$-chain reducible.\label{lem:orchardreducible}
\end{lemma}

\begin{proof}
Let $N$ be an \leooo{undirected} \blue{binary phylogenetic} network \blue{that can be} orchard \blue{rooted at a pendant edge $e_{\rho}$}, and let $N_d$ be an orchard orientation \blue{of $N$ rooted at $e_{\rho}$}. \blue{Since $N_d$ is orchard, it can be reduced to a single cherry by a sequence $S$ of cherry-picking operations.} Using $G(N_d)$, \blue{we construct a directed binary phylogenetic network $N'_d$ from $N_d$ as follows.} First, for each side $s=(u, v)$ of $G(N_d)$ in which $v$ is a tree vertex, delete all of the leaves of $N_d$ on $s$, and suppress the resulting vertices of in-degree one and out-degree one. \blue{Second}, for each pairs of sides $s=(u, r)$ and $s'=(u', r)$ of $G(N_d)$ in which $r$ is a reticulation, delete all of the leaves of $N_d$ on $s$ and $s'$ except one leaf whose parent is a parent of $r$, \blue{and suppress the resulting vertices of in-degree one and out-degree one}. This gives $N'_d$. \blue{We next show that $N'_d$ is orchard by showing that a modification of $S$ applied to $N'_d$ reduces $N'_d$ to a single cherry.}

\blue{Noting that $G(N_d)=G(N'_d)$, if we apply $S$ to $N'_d$, then the sequence halts because we have reached either (i) a side $s=(u, v)$ of $G(N_d)$ in which $v$ is a tree vertex and all of the leaves of $N_d$ on $s$ have been deleted, or (ii) two distinct sides $s=(u, r)$ and $s'=(u', r)$ of $G(N_d)$ in which $r$ is a reticulation and all of the leaves on $s$ and $s'$ have been deleted except one leaf, $x$ say, whose parent is also a parent of $r$ in $N_d$.} Without loss of generality, we may assume that $x$ is on $s$.

\blue{Assume we have reached (i).} \blue{Then, in applying $S$ to $N'_d$}, the vertex $v$ is suppressed and $u$ gets the child of $v$, say $x'$, as a child. At this point, \blue{in the analogous application of $S$ to $N_d$}, the vertex $u$ is the root of a pendant subtree, \blue{in which case} we may assume that leaves of this pendant subtree are systematically deleted via picking cherries until the only leaf of the pendant subtree that remains is $x'$. Hence, \blue{by omitting these cherry-picking operations from $S$}, we can continue to apply the resulting sequence to $N'_d$.

\blue{Now assume that we have reached (ii).} \blue{Then, applying $S$ to $N'_d$}, the child of $r$ is a leaf, $y$ say. At this point, \blue{in the analogous application of $S$ to $N_d$}, the next operation involves picking a reticulated cherry $R$, where $y\in R$. There are two possibilities \blue{depending on whether $x\in R$}. If $R=\{x,y\}$, then picking $\{x, y\}$ results in each of \blue{$u$ and $u'$} being the roots of pendant \red{phylogenetic} subtrees, \blue{in which case} we may assume that the leaves in the pendant \red{phylogenetic} subtree rooted at $u$ are systematically deleted via picking cherries until $x$ is a child of $u$ and, similarly, the leaves in the pendant \red{phylogenetic} subtree rooted at $u'$ are systematically deleted via picking cherries until $y$ is a child of $u'$. Therefore, by picking the reticulated cherry $\{x, y\}$, and \blue{omitting the operations involving the picking cherries of the two pendant \red{phylogenetic} subtrees from $S$}, we can continue to apply the resulting sequence to $N'_d$.

\blue{If} $R=\{y, z\}$, \blue{then} $z$ is a leaf of $N_d$ on side $s'$. \blue{Here}, picking $\{y, z\}$ again results in \blue{$u$ and $u'$} being the roots of pendant \red{phylogenetic} subtrees, and so we may assume that the leaves in each of these pendant \red{phylogenetic} subtrees are systematically deleted until $x$ is a child of $u$ and $z$ is a child of $u'$. Therefore, by replacing $\{x, y\}$ with $\{y, z\}$ in $S$ and omitting the operations involving the picking of cherries of the two pendant \red{phylogenetic} subtrees, we can continue to apply the resulting sequence to $N'_d$ but with $z$ replaced by $y$. \blue{A routine induction argument now shows that $N'_d$ is orchard.}

\blue{Let $N'$ denote the underlying undirected binary phylogenetic netwoork of $N'_d$. Using an argument analogous to that in} the proof that showed the class of binary tree-child networks is rooted $3$-chain reducible, \blue{each side of $G(N')$ contains at most three leaves}. \blue{Furthermore, if $N^r$ is the undirected binary phylogenetic network obtained from $N$ by applying a rooted $3$-chain reduction rooted at $e_{\rho}$, then $N^r$ can be obtained from $N'$ by adding leaves and, if necessary, relabelling leaves}. \blue{It follows by Lemma~\ref{lem:orchardLeafAddable} that, as $N'$ can be orchard rooted at $e_{\rho}$}, the class of binary orchard networks is rooted $3$-chain reducible.
\end{proof}

\blue{Since deciding if a directed binary phylogenetic network with $n$ vertices is orchard} takes $O(n)$ time~\cite{janssen2018cpn}, it follows by combining Proposition~\ref{prop:3conditionslChainRed} with Lemmas~\ref{lem:orchardLeafAddable}, \ref{lem:orchardblobdetermined}, and~\ref{lem:orchardreducible} that \blue{Theorem~\ref{the:combining} holds for} the class of orchard networks.

\subsection{Valid networks}

Recall that a \blue{directed binary phylogenetic} network is \emph{valid} if deleting any reticulation arc and suppressing its end-vertices results in a directed binary phylogenetic network. \blue{In particular, this implies that the resulting directed graph has no} parallel arcs \blue{and no} unlabelled out-degree-$0$ vertices. It is \blue{shown in} \cite{murakami2018reconstructing} that a \blue{directed binary phylogenetic} network $N$ is valid if and only if $N$ contains no stack \blue{reticulations} and \blue{no} camels. A \emph{camel} is \blue{two sibling reticulations} in which the common parent and one of the reticulations share a common parent. \blue{An illustration of a camel is shown in} Figure~\ref{fig:camel}. Using this characterization \blue{and taking the same approach as that taken for binary tree-child and binary stack-free networks}, it can be shown that the class of \blue{binary} valid networks is rooted $3$-chain reducible, leaf-addable, and blob-determined, and that deciding if a directed binary phylogenetic network $N$ is valid takes $O(n)$ time, where $n$ is the number of vertices of $N$. \blue{It now follows that Theorem~\ref{the:combining} holds for the class of binary valid networks.}

\begin{figure}[h]
\centering
 \begin{tikzpicture}
	 \tikzset{edge/.style={thick}}
     \tikzset{arc/.style={-Latex,thick}}
	 \begin{scope}[xshift=0cm,yshift=0cm]
    \draw[thick, fill, radius=0.06] (0,0) circle;
	\draw[arc] (-.5,.5) -- (0,0);
	\draw[arc] (0,0) -- (0,-.5);
    \draw[thick, fill, radius=0.06] (1,1) circle node[left] {$v$};
    \draw[thick, fill, radius=0.06] (2,0) circle node[right] {$w$};
    \draw[arc] (1,1) -- (0,0);
    \draw[arc] (1,1) -- (2,0);
    \draw[thick, fill, radius=0.06] (2,2) circle;
    \draw[arc] (2,2) -- (1,1);
    \draw[arc] (2,2) -- (2,0);
    \draw[arc] (2,0) -- (2,-.5);
    \draw[arc] (2,2.5) -- (2,2);
    \end{scope}
	\end{tikzpicture}
\caption{\label{fig:camel} A camel, \blue{where the common parent $v$ of the two reticulations and $w$, one of the two reticulations, share a common parent (also see \cite[Figure~3]{hayamizu2018})}.}
\end{figure}

\subsection{Reticulation-visible networks}
\label{reticulation-visible}

\blue{Lastly, in this section, we show that Theorem~\ref{the:combining} holds for the class of reticulation-visible networks.} Recall that a \blue{directed binary phylogenetic} network $N$ is \emph{reticulation-visible} if, for each reticulation $r$, there is a leaf $x$ such that every path from the root of $N$ to $x$ traverses $r$, in which case $r$ is said to \emph{be visible from} $x$. It is easy to see that the class of \blue{binary} reticulation-visible networks is leaf-addable and blob-determined. \blue{We now show that this class is rooted $3$-chain reducible}.



\begin{lemma}
The class of \blue{binary} reticulation-visible networks is \remiee{rooted} $3$-chain reducible.\label{lem:retvis-3chain}
\end{lemma}

\begin{proof}
Let $N$ be an undirected \blue{binary phylogenetic} network \blue{that can be} reticulation-visible \blue{rooted at a pendant edge $e_{\rho}=\{v_{\rho}, x_{\rho}\}$, where $x_{\rho}$ is a leaf, and let} $N_d$ be a reticulation-visible orientation of $N$ \blue{rooted at $e_{\rho}$}. In a now familiar way, \blue{we next construct a directed binary phylogenetic network $N'_d$ from $N_d$ as follows}. \blue{Let $s_1$ and $s_2$ denote the two sides of $G(N_d)$ leaving $v_{\rho}$. First,} for each side of $G(N_d)$ \blue{that is neither $s_1$ nor $s_2$}, delete all except one of the leaves of $N_d$ on $s$ \blue{and suppress the resulting vertices of in-degree one and out-degree one}. \blue{At this stage of the construction, the resulting directed binary phylogenetic network} is reticulation-visible since, if a reticulation $r$ of $N_d$ is visible from a leaf $x$ \blue{that is} on side $s$ of $G(N_d)$, then $r$ is visible from each leaf on side $s$. \blue{Continuing the construction,} delete all leaves on \blue{$s_1$ and $s_2$}. As no reticulation \blue{of $N_d$} is visible from any of these leaves, \blue{the resulting directed binary phylogenetic network $N'_d$ is reticulation-visible}. \blue{Analogous} to the way in which the proof of Lemma~\ref{lem:3chainreducibleTCSF} \blue{is completed}, it can now be shown that the class of \blue{binary} reticulation-visible networks is rooted $3$-chain reducible.
\end{proof}

\blue{We next show that deciding if a directed binary phylogenetic network is reticulation-visible can be computed in polynomial time, thereby, in combination with the above, establishes Theorem~\ref{the:combining} for the class of reticulation-visible networks.}

\begin{lemma}
\blue{Let $N$ be a directed binary phylogenetic network. Deciding if $N$ is} reticulation-visible can be done in time $O(kn)$, \blue{where $k$ and $n$ are the number of reticulations and vertices of $N$, respectively}.
\label{visiblelast}
\end{lemma}

\begin{proof}
\remiee{To \blue{decide} if $N$ is reticulation-visible, \blue{it suffices to} check, for each reticulation $r$, whether or not, for each leaf $x$ of $N$, there is directed path from the root of $N$ to $x$ in the directed graph obtained from $N$ by \blue{deleting} $r$ \blue{and its incident arcs}. \blue{If there is no such path for some leaf $y$, then $r$ is visible from $y$; otherwise $r$ is not visible from any leaf.} In a digraph with $n$ vertices, determining which vertices can be reached from some fixed vertex, \blue{in this case the root of $N$}, can be done in $O(n)$ time. As \blue{$N$ has} $k$ reticulations, \blue{the lemma now follows}.}
\end{proof}

\subsection{Undirected tree-based networks}\label{sec:ubtbn}

\begin{figure}
\centering
 \begin{tikzpicture}
	 \tikzset{edge/.style={thick}}
     \tikzset{arc/.style={-Latex,thick}}
     	 \begin{scope}[xshift=0cm,yshift=0cm]
    \draw[thick, fill, radius=0.06] (5,-.5) circle node[below] {$y$};
    \draw[edge,dotted] (5.5,-.5) -- (5,-.5);
    \draw[thick, fill, radius=0.06] (5,.5) circle node[below] {$x$};
    \draw[edge,dotted] (5.5,.5) -- (5,.5);
    \draw[thick, fill, radius=0.06] (5.5,-.5) circle;
    \draw[thick, fill, radius=0.06] (5.5,.5) circle;
    \draw[edge,dotted] (5.5,-.5) -- (5.5,.5);
    \draw[thick, fill, radius=0.06] (6,1.5) circle;
    \draw[edge,dotted] (6,1.5) -- (5.5,.5);
    \draw[thick, fill, radius=0.06] (6,-1.5) circle;
    \draw[edge,dotted] (6,-1.5) -- (5.5,-.5);
    \draw[thick, fill, radius=0.06] (6.5,1) circle;
    \draw[edge] (6.5,1) -- (6,1.5);
    \draw[thick, fill, radius=0.06] (6.5,-1) circle;
    \draw[edge] (6.5,-1) -- (6,-1.5);
    \draw[thick, fill, radius=0.06] (8,1.5) circle;
    \draw[edge] (8,1.5) -- (6,1.5);
    \draw[thick, fill, radius=0.06] (7.5,1) circle;
    \draw[edge] (7.5,1) -- (8,1.5);
    \draw[thick, fill, radius=0.06] (8,-1.5) circle;
    \draw[edge] (8,-1.5) -- (6,-1.5);
    \draw[thick, fill, radius=0.06] (7.5,-1) circle;
    \draw[edge] (8,-1.5) -- (7.5,-1);
    \draw[edge] (9.5,0) -- (8,1.5);
    \draw[edge] (9.5,0) -- (8,-1.5);
    \draw[thick, fill, radius=0.06] (8.25,0) circle;
    \draw[thick, fill, radius=0.06] (9.5,0) circle;
    \draw[edge] (9.5,0) -- (8.25,0);
    \draw[edge] (6.5,1) -- (8.25,0);
    \draw[edge] (8.25,0) -- (6.5,-1);
    \draw[edge] (6.5,-1) -- (7.5,1);
    \draw[edge] (7.5,-1) -- (6.5,1);
    \draw[edge] (7.5,-1) -- (7.5,1);

         \end{scope}
	 \begin{scope}[xshift=6cm,yshift=0cm]
    \draw[thick, fill, radius=0.06] (5,-.5) circle node[below] {$y$};
    \draw[arc] (5.5,-.5) -- (5,-.5);
    \draw[thick, fill, radius=0.06] (5,.5) circle node[below] {$x$};
    \draw[arc] (5.5,.5) -- (5,.5);
    \draw[thick, fill, radius=0.06] (5.5,-.5) circle;
    \draw[thick, fill, radius=0.06] (5.5,.5) circle;
    \draw[arc,dotted] (5.5,-.5) -- (5.5,.5);
    \draw[thick, fill, radius=0.06] (6,1.5) circle;
    \draw[arc] (6,1.5) -- (5.5,.5);
    \draw[thick, fill, radius=0.06] (6,-1.5) circle;
    \draw[arc] (6,-1.5) -- (5.5,-.5);
    \draw[thick, fill, radius=0.06] (6.5,1) circle;
    \draw[arc,dotted] (6.5,1) -- (6,1.5);
    \draw[thick, fill, radius=0.06] (6.5,-1) circle;
    \draw[arc] (6.5,-1) -- (6,-1.5);
    \draw[thick, fill, radius=0.06] (8,1.5) circle;
    \draw[arc] (8,1.5) -- (6,1.5);
    \draw[thick, fill, radius=0.06] (7.5,1) circle;
    \draw[arc] (7.5,1) -- (8,1.5);
    \draw[thick, fill, radius=0.06] (8,-1.5) circle;
    \draw[arc,dotted] (8,-1.5) -- (6,-1.5);
    \draw[thick, fill, radius=0.06] (7.5,-1) circle;
    \draw[arc] (8,-1.5) -- (7.5,-1);
    \draw[fill] (8.81,0.09) rectangle (8.99,-0.09);
    \draw[arc,dotted] (9.5,0) -- (8,1.5);
    \draw[arc] (9.5,0) -- (8,-1.5);
    \draw[thick, fill, radius=0.06] (8.25,0) circle;
    \draw[thick, fill, radius=0.06] (9.5,0) circle;
    \draw[arc,dotted] (8.9,0) -- (8.25,0);
    \draw[arc] (8.9,0) -- (9.5,0);
    \draw[arc] (6.5,1) -- (8.25,0);
    \draw[arc] (8.25,0) -- (6.5,-1);
    \draw[arc] (6.5,-1) -- (7.5,1);
    \draw[arc] (7.5,-1) -- (6.5,1);
    \draw[arc,dotted] (7.5,-1) -- (7.5,1);
    \end{scope}
	\end{tikzpicture}
\caption{\label{fig:nottreebasedbutTBorientable} \remiee{An undirected \blue{binary phylogenetic} network that has a tree-based orientation, but is not tree-based. Left, the undirected \blue{binary phylogenetic} network, \blue{$N$ say}, from \cite{Francis2018}, in which all potential root edges are solid. \red{It follows from the discussion related to Figure~\ref{fig:nottreebased} that $N$ has no base-tree.}
Right, a tree-based orientation of \blue{$N$}, where the root is indicated with a square vertex. The base-tree of the orientation (solid arcs) does not \blue{induce} a base-tree of \blue{$N$}.}}
\end{figure}


In this subsection, we prove that \blue{the decision problem of} deciding whether an undirected \blue{binary phylogenetic} network has a tree-based orientation is NP-complete but can be solved with our FPT algorithms. To prove NP-hardness, we \remiee{use \blue{a reduction from the} NP-complete problem of deciding whether an undirected \blue{binary phylogenetic} network~$N$ is tree-based}, which is defined as follows~\cite{Francis2018}. Such a network~\blue{$N$} is {\em tree-based} if \blue{$N$} can be obtained from an undirected binary phylogenetic tree, \blue{called a {\em base-tree}}, by subdividing the edges of the tree, and adding edges between the subdividing vertices. While it can be decided in polynomial-time whether a directed binary phylogenetic network is tree-based, the \blue{analogous question for} undirected binary phylogenetic networks is NP-complete~\blue{\cite{Francis2018}}. Nevertheless, \blue{the FPT algorithms}, \blue{Algorithm~\ref{alg:binary_blob_orientation_TC_SF} and~\ref{alg:C_orientation}} of Section~\ref{sec:fptalgs}, can be used to decide this latter question. \blue{It is shown in}~\cite{Francis2018} that an undirected binary phylogenetic network $N$ is tree-based if and only if $N$ can be tree-based rooted at each of its pendant edges. \blue{Thus, as the class of directed binary tree-based networks is $\ell$-chain reducible, leaf-addable, and blob-determined}, we can run Algorithm~\ref{alg:binary_blob_orientation_TC_SF} or~\ref{alg:C_orientation} and simply check whether \blue{the set of} returned edges where the network can be tree-based rooted at \blue{includes each of the pendant edges}. Note that, if $N$ has a tree-based orientation, then $N$ is not necessarily tree-based. See Figure~\ref{fig:nottreebasedbutTBorientable} for such an example.

\begin{theorem}\label{thm:npcomplete}
The \blue{decision} problem {\sc Tree-Based Orientation} is NP-complete.
\end{theorem}

\begin{proof}
Given an orientation of an \blue{undirected binary phylogenetic} network, we can check in \revL{polynomial} time whether it is tree-based \revL{\cite{FrancisSteel2015}}. Thus the problem is in the class NP. For NP-hardness, we reduce from the problem of deciding whether an undirected \blue{binary phylogenetic} network $N$ is tree-based, which is NP-complete~\cite{Francis2018}. The reduction works as follows. Pick an arbitrary leaf $\ell$ of $N$, adjoin a copy of $K_4$ to $N$ by subdividing an edge of $K_4$ with a vertex $v$, and identifying $v$ and $\ell$. Let $N'$ denote the resulting undirected \blue{binary phylogenetic} network, which is our instance for {\sc Tree-Based Orientation}. This construction can be done in constant time given a leaf.

If $N$ is tree-based, then $N$ \leo{can be tree-based rooted} at $\ell$ \remiee{\cite[Theorem~4]{Francis2018}}. Observe that $K_4$ with an added leaf can be tree-based rooted at an edge incident to the pendant edge, \latest{see Figure~\ref{fig:K4}}.
Since the class of binary tree-based networks is blob-determined, we can combine these \blue{two} tree-based orientations to \blue{give} a tree-based orientation of $N'$.

\begin{figure}[h]
    \centering
    \begin{tikzpicture}
	 \tikzset{edge/.style={thick}}
     \tikzset{arc/.style={-Latex,thick}}
	 \begin{scope}[xshift=0cm,yshift=0cm]
    \draw[thick, fill, radius=0.06] (0,0) circle;
    \draw[thick, fill, radius=0.06] (.5,-.5) circle node[above right] {$\rho$};
    \draw[thick, fill, radius=0.06] (0,-1) circle;
    \draw[thick, fill, radius=0.06] (0,-2) circle;
    \draw[thick, fill, radius=0.06] (-1,-1) circle;
    \draw[thick, fill, radius=0.06] (1,-1) circle;
	\draw[arc] (0,0) -- (-1,-1);
    \draw[arc] (.5,-.5) -- (0,0);
    \draw[arc,dotted] (.5,-.5) -- (1,-1);
    \draw[arc] (-1,-1) -- (0,-1);
    \draw[arc,dotted] (-1,-1) -- (0,-2);
    \draw[arc] (0,-2) -- (1,-1);
    \draw[arc] (0,-1) -- (0,-2);
    \draw[arc,dotted] (0,0) -- (0,-1);
    \draw[arc] (1,-1) -- (2,-1);
    \draw[fill=lightgray] (3,-1) ellipse (1 and 1);
	\draw (3,-1) node {$N$};
	\draw (1.5,-2.5) node {$N'$};
    \end{scope}
	\end{tikzpicture}
   \caption{\latest{Illustration of the proof of Theorem~\ref{thm:npcomplete}, indicating a tree-based orientation of the part of~$N'$ that is not part of~$N$. The solid arcs form a base-tree. If~$N$ is tree-based, this orientation can be extended to a tree-based orientation of~$N'$.}
    \label{fig:K4}}
\end{figure}

Now suppose $N'$ has a tree-based orientation, and say it can be tree-based rooted at $e_\rho$. Then $e_\rho$ must be an edge of the subdivided $K_4$ \latest{(because otherwise the edge incident to the subdivided~$K_4$ would be directed towards the subdivided~$K_4$ which is impossible as the subdivided~$K_4$ does not contain any leaves)}. Hence $N$ can be tree-based rooted at $\ell$.
\end{proof}

\end{document}